\documentclass[prx,aps,twocolumn,notitlepage,superscriptaddress,showpacs,nofootinbib]{revtex4-2}

\usepackage{enumerate,appendix}
\usepackage{amsmath, amsthm, amssymb,commath}
\usepackage{color,calc,graphicx}
\usepackage[usenames,dvipsnames,svgnames,table,cmyk,hyperref]{xcolor}
\usepackage[colorlinks]{hyperref}
\usepackage{optidef}
\hypersetup{
	colorlinks = true,
	urlcolor = {blue},
	citecolor = {blue},
	linkcolor= {blue}
}

\usepackage{graphicx}
\usepackage{amsmath}
\usepackage{latexsym}
\usepackage{bbm}

\usepackage[charter,cal=cmcal,sfscaled=false]{mathdesign}
\usepackage{booktabs}
\usepackage{multirow} 
\usepackage{subcaption}
\usepackage{dcolumn}
\usepackage{mathrsfs}

\def \be {\begin{equation}}
\def \ee {\end{equation}}

\newcommand{\tr}{\mathrm{Tr}}
\newcommand{\Tr}{\mathrm{Tr}}

\newcommand{\ket}[1]{|#1\rangle}
\newcommand{\bra}[1]{\langle#1|}

\def \Im{\mathrm{Im}\,}

\def \del{\partial}

\def \argmin{\mathop{\rm argmin}}

\def \cH{{\cal H}}
\def \cX{{\cal X}}

\def \sofc2{{\cal S}({\mathbb C}^2)}

\def\>{\rangle}
\def\<{\langle}

\newtheorem{theorem}{Theorem}
\newtheorem{lemma}[theorem]{Lemma}

\def\Label#1{\label{#1}\ [\ \text{#1}\ ]\ }
\def\Label{\label}

\begin{document}

\title{Finding the optimal probe state for multiparameter quantum metrology using conic programming}

\author{Masahito Hayashi}
\email{hmasahito@cuhk.edu.cn}
\affiliation{School of Data Science, The Chinese University of Hong Kong,
Shenzhen, Longgang District, Shenzhen, 518172, China}
\affiliation{International Quantum Academy (SIQA), Futian District, Shenzhen 518048, China}
\affiliation{Graduate School of Mathematics, Nagoya University, Nagoya, 464-8602, Japan}
\author{Yingkai Ouyang}
\email{y.ouyang@sheffield.ac.uk}
\affiliation{Department of Physics \& Astronomy, University of Sheffield, Sheffield, S3 7RH, United Kingdom}

\begin{abstract}
The aim of the channel estimation is
to estimate the parameters encoded in a quantum channel. For this aim, it is allowed to choose 
the input state as well as the measurement to get the outcome. 
Various precision bounds are known for the state estimation.
For the channel estimation,
the respective bounds are determined 
depending on the choice of the input state.
However, determining the optimal input probe state and the corresponding precision bounds in estimation is a non-trivial problem, particularly in the multi-parameter setting, where parameters are often incompatible. 
In this paper, we present a conic programming framework that allows us to determine the optimal probe state for
the corresponding multi-parameter precision bounds. 
The precision bounds we consider include
the Holevo-Nagaoka bound and the tight precision bound that give the optimal performances of 
correlated and uncorrelated measurement
strategies, respectively.
Using our conic programming framework, we discuss the optimality of a maximally entangled probe state in various settings.
We also apply our theory to analyze the canonical field sensing problem using entangled quantum
probe states.
\end{abstract}

\maketitle

\section{Introduction}

Channel estimation \cite{escher2011general,Hayashi11,pirandola2019fundamental,zhou2021asymptotic} utilizes quantum resources to allow us to estimate parametrized quantum processes with unprecedented precision. In channel estimation, a quantum channel that embeds physical parameters of interest acts on an initial probe state. 
For instance, consider estimating the physical parameters associated with a classical field that interacts linearly with an ensemble of qubits. 
In such a scenario, the quantum channel that describes the dynamics of the quantum system embeds the unknown field parameters. In addition, we have the flexibility to select the initial state, or probe state, of the quantum system. 
With access to multiple queries of the quantum channel, the objective is to estimate the embedded physical parameters with maximum precision.
An fundamental challenge in channel estimation is that of selecting the probe state that gives the optimal precision in estimation.

In quantum state estimation \cite{HELSTROM1967101,helstrom,holevo,nagaoka89,HM08}, we have a family of quantum states that depend continuously on physical parameters. We may interpret the channel estimation problem as a quantum state estimation problem, where the quantum states in the latter problem are the output of a quantum channel that acts on the probe state.
In general, precision bounds for the quantum state estimation depend on the type of measurement strategy that we choose. 
Two such families are correlated and uncorrelated measurement strategies.
Correlated measurement strategies operate jointly on infinite copies of quantum systems, whereas uncorrelated measurement strategies act on one copy at a time.
In the single parameter setting, the upper bound to the ultimate precision is the celebrated quantum Cramer Rao bound, which is tight for both correlated and uncorrelated measurement strategies \cite{HELSTROM1967101,helstrom,holevo,nagaoka89,hayashi,HM}.
For our channel estimation problem, we set the probe state to be entangled over systems A and C, and the quantum channel maps system A to B for both uncorrelated and correlated measurement strategies (see Figures \ref{fig1:operational-meaning} and \ref{fig2:operational-meaning}).

\begin{figure}[htbp]  
    \centering 
        \includegraphics[width=0.5\textwidth]{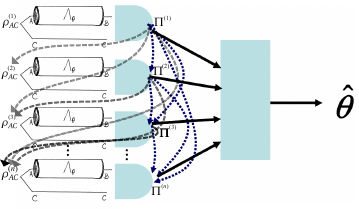}  
        \caption{{\bf Channel estimation by uncorrelated measurement and entangled probe states.}
In the $j$-th experiment for $j=1,\dots,n$, 
we input the probe state $\rho_{AC}^{(j)}$ over systems $A$ and $C$ of equal size, and measure the output entangled state
$(\Lambda_\theta \otimes \iota_C)(\rho_{AC}^{(j)})$
via the measurement $\Pi^{(j)}$.
Then, we obtain $j$-th measurement outcome $\omega^{(j)}$.
The probe state $\rho_{AC}^{(j)}$ and the measurement $\Pi^{(j)}$
are adaptively chosen by the previous  
measurement outcomes $\omega^{(1)}, \ldots, \omega^{(j-1)}$.
Finally, we decide our estimate $\hat{\theta}$
from the measurement outcomes $\omega^{(1)}, \ldots, \omega^{(n)}$.
The precision bound $J_1$ is the ultimate precision bound 
under the strategy of uncorrelated measurements.        } 
    \label{fig1:operational-meaning}
\end{figure}

\begin{figure}[htbp]  
    \centering 
        \includegraphics[width=0.46\textwidth]{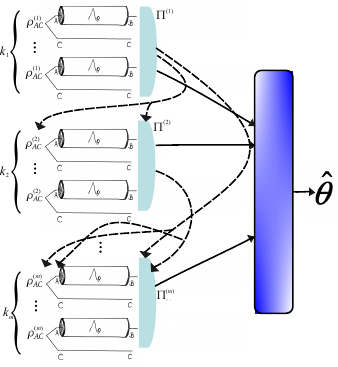}  
        \caption{
        {\bf Channel estimation by correlated measurement and entangled probe states.}
We have $m$ steps, and the $j$-th step has $k_j$ experiments.
In the $j$-th step, 
we input the probe state $\rho_{AC}^{(j)}$ over systems $A$ and $C$ 
for $k_j$ experiments.
We measure the output entangled state
$(\Lambda_\theta \otimes \iota_C)(\rho_{AC}^{(j)})^{\otimes k_j}$
via the measurement $\Pi^{(j)}$
and obtain $j$-th measurement outcome $\omega^{(j)}$.
The probe state $\rho_{AC}^{(j)}$ and the measurement $\Pi^{(j)}$
are adaptively chosen by the previous  
measurement outcomes $\omega^{(1)}, \ldots, \omega^{(j-1)}$.
Finally, we decide our estimate $\hat{\theta}$
from the measurement outcomes $\omega^{(1)}, \ldots, \omega^{(m)}$.
In this scenario, we use the unknown channel for a total of $k_1+ \cdots +k_m$ times.
The precision bound $J_5$ is the ultimate precision bound 
under the strategy of correlated measurements.        } 
    \label{fig2:operational-meaning}
\end{figure}

In the multiparameter setting \cite{D_and_D_2020,sidhu2020geometric}, the theory for quantum state estimation is far richer than the single parameter case; there are several distinct Cramer-Rao type bounds on the optimal precision for both correlated and uncorrelated measurement strategies \cite{HO}.
These bounds include the Holevo-Nagaoka (HN) bound \cite{holevo,nagaoka89,
HM,Albarelli2019_PRL,sidhu2020tight} \footnote{In the literature it is often also called the Holevo Cramer-Rao bound.} for correlated measurement strategies and the symmetric logarithmic derivative (SLD) bound \footnote{This is a direct generalization of the quantum Cramer Rao bound that is based on the SLD.}.
For uncorrelated measurement strategies, the HN and SLD bounds are not tight. 
The ultimate precision bound, namely the tight bound, for uncorrelated measurement strategies was initially formulated in \cite{Haya} as an optimization problem with uncountably infinite constraints, and was later reformulated as a conic optimization of an operator with dimension proportional to the size of the probe state and over a certain separable cone in \cite{HO}. 
The Nagaoka-Hayashi (NH) bound \cite{Nagaoka-generalization,Hayanoncomm,CSLA} gives a simpler albeit weaker precision bound for uncorrelated measurement strategies.
Recently, it was shown that the HN, SLD, tight, and NH bound can be obtained from the optimal value of conic programs with the same objective function, but optimized over different cones \cite{HO}.
In Section \ref{sec:estimation}, we review multiparameter quantum state estimation, particularly with respect to the conic programming framework of Ref.~\cite{HO}.

From an operational point of view, the tight bound is the most significant for channel estimation, because it describes the ultimate precision bound using uncorrelated measurement strategies. In uncorrelated measurement strategies, after the preparation of each probe state and the action of the quantum channel, we measure the evolved probe state as depicted in Figure \ref{fig1:operational-meaning}. We repeat the procedure for multiple copies of probe states, where the explicit form of each probe state and the corresponding measurement strategy on it is informed by all prior measurement outcomes.

We may also be interested in correlated measurement strategies for channel estimation, which are more challenging than uncorrelated measurement strategies to implement in practice. 
In this case, the relevant precision bound is the HN bound.
For the HN bound, we are allowed to perform correlated measurements across multiple copies of probe states (see Figure \ref{fig2:operational-meaning}).
After each batch of correlated measurements, we can update our choice of probe states and measurement strategies.

Determining the optimal precision for the channel estimation problem entails a two-step optimization process. 
First, we fix the probe state, and find the precision according to the appropriate Cramer-Rao type bound. Namely, we consider a set of parametrized states when the multi-parameter quantum channel acts on the probe state, and evaluate various Cramer-Rao type bounds on the set of parametrized states by solving an appropriate optimization problem.
Second, we optimize various Cramer-Rao type bounds by changing the probe state.

The two-step optimization is a non-trivial problem.
Even if the first optimization for a fixed probe state is efficient, the optimized precision bound is not necessarily a convex function of the probe state, and hence the second optimization over the probe state need not be so easy to perform. 
Remarkably, this challenging two-step optimization has been addressed in the single parameter setting \cite{Liu,Altherr}.
In particular, the references
\cite{Liu,Altherr} derived a semidefinite programming (SDP) form 
for the maximum SLD Fisher information
in the channel estimation that embeds one parameter.
However, the question of how to find both the optimal precision and optimal probe state for multi-parameter quantum channels remains unresolved.

The question of 
finding optimal probe states for estimating multiple parameters embedded in a quantum channel has recently been considered \cite{gorecki2020quantum,friel2020attainability}.
For probe states without ancilla assistance, Ref.~\cite{friel2020attainability} numerically finds the optimal two-qubit probe state with the HN bound when the channel models 3D-field sensing with independent and identical amplitude (i.i.d.) damping on the two qubits.
Ref.~\cite{gorecki2020quantum} considers the HN bound with unitary quantum channels. 
This leaves open the questions of how to evaluate the optimal probe state and corresponding bounds for the tight, NH and SLD bound for the general problem of quantum channel estimation.

In the channel estimation problem where we estimate the parameters embedded in the quantum channel, the set of parameters is continuous. 
If we discretize the channel estimation problem, we would obtain the problem of discriminating a discrete set of quantum channels \cite{acin-PhysRevLett.87.177901,sacchi-PhysRevA.71.062340,hay2009,pirandola2019fundamental,zhuang-PhysRevLett.125.080505,wilde2020amortized}.
Recently, the channel discrimination problem was formulated as a convex program, and this formulation made it possible to determine the optimal strategy to discriminate a pair of quantum channels \cite{Nakahira-and-Kato-2021}.
However, it is not so clear how to extend this result to the continuous parameter setting that we require in the channel estimation problem.

In this paper,
we give a framework that uses conic programs to find the optimal probe state for channel estimation with respect to multi-parameter bounds. The multi-parameter bounds include the SLD bound, the HN bound, the tight bound, and the NH bound.
For the SLD bound, the HN bound, and the NH bound, we show that the conic programs are also in fact semidefinite programs (SDP) and hence can be efficiently solved.
For the tight bound, we will require the same technique for optimizing over the separable cone as considered previously \cite{HO}.

We furthermore apply our framework to unravel situations where the maximally entangled state is optimal, 
and also study both numerically and theoretically the canonical problem of field sensing in the presence of collective amplitude damping, and in the multiparameter setting.

With our resolution of the non-trivial problem of simultaneously optimizing the  multiparameter Cramer-Rao bound and the corresponding probe state, we expect our theoretical contributions to have ramifications beyond the field-sensing application that we have explored. 
Namely, one will be able to apply our framework to determine optimal precision bounds and optimal probe states for a plethora of quantum sensing problems, where the unknown parameter is embedded in the underlying quantum channel, for instance, in many quantum imaging problems \cite{moreau2019imaging}.

\section{Overview of our main results}

We establish the following main results.\\

\noindent $\star$ {\bf A conic programming framework for channel estimation:- }The conic programming framework that we establish in Section \ref{sec:channel-estimation} proves that the optimal values of conic programs can correspond to Cramer-Rao type bounds for channel estimation.
Our framework applies in the multiparameter setting.
This main result is given Theorem \ref{NZD}, where we prove that the optimal values of the following are equivalent.
\begin{enumerate}
\item Fix the probe state, solve a minimization using the conic programs in \cite{HO} to find the best precision bound. Next solve a maximization for the best probe state to use. 
\item Solve a minimization for the conic programs that we introduce.
\end{enumerate}
The optimal values of the above two problems give the precision bounds for channel estimation.
Problem 1 is often not a convex program, and hence it is not clear how to solve it numerically. In contrast, in problem 2, the conic programs are convex programs, and we can solve them efficiently.
Hence, our framework allows one to solve the channel estimation problem efficiently.

Moreover, the conic programs' optimal solutions also allow us to evaluate the corresponding optimal probe states. Namely, we define a matrix-valued function $\rho_A(Y)$ that maps the optimal solution $Y^*$ of the conic program to a density matrix on system A (see Figures \ref{fig1:operational-meaning} and \ref{fig2:operational-meaning}), and the purification of such a density matrix to the joint system AC gives the optimal probe state.
Moreover, we discuss how to use the optimal solution of our conic program to calculate the optimal probe state. Hence, it becomes possible to determine the optimal probe state for channel estimation for the HN bound, the SLD bound, the NH bound, and the tight bound.
Of operational significance are our results on the channel versions of the HN bound and the tight bound, which give the optimal performances of correlated and uncorrelated measurement strategies, respectively. \newline

\noindent $\star$ {\bf On the optimality of the maximally entangled probe state:- }
There are many situations where the maximally entangled probe state is the optimal probe state to use in channel estimation. In Section \ref{sec:maxent}, we explore this possibility, and begin by considering the scenario where the input probe state is a maximally entangled state. 
We establish Theorem \ref{NNT} and Theorem \ref{NNT2}, which shows that a constraint on a dual matrix-valued variable is equivalent to the optimality of the maximally entangled state for the SLD bound and the HN bound respectively.
These theoretical results allow us to find situations when the maximally entangled state is the optimal probe state for both the SLD bound and the HN bound. 
We consider the following examples.
\newline

\noindent $\triangleright$ {\bf Noisy qubit channel with one parameter:- }
We consider the quantum channel as a unitary evolution afflicted with depolarizing noise, where the parameters are embedded into the unitary part of the channel. More precisely, we first consider a quantum channel that describes the mixture of a unitary qubit evolution and replacement by a completely mixed state. The unitary qubit evolution is generated by the product of the single parameter that we want to estimate and a Pauli matrix. 
We use Theorem \ref{NNT} to establish Theorem \ref{thm:1param}, where we prove that the maximally entangled state is the optimal probe state, and furthermore, the tight bound, the NH bound, the SLD bound and the HN bound are all equivalent.
\newline

\noindent $\triangleright$ {\bf Generalized Pauli channels:- } 
On a qudit system, a generalized Pauli channel applies generalized Paulis randomly on the qudit state according an apriori determined probability distribution. The channel estimation task here is to estimate this probability distribution. We use Theorem \ref{NNT}
to establish Theorem \ref{thm:Pauli}, which 
shows that the optimal probe state for the channel estimation problem for generalized Pauli channels under the SLD bound is the maximally entangled state.
\newline

\noindent $\triangleright$ {\bf Noisy SU(2) channels:- }
We consider a quantum channel that is the mixture of unitary evolution according to the spin-$j$ representation of SU(2) unitary evolution according to the spin-$j$ representation of SU(2) and replacement by a completely mixed state. 
We also give equal weightage to each of the three parameters that we estimate. 
We use Theorem \ref{NNT} to establish Theorem \ref{thm:SU2-maxent}, where we derive the analytical form for SLD bound. 
Furthermore, Theorem \ref{thm:SU2-maxent} shows that under purely SU(2) unitary evolution with zero noise, all the precision bounds $J_1,\dots, J_5$ are equivalent, and the maximally entangled state is the optimal probe state for all the precision bounds.
\newline

\noindent $\star$ {\bf Noisy field sensing:- }
In Section \ref{sec:field sensing}, we revisit the problem of quantum field sensing \cite{toth2014quantum} in the multiparameter setting.
In the noiseless setting, this is equivalent to the channel estimation problem for noisy SU(2) channels that we explore in Section \ref{sec:field sensing}.
The noise model we consider here is different from that in Section \ref{sec:field sensing}, instead of a depolarizing type of channel, we consider noise introduced by collective amplitude damping \cite[Eq. (7)]{PhysRevA.58.3491}. 
The channel that we consider differs from \cite{friel2020attainability} in two ways.
First, we consider collective amplitude damping while \cite{friel2020attainability} considers i.i.d. amplitude damping.
Second, we model the channel using a master equation, considering collective amplitude damping that occurs during the SU(2) evolution, whereas the channel in \cite{friel2020attainability} considers i.i.d. amplitude damping that occur after the unitary evolution.
We give corresponding plots of various precision bounds against the noise parameter in Figure \ref{fig:fixn} and Figure \ref{fig:APJgam}.

Using the MatLab computer code, we numerically determine the optimal probe state for the NH bound, the HN bound, and the SLD bound, and numerically evaluate the corresponding precision bounds. 
We numerically ascertain that in the noiseless setting, the maximally entangled state on the symmetric subspace is the optimal probe state, which agrees with our result in Theorem \ref{thm:SU2-maxent}.

When there is non-vanishing noise, we numerically ascertain that the SLD precision bound cannot be optimal. We furthermore prove this fact in Theorem \ref{THNM}, where in the proof, we calculate the expectations of the commutators of the symmetric logarithmic derivatives of different angular momentum generators.
\newline

\noindent $\star$ {\bf Semi-definite programming formulations:- }
To maximize the accessibility of our work, in Section \ref{sec:SDPs},
we give formulate the mathematical optimization problems that correspond to $J_2$, $J_3$, $J_4$, $J_5$ as semidefinite programs to be used with the CVX package and provide the corresponding MatLab code.

\section{Quantum state estimation}
\label{sec:estimation}

In the quantum parameter estimation problem, given a set of parameters $\Theta \subseteq \mathbb R^d$, we consider $d$-parameter state family $\{\rho_\theta : \theta = (\theta^1, \dots , \theta^d) \in \Theta \}$
on the Hilbert space ${\cal H}$.
Denoting the set of density matrices on ${\cal H}$ by ${\cal S}({\cal H})$,
we note that the quantum states $\rho_\theta$ are elements of ${\cal S}({\cal H})$, which means that they are positive semidefinite operators with unit trace. 
We define the true parameter to be $\theta_0$, and our goal is to construct an estimator $\hat \theta$ that will estimate $\theta_0$.
We furthermore consider $\Theta$ as a continuous set, and where the quantum states $\rho_\theta$ are differentiable with respect to $\theta$, so that we can define the partial derivatives $\rho_\theta$. In particular, in the neighborhood of the true parameter $\theta_0$, we define $D_j:= \frac{\partial}{\partial \theta^j}\rho_\theta|_{\theta_0}$,
and $\rho:= \rho_{\theta_0}$.

We can describe a measurement using a set of positive operators $\Pi = \{ \Pi_x : x\in \mathcal X\}$ where the completeness condition $\sum_{x \in \mathcal X}\Pi_x = I $ holds, where $\mathcal X$ is a set of classical labels. 
By the Born rule, 
when we perform a measurement $\Pi$ on a quantum state $\rho_\theta$, we will obtain the classical label $x$ and the state 
$\Pi_x \rho_\theta / \tr(\Pi_x \rho_\theta )$ 
with probability $p_\theta(x) = \tr(\Pi_x \rho_{\theta})$.
Upon measurement, the classical label $x$ is a random variable, and we denote $\mathbb E_{\bm{\theta}}[f(x) |\Pi]$ as the expectation of $f(x)$, the function $f$ of the classical label $x$, with probability distribution obtained according to the Born rule.
Now given a measurement $\Pi$ and an estimator $\hat{\bm{\theta}}$, where
the estimator $\hat{\bm{\theta}}$ is to be a function of the classical label $x$, we denote  
$\hat{\Pi}=(\Pi , \hat{{\bm{\theta}}})$ as an {\it estimator}. 
When the true parameter $\theta_0$ is equal to $\theta$, we define the mean-square error (MSE) matrix for the estimator $\hat{\Pi}$ as  
\begin{align}\nonumber
V_{{\bm{\theta}}}[\hat{\Pi}]
&=
\sum_{i,j=1}^d |i\>\<j|
 \sum_{x\in\cX} \tr{\rho_{\bm{\theta}}\Pi_x}({\hat{\theta}^i}(x)-\theta^i)({\hat{\theta}^j}(x)-\theta^j)  \\
&=
\sum_{i,j=1}^d |i\>\<j|
 \mathbb E_{\bm{\theta}}\big[({\hat{\theta}^i}(x)-\theta^i)({\hat{\theta}^j}(x)-\theta^j)|\Pi\big] \notag.
\end{align}
In multiparameter quantum metrology, the objective is to 
find an optimal estimator $\hat{\Pi}=(\Pi , \hat{{\bm{\theta}}})$ that in some sense minimizes the MSE matrix.
More precisely, the quantity to minimize is $\tr G V_{{\bm{\theta}}}[\hat{\Pi}]$, where a weight matrix $G$, a size $d$ positive semidefinite matrix, quantifies the relative importance of the different parameters in our parameter estimation problem. If we give all parameters an equal importance, we may choose $G$ to be the identity matrix.

When for all $i = 1,\dots, d$, the expectation of our estimator is equal to the true value of the parameter $\theta_0$ when $\theta_0 = \theta$, we have the condition 
\begin{align}
\mathbb E_{\bm{\theta}}\big[{\hat{\theta}^i}(x)|\Pi\big]&
=\sum_{x\in\cX} {\hat{\theta}^i}(x) \tr{\big[\rho_{{\bm{\theta}}}\Pi_x\big]}
=\theta^i \label{MK},
\end{align}
which means that our estimator $\hat{\Pi}=(\Pi , \hat{{\bm{\theta}}})$ is unbiased at $\theta_0 = \theta$. 
When \eqref{MK} holds for all $\theta \in \Theta$, we say that our estimator is globally unbiased. 
Unfortunately, globally unbiased estimators need not exist, and hence we consider estimators that are unbiased in the neighborhood of the true parameter $\theta_0$. 
This motivates us to take partial derivatives on both sides of \eqref{MK}, and consider  
\begin{align}
\frac{\del}{\del\theta^j}\mathbb E_{\bm{\theta}}\big[{\hat{\theta}^i}(x)|\Pi\big]&=
\sum_{x\in\cX} {\hat{\theta}^i}(x)\tr{   D_j \Pi_x}
=\delta_i^j \label{M1},
\end{align}
where $D_j = \frac{\del}{\del\theta^j}\rho_{{\bm{\theta}}} $.
When \eqref{MK} holds for all $i=1,\dots, d$ for a fixed $\theta$ where $\theta_0 = \theta$, and when \eqref{M1} holds for all $i,j = 1,\dots, d$, then we say that the estimator $\hat{\Pi}=(\Pi , \hat{{\bm{\theta}}})$ satisfies the locally unbiased condition.

 For any weight matrix $G = \sum_{i,j=1}^d g_{i,j}|i\>\<j|$, the fundamental precision limit \cite{HO} is given by
\be\label{qcrbound}
C_{\bm{\theta}}[G]:=
\min_{\hat{\Pi}\mathrm{\,:l.u.at\,}{\bm{\theta}}}\Tr{ \big[G V_{\bm{\theta}}[\hat{\Pi}]\big]}, \notag
\ee
where the minimization is carried out for all possible estimators under the locally unbiasedness condition, which is indicated by l.u.~at ${\bm{\theta}}$. 
Note that when we impose only the condition \eqref{M1}, the above minimum is attained by $\hat{\Pi}$ satisfying \eqref{MK}. Therefore, $C_{\bm{\theta}}[G]$ can be considered as the minimum only with the condition \eqref{M1}.
In the following, we consider the above minimization.
Hence, it is sufficient to focus on the operators $\rho_\theta$ and
$(D_j)_j$. 
Thus, the pair $(\rho_\theta, (D_j)_j)$ is called a model
in the following.

As discussed in Ref.~\cite{HO}, any lower bound to weighted trace of the MSE matrix is a Cramer-Rao (CR) type bound. The fundamental precision limit $C_{\bm{\theta}}[G]$ is one such lower bound which is tight, and hence refered to as the {\it tight CR} bound \cite{HO}.  
Operationally, we may attain the tight CR bound using an uncorrelated measurement strategy in the asymptotic setting.
This means that we would attain this tight CR bound by performing optimal measurements independently on asymptotically many individual copies of the probe states.
In quantum parameter estimation, one may also consider correlated measurement strategies, where we would allow joint measurements over multiple, and potentially infinite number of copies of probe states.
The Holevo-Nagaoka (HN) bound \cite{holevo,nagaoka89,HM,Albarelli2019_PRL,sidhu2020tight} is a CR bound that describes the ultimate precision in this correlated measurement strategy scenario, and the HN bound can be strictly smaller than the tight CR bound \cite{HO}.

While it is difficult to evaluate the tight CR bound exactly, one can nonetheless use a semidefinite program (SDP) to approximate it \cite{HO}. As the precision of the approximation increases, the complexity of the SDP also increases. 
Other CR-type bounds are more efficient to evaluate, such as the Nagaoka-Hayashi (NH) bound, the HN bound, and the SLD bound.
Recently, Ref.~\cite{HO} used the language of conic programming to clarify the relationship between these CR-type bounds. 
Namely, let us consider an operator $X$ that we construct from any estimator $\hat \Pi = (\Pi,\hat{\theta})$, where 
\begin{align}
&X(\Pi,\hat{\theta})\nonumber \\
:=& \sum_{
{x\in\mathcal X}
}
\Big(|0\rangle+ \sum_{i=1}^d \hat{\theta}^i(
{x}
) |i\rangle\Big) 
\Big(\langle 0|+ \sum_{i=1}^d \langle i| \hat{\theta}^i(
{x}
)\Big)
\otimes 
\Pi_{x}. \notag
\end{align}
This operator $X(\Pi,\hat{\theta})$ acts on the vectors in $\mathcal R \otimes \mathcal H$, 
where $\mathcal R = \mathbb R^{d+1}$ is spanned by $d+1$ basis vectors $|0\>,|1\>,\dots, |d\>$.

We may write the trace of the weighted MSE matrix $\Tr{G V_{\bm{\theta}}[\hat{\Pi}]}$ using $X(\Pi,\hat{\theta})$. Namely, 
\begin{align}
\Tr{G V_{\bm{\theta}}[\hat{\Pi}]} = \Tr ( G \otimes \rho) X(\Pi,\hat{\theta}) \Label{o1}.
\end{align}
Next, note that the completeness condition ${\sum_{x \in \mathcal X} \Pi_x = I_{\mathcal H}}$ using $X(\Pi,\hat{\theta})$ implies that 
\begin{align}
\Tr_{{\cal R}} (|0\rangle \langle 0| \otimes I_{{\cal H}})  
X(\Pi,\hat{\theta})&=I_{{\cal H}} \Label{c1}.
\end{align}
Hence, we may interpret \eqref{c1} as a rewriting of the completeness condition $\sum_{x \in \mathcal X} \Pi_x = I_{\mathcal H}$.
Next, we note that the condition \eqref{M1} for a locally unbiased estimator guarantees 
\begin{align}
\Tr (
\frac{1}{2}(|0\rangle \langle i| +|i\rangle \langle 0 |)
\otimes D_j )  X(\Pi,\hat{\theta}) &=\delta_{i,j}  \Label{c2} . 
\end{align}
Hence, we may interpret \eqref{c2} as a rewriting of the locally unbiased condition.

Note that the operator $X(\Pi,\hat{\theta})$ has a tensor product structure. 
Namely, we may consider $X(\Pi,\hat{\theta})$ as an element of a cone generated by separable states on $\mathcal R \otimes \mathcal H$.
More precisely, this separable cone $\mathcal S^1$ 
is the convex hull of 
the set of operators that are 
a tensor product of a real positive semidefinite matrix on $\mathcal R$ and 
a complex positive semidefinite operators on $\mathcal H$ with bounded norm.
Hence, it is natural to consider the minimization:
\begin{align}
S_1:= \min_{X \in \mathcal S^1}
\{\Tr ( G \otimes \rho) X |
\eqref{c1},\eqref{c2} \hbox{ hold.}
\}
\notag
\end{align} 
The reference \cite{HO} showed that $S_1$ is in fact equal to the tight CR type bound.
The reference \cite{HO} also showed that if we consider the minimization of 
$\Tr ( G \otimes \rho) X $ subject to the conditions \eqref{c1}, \eqref{c2}, but over suitable cones that contain $\mathcal S^1$, the optimal value can be made to be equal to the NH bound, the HN bound and the SLD bound.
Given this phenomenon, it is instructive to revisit the cones over which we can optimize $X$.

We now proceed to define several cones over which we like to optimize $X$.
Now let us define $\mathcal B$ as the vector space spanned by the tensor product of real symmetric matrices on $\mathcal R$ and bounded complex Hermitian matrices on $\mathcal H$. 
This means that we can write
\begin{align}
{\cal B}(\mathcal{R},{\cal H})&:=\Big\{\sum_{j=0}^d \sum_{k=0}^d
|k\rangle \langle j| \otimes X_{k,j} \Big|
X_{k,j} \in {\cal B}_{\rm sa}({\cal H}), X_{k,j}=X_{j,k} \Big\},\notag
\end{align}
where ${\cal B}_{\rm sa}({\cal H})$ denotes the set of self-adjoint (Hermitian) matrices on $\mathcal H$ with bounded norm.
When there is no confusion, 
${\cal B}(\mathbb{R}^{d+1},{\cal H})$ is simplified to ${\cal B}$.
Let us consider the cone $\mathcal S^2$ as 
\begin{align}
\mathcal S^2&:=\{X \in {\cal B}| \langle v| X|v\rangle \ge 0 \hbox{ for all } |v\> \in  \mathbb C^{d+1} \otimes {\cal H}\} .\notag
\end{align}
Relaxing the condition of ${\cal B}$
we extend the space ${\cal B}$ as
\begin{align}
{\cal B}''&:=\Big\{\sum_{j=0}^d \sum_{k=0}^d
|k\rangle \langle j| \otimes X_{k,j} \Big|
X_{k,0} \in {\cal B}_{\rm sa}({\cal H}), X_{k,j}=(X_{j,k})^\dagger \Big\} .\notag
\end{align}

\begin{table*}[htbp]
    \centering
    \rowcolors{2}{white}{gray!15}
    \label{tab:symbols}
    \begin{tabular}{>{$}c<{$} l}
    \toprule
    \textbf{Symbol} & \textbf{Meaning} \\
    \midrule
    d & the number of parameters to be estimated \\
    G & weight matrix, real positive semidefinite $d \times d$ matrix\\
    \theta_0 & the parameters' true value, a vector in $\mathbb R^d$  \\
    \hat \theta & estimator of the $d$ parameters, a vector in $\mathbb R^d$  \\
     \Theta \subseteq \mathbb R^d & set of all possible parameter vectors \\
    \Lambda_{\theta}  & quantum channel (mapping system A to B) that embeds $\theta$\\
    T_{\theta} & Choi matrix $\Lambda_{\theta}$ \\
    T_{\theta_0} & Choi matrix $\Lambda_{\theta_0}$ at the parameter's true value $\theta_0$ \\
    \rho_{AC}  & a probe state parametrized by $\theta$ \\
F_j & Choi matrix's partial derivative $\frac{\partial}{\partial \theta^j} T_\theta|_{\theta=\theta_0}$. \\
S_1 &CR-type bound: tight bound, optimized over cone $\mathcal S^1 = \mathcal S^1_{BC}$\\
S_2 &CR-type bound: Nagaoka-Hayashi bound, optimized over cone $\mathcal S^2 = \mathcal S^2_{BC}$\\
S_3 &CR-type bound: bound from optimization over PPT cone, optimized over cone $\mathcal S^3 = \mathcal S^3_{BC}$\\
S_4 &CR-type bound: SLD bound, optimized over cone $\mathcal S^4 = \mathcal S^4_{BC}$\\
S_5 &CR-type bound: Holevo-Nagaoka bound, optimized over cone $\mathcal S^5( (\Lambda_\theta \otimes \iota_C)( \rho_{AC})) = \mathcal S^5_{BC}( (\Lambda_\theta \otimes \iota_C)( \rho_{AC}))$\\
    \mathcal S^1, \mathcal S^2, \mathcal S^3, \mathcal S^4, \mathcal S^5((\Lambda_\theta \otimes \iota_C)\rho_{AC}),  & cones in $\mathcal R_C \otimes \mathcal H_B \otimes \mathcal H_C$ that correspond to the CR-type bounds $S_1, S_2, S_3, S_4, S_5$ \\
    \bar S_k \coloneqq \min_{\rho_{AC} } \bar S_k [\rho_{AC}] & CR-type precision bound for the channel estimation problem using the optimal probe state $\rho_{AC}$\\
 \mathcal S^1_{BA},
 \mathcal S^2_{BA},
 \mathcal S^3_{BA},
 \mathcal S^4_{BA},
 \mathcal S^5_{BA}(T)
 &
 cones in $\mathcal R_C \otimes \mathcal H_B \otimes \mathcal H_A$ that 
 are analogous to the cones in the conic optimizations of 
 $S_1, S_2, S_3, S_4, S_5$ \\
    J_1, J_2, J_3, J_4, J_5,  & 
    conic optimizations on cones 
    $ \mathcal S^1_{BA},
 \mathcal S^2_{BA},
 \mathcal S^3_{BA},
 \mathcal S^4_{BA},
 \mathcal S^5_{BA}(T)$\\    
 Y & optimization variable for the cones 
 $ \mathcal S^1_{BA},
 \mathcal S^2_{BA},
 \mathcal S^3_{BA},
 \mathcal S^4_{BA},
 \mathcal S^5_{BA}(T)$\\
\rho_{A}(Y) &  
When $Y$ is the optimal solution of CR-type bound, 
the purification of $\rho_{A}(Y)$ is the optimal entangled state\\
|\Phi\>\<\Phi|  & maximally entangled state, an example of a $\rho_{AC}$\\
\mathcal K & set of matrices on system A and B for which the partial trace on B is proportional to the identity matrix\\
    \bottomrule
    \end{tabular}
    \caption{Notations for the channel estimation problem.}
\label{table:notations}
\end{table*}

Let us define 
${\cal S}(\mathbb C^{d+1}\otimes {\cal H})_{\rm PPT}$
as the set of self-adjoint operators on $\mathbb C^{d+1} \otimes {\cal H}$ with positive partial transpose, and define $\mathcal S^3$ as ${\cal S}(\mathbb C^{d+1}\otimes {\cal H})_{\rm PPT} \cap {\cal B}''$.
Likewise, we define the set ${\cal S}(\mathbb C^{d+1}\otimes {\cal H})_{\rm P}$ as the set of 
positive semi-definite self-adjoint operators on $\mathbb C^{d+1} \otimes {\cal H}$,
and define 
$\mathcal S^4$ as ${{\cal S}(\mathbb C^{d+1}\otimes {\cal H})_{\rm P} \cap {\cal B}''}$.
Then, for $k=1,2,3,4$,
we define 
\begin{align}
S_k:=
 \min_{X \in {\cal S}^k}
\{\Tr ( G \otimes \rho) X |
\eqref{c1},\eqref{c2} \hbox{ hold.}
\}
\Label{o1-T}.
\end{align}
The relation 
\begin{align}
\mathcal S^1 \subset \mathcal S^2 
\subset \mathcal S^3
\subset \mathcal S^4 
\notag
\end{align}
 implies 
\begin{align}
S_1 \ge S_2 \ge S_3 \ge S_4.\label{IN1}
\end{align} 
 
In addition, we introduce a linear constraint to the operator $X \in {\cal B}''$ as
\begin{align}
\Tr X( ( |j \rangle \langle i |-  |i \rangle \langle j |)\otimes T )=0\Label{c3}
\end{align}
for $i,j=1,2, \ldots, d$ and a trace-class self-adjoint operator $T$.
We may simplify \eqref{c3} to 
\begin{align}
\Im \Tr X(  |j \rangle \langle i |\otimes T )=0.
\notag
\end{align}
by noting that 
$\Tr X(  |i \rangle \langle j |\otimes T )$
is the complex conjugate of $\Tr X(  |j \rangle \langle i |\otimes T )$
because 
$\Tr_R X (|i \rangle \langle j |\otimes I)
=(\Tr_R X (|j \rangle \langle i |\otimes I))^\dagger$. 
Using this linear constraint, we define the subspace ${\cal B}_T''$ of ${\cal B}''$
as
\begin{align}
{\cal B}_T'':=\{X \in {\cal B}''|
\eqref{c3} \hbox{ holds.}
\}.
\notag
\end{align}
Next, given a density matrix $\rho$ on $\mathcal H$, 
let us define ${\cal S}^{5}(\rho)$
as ${\cal S}(\mathbb C^{d+1}\otimes {\cal H})_{\rm P}\cap {\cal B}''_\rho$.
We consider the minimization:
\begin{align}
&S_5:= \min_{X \in {\cal S}^5(\rho)}
\{\Tr ( G \otimes \rho) X |
\eqref{c1},\eqref{c2} \hbox{ hold.}
\}.
\notag
\end{align}
Note that only the cone $\mathcal S^5(\rho)$ depends on $\rho$, for the cones $\mathcal S^1,\mathcal S^2,\mathcal S^3,\mathcal S^4$ are independent of $\rho$.
Since we have the relation
\begin{align}
 {\cal S}^2
\subset {\cal S}^5(\rho)
\subset {\cal S}^4,
\notag
\end{align}
we have the following relations
\begin{align}
S_4 \le S_5 \le S_2. \label{IN2}
\end{align}
Since $S_k$ depends on
the model $(\rho, (D_j)_j)$, i.e., the probe state $\rho$ and the partial derivatives of the probe state $D_j$ for $k=1,2,3,4,5$,
we can also write $S_k$ as
\begin{align}
S_k[\rho,(D_j)_j]\notag
\end{align}
 to emphasize the CR-type bounds' dependence on $\rho$ and $D_j$.
 
Ref.~\cite{HO} showed that 
$S_2$ equals the Nagaoka-Hayashi bound (NH bound) studied in Ref.~\cite{Nagaoka-generalization,Hayanoncomm,CSLA}.
Also, Ref.~\cite{HO} showed that $S_4$ equals the SLD bound, and
$S_5$ equals the HN bound.
In the single parameter case, i.e., $d=1$,
the SLD bound is attainable. Hence, we have the equality in \eqref{IN1} and \eqref{IN2}, i.e., the equation:
\begin{align}
S_1 = S_2 = S_3 =S_5= S_4.\label{IN3}
\end{align}

For further discussion, we prepare several notations.
We use the notation $X\circ Y:= \frac{1}{2} (XY+YX)$.
We define the SLD operator $L_j$ on $\cH$ as
\begin{align}
\rho \circ L_j=D_{j}.
\notag
\end{align}
We denote the SLD fisher information matrix by $J_{\rm SLD}$.
We define $L^i:= \sum_{j=1}^d(J_{\rm SLD}^{-1})^{i,j}L_j$ as a linear combination of SLD operators that depend on the $i$th row of the inverse SLD Fisher information matrix.
These operators $L^i$ satisfy the constraint $ \Tr D_j L^i=
\Tr D_j\sum_{j'=1}^d(J_{\rm SLD}^{-1})^{j,j'}L_{j'}
=\sum_{j'=1}^d(J_{\rm SLD}^{-1})^{i,j'} \Tr D_j L_{j'}
=\sum_{j'=1}^d(J_{\rm SLD}^{-1})^{i,j'} J_{{\rm SLD}, j,j'}
=\delta_j^i$.
When we need to clarify the dependence of the model
$(\rho, (D_{j})_j)$,
$L^i$, $L_i$, and $J_{\rm SLD}$ are denoted by
$L^i[\rho, (D_{j})_j]$, $L_i[\rho, (D_{j})_j]$, 
and $J_{\rm SLD}[\rho, (D_{j})_j]$,
respectively.

Now, we consider the $n$-fold tensor product system $\cH^{\otimes n}$.
Given an operator $X$ on $\cH$, we define the operator 
$X^{(n)}$ on $\cH^{\otimes n}$
 as $X^{(n)}:= \sum_{i=1}^n X^{(n)}_i$, where
$X^{(n)}_i:= I^{\otimes i-1}\otimes X \otimes I^{n-i} $.
Then, we call the model $(\rho^{\otimes n},(D_{j}^{(n)})_j)$
as the $n$-fold extension of the state model 
$(\rho,(D_{j})_j)$.
Then, we have
\begin{align}
L_i[\rho^{\otimes n},(D_{j}^{(n)})_j]&= (L_i[\rho, (D_{j})_j])^{(n)}\notag \\
J_{\rm SLD}[\rho^{\otimes n},(D_{j}^{(n)})_j]&= n J_{\rm SLD}[\rho, (D_{j})_j]) \notag\\
L^i[\rho^{\otimes n},(D_{j}^{(n)})_j]&=\frac{1}{n} (L^i[\rho, (D_{j})_j])^{(n)}\notag .
\end{align}
Also, the relation
\begin{align*}
S_k[\rho,(D_j)_j]=
n S_k[\rho^{\otimes n},(D_{j}^{(n)})_j] 
\end{align*}
holds for $k=4$ \cite{Nagaoka} and for $k=5$
\cite[Lemma 4]{HM}.

\section{Channel estimation} \label{sec:channel-estimation}

In the previous section, the optimizations $S_k$ depend on the quantum model comprising of density operators $\rho$ and their partial derivatives with respect to the parameters to be estimated.
Here, we consider the channel estimation problem, where 
we have a $d$-parameter channel family $\{ \Lambda_\theta\}$, where the input system is ${\cal H}_A$
and the output system is ${\cal H}_B$.
As we can see, the quantum channel $\Lambda_\theta$ embeds the $d$ parameters $\theta$ to be estimated. After the quantum channel $\Lambda_{\theta}$ maps a probe state on ${\cal H}_A$ to an output state $\Lambda_\theta(\rho)$ on $\mathcal H_B$,
With asymptotically many copies of quantum states $\Lambda_\theta(\rho)$,
we can perform appropriate measurements, either in a correlated or uncorrelated setting, to obtain a probability distribution that depends on the parameters, from which we may construct the best informated estimator for the parameters. 

In the channel estimation problem, it is typical to consider the purification of the probe state $\rho$ to a pure state on $\mathcal H_A \otimes \mathcal H_C$. Here, the Hilbert space $\mathcal H_C$, isomorphic to $\mathcal H_A$, is an ancillary system that the quantum channel has not access to.
Our measurement strategies however do have access the ancillary system. 
This setting follows the preceding paper \cite{Liu} which studies the optimization of the one-parameter case under the assumption that the ancilla system is available. 
We remark that considering measurement strategies that do not have access to this ancilla system is highly non-trivial, and beyond the scope of our current study. In what follows, 
we assume that we have full access to the ancilla system, which in turn means that 
the quantities of interest are
\begin{align}
\bar{S}_k[\rho_{AC}]
:=
S_k[(\Lambda_\theta \otimes  \iota_C )( \rho_{AC}) ,  ( \frac{\partial}{\partial \theta_j}(\Lambda_\theta \otimes  \iota_C )(  \rho_{AC}))_j ],\label{Sk-AC}
\end{align}
where $\iota_C$ denotes an identity channel on system $\mathcal H_C$.
With complete access to the ancilla system, the corresponding CR-type precision bounds are then 
\begin{align}
\bar S_k := \min_{\rho_{AC}} \bar{S}_k[\rho_{AC}]
\Label{min-Sk-AC}
\end{align}
where the minimization is over all pure density operators on $\mathcal H_A \otimes \mathcal H_C$,
and the accessible measurements for our estimation are applied to 
the tensor product system ${\cal H}_B\otimes {\cal H}_C$.
Since the whole ancialla system is accessible,
when $\rho_{AC}$ is a mixed state,
we are allowed to retake the system ${\cal H}_C$
to contain the reference of the purification of $\rho_{AC}$.
This explains why 
the range of the above minimization is limited to all pure density operators on $\mathcal H_A \otimes \mathcal H_C$.

It is not immediately obvious how one would solve the optimization in \eqref{min-Sk-AC}. Even if one solves the inner optimization for $S_k$, it is unclear if the subsequent optimization in $\rho_{AC}$ is a tractable optimization problem, such as a convex problem. 
We overcome these difficulties. Namely,
 we construct conic programs with optimal values that are precisely equal to \eqref{min-Sk-AC}, which allows us to find \eqref{min-Sk-AC} using only a single optimization program.

The first tool that we use is the Choi matrix of a quantum channel \cite{choi1975completely}. 
We denote the Choi matrix of $\Lambda_\theta$ by $T_\theta$. 
Given any orthonormal basis $\{ |e_j\rangle \}$ of ${\cal H}_A$, 
we can define the unnormalized maximally entangled state 
$|I\rangle:=\sum_{j}|e_j\rangle|e_j\rangle$ on 
${\cal H}_{A'} \otimes {\cal H}_{A}$, where 
${\cal H}_{A'}$ is isomorphic to ${\cal H}_{A}$.
Then, the Choi matrix $
T_\theta$ of $\Lambda_\theta$ is an operator on ${\cal H}_B \otimes {\cal H}_{A}$ and is given by
\begin{align*}
T_\theta:= (\Lambda_\theta \otimes \iota)
( |I\rangle \langle I|).
\end{align*}
Since $\rho= \Tr_{A}[|I\rangle \langle I|  (I_{A'} \otimes \rho)]$ for any input state $\rho$ on ${\cal H}_{A}$, 
where $I_{A'} = \sum_j |e_j\>\<e_j|$
we have
\begin{align}
& \Lambda_\theta(\rho)=
\Lambda_\theta(\Tr_{A}[|I\rangle \langle I|  (I_{A'} \otimes \rho)]) \nonumber \\
=&
\Tr_{A}[ (\Lambda_\theta \otimes \iota)
(|I\rangle \langle I|)  (I_B \otimes \rho))] 
=\Tr_{A} [T_\theta (I_B \otimes \rho)]\label{action of channel},
\end{align}
where $I_B$ denotes the identity operator on $\mathcal H_B$.
Here, \eqref{action of channel} allows us to rewrite $\Lambda_\theta(\rho)$ in terms of 
the Choi matrix $T_{\theta}$ and the input probe state $\rho$.

Next, we can consider the partial derivatives of $\Lambda_\theta(\rho)$
in terms of the Choi matrix derivatives. Namely,
\begin{align}
\frac{\partial}{\partial \theta^j}\Lambda_\theta(\rho)
=\Tr_A[ \frac{\partial}{\partial \theta^j} T_\theta (I_B \otimes \rho)].
\notag
\end{align}

When the parameter $\theta$ is in the neighborhood of the true parameter $\theta_0$, 
we denote the corresponding Choi matrix $T_{\theta_0}$ by $T$ 
and its derivatives as $F_j:= \frac{\partial}{\partial \theta^j} T_\theta|_{\theta=\theta_0}$.
In this case, we can write
\begin{align}
\frac{\partial}{\partial \theta^j}\Lambda_\theta(\rho)|_{\theta=\theta_0}
=\Tr_A[ F_j (I_B \otimes \rho)].\label{action of channel derivatives}
\end{align}

Since we allow access to an ancilla system $\mathcal H_C$ for both preparation and measurement of the probe state, instead of $\rho$ as the input state, we consider $\rho_{AC}$ on $\mathcal H_A \otimes \mathcal H_C$ as the input state. When the quantum channel still maps $\mathcal H_A$ to $\mathcal H_B$, then the output state that corresponds to our input state $\rho_{AC}$ at the true parameter value $\theta_0$ is 
\begin{align}
(\Lambda_\theta \otimes \iota_C)(\rho_{AC})|_{\theta=\theta_0}
&=\Tr_A[ (T \otimes I_C)(I_B \otimes \rho_{AC})], \label{output state}
\end{align}
and its derivatives are
\begin{align}
\frac{\partial}{\partial \theta^j}(\Lambda_\theta \otimes \iota_C)(\rho_{AC})|_{\theta=\theta_0}
=&
\Tr_A[ (F_j  \otimes I_C)(I_B \otimes \rho_{AC})].\label{output state derivatives}
\end{align}
In this setting, 
the joint system ${\cal H}_B \otimes {\cal H}_C$
is accessible for our measurement for our estimation. Then,
the CR-type bound 
$\bar{S}_k[\rho_{AC}]$
of channel estimation problem as given in \eqref{Sk-AC} can be written as 
\begin{align}
&\bar{S}_k[\rho_{AC}] \notag\\
=&
S_k[\Tr_A[ (T\otimes I_C) (I_B \otimes \rho_{AC})],
(\Tr_A [(F_j\otimes I_C)( I_B \otimes \rho_{AC})])_j ]
\label{Sk-AC-Choi}
\end{align}
with $k=1,2,3,4,5$. 
In the following, the pair $(T,(F_j)_j)$ is called the channel model.
In particular, when we need to clarify the dependence of the channel model,
$\bar{S}_k[\rho_{AC}]$ is denoted by $\bar{S}_k[T,(F_j)_j,\rho_{AC}]$.

Note that the operators in the left sides of \eqref{output state} and \eqref{output state derivatives} act on the space $\mathcal H_B \otimes \mathcal H_C$.
In the conic programming formulation of \eqref{Sk-AC}, we require the introduction of the space ${\cal R}_C = \mathbb C^{d+1}$, and consider conic programming over the tripartite system 
${\cal R}_C \otimes {\cal H}_B \otimes {\cal H}_C$.
For $k=1,2,3,4$, we define the cones we consider are operators on ${\cal R}_C \otimes {\cal H}_B \otimes {\cal H}_C$ which are equal to ${\cal S}^k$ as explained in Section \ref{sec:estimation} under the choice ${\cal H}={\cal H}_B \otimes {\cal H}_C$ as ${\cal S}^k_{BC}$.
Following the notation of Section \ref{sec:estimation}, we use $X$ on ${\cal R}_C \otimes {\cal H}_B \otimes {\cal H}_C$ to optimize within the appropriate cones.

After solving \eqref{Sk-AC}, we need to perform a subsequent minimization over all states $\rho_{AC}$ on the space $\mathcal H_A \otimes \mathcal H_C$, which is a challenging task.
To overcome this challenge, we formulate new conic programs with optimal values equal to those given by \eqref{min-Sk-AC}.

Our formulation of new conic programs draws upon the insight that the space $\mathcal H_C$ is in fact 
isomorphic to the space  $\mathcal H_A$, and that we may consider optimizing over cones on 
${\cal R}_C \otimes {\cal H}_B \otimes {\cal H}_A$
instead of 
${\cal R}_C \otimes {\cal H}_B \otimes {\cal H}_C$.
This means that instead of considering the cones ${\cal S}^k_{BC}$,
we like to consider the cones ${\cal S}^k_{BA}$, where we obtain ${\cal S}^k_{BA}$
by replacing $\mathcal H$ with ${\cal H}_B \otimes {\cal H}_A$ in the definition of $\mathcal S^k$.

One challenge in this idea is to be able to construct $\rho_{AC}$ from the optimal solution of the cone. Denoting $Y$, an operator on ${\cal R}_C \otimes {\cal H}_B \otimes {\cal H}_A$, as the optimization variable in this case, this challenge can be reduced to that of finding an appropriate density operator $\rho_A(Y)$ on $\mathcal H_A$ that depends on $Y$, and subsequently purifying $\rho_A(Y)$ to $\rho_{AC}$.
To formulate the objective function of the conic programs on $Y$, 
we revisit the optimization of \eqref{Sk-AC-Choi}.
Note that, the objective function of \eqref{Sk-AC-Choi} is 
\begin{align}
&\Tr[  G \otimes \Tr_A [ (T\otimes I_C) (I_B \otimes \rho_{AC})] X ]\notag\\
=&
\Tr[  G   \otimes( ( T\otimes I_C) (I_B \otimes \rho_{AC})) (I_A \otimes X )]\notag\\
=&
\Tr[(  G \otimes T) \Tr_C[(I_{RB} \otimes \rho_{AC}) (I_A \otimes X )] ].\notag
\end{align}
Identifying $Y$ as $ \Tr_C[(I_{RB} \otimes \rho_{AC}) (I_A \otimes X )]$, the objective function of the our conic programs becomes
\begin{align}
\Tr[Y ( G \otimes T) ].\notag
\end{align}
We establish an appropriate relation between the variable $Y$ and some $\rho_{AC}$ through the conditions:
  
\begin{description}
\item[(i)]
Given fixed $Y$ on $\mathcal R_C \otimes \mathcal H_B \otimes \mathcal H_A$, there exists a state $\rho_A$ on ${\cal H}_A$ such that 
\begin{align}
\Tr_{R}[ Y (|0\rangle \langle 0|\otimes I_{AB})]
=
I_B \otimes \rho_A\Label{NMZA}.
\end{align}
\item[(ii)]
\begin{align}
\frac{1}{2}\Tr[ Y
( (|0\rangle \langle j'|+|j'\rangle \langle 0|) \otimes F_j)] =\delta_{j,j'}\notag
\end{align}
for $j,j'=1, \ldots, d$.
\end{description}
While the condition (ii) is a linear constraint, it is not immediately apparent how condition (i) can be written as a linear constraint.
However, we point out that the condition (i) is in fact equivalent to the following linear constraint.
\begin{description}
\item[(i')]
Let $\{|b\rangle\}$ be any orthonormal basis of ${\cal H}_B$.
For $b\neq b'$,
we have
\begin{align}
&\Tr_{RB}[ Y (|0\rangle \langle 0|\otimes 
I_A
\otimes |b\rangle \langle b'|)]
=0 ,\Label{NBFY2}\\
&\Tr_{RB} [Y (|0\rangle \langle 0|\otimes 
I_A
\otimes |b\rangle \langle b|)]\notag
\\
&=
\Tr_{RB}[ Y (|0\rangle \langle 0|\otimes 
I_A
\otimes |b'\rangle \langle b'|)]\Label{NBFY}
\end{align}
as operators on $\cH_A$.
Also, 
\begin{align}
\Tr[ Y (|0\rangle \langle 0|\otimes I_A\otimes |b\rangle \langle b|)]=1.\Label{ZKT}
\end{align}
\end{description}

The equivalence between (i) and (i') is shown as follows.
Since (i) $\Rightarrow$ (i') is trivial, 
we show (i') $\Rightarrow$ (i). 
We assume that (i').
We denote
$
\Tr_B [\Tr_{R}[ Y (|0\rangle \langle 0|\otimes I_{AB})]
(I_A \otimes |b\rangle \langle b|)]
=
\Tr_{RB}[ Y (|0\rangle \langle 0|\otimes 
I_A
\otimes |b\rangle \langle b|)]
$ by $\rho_A(Y)$.
Then, 
\eqref{ZKT} guarantees that 
$\rho_A(Y)$ is a state.
Also, \eqref{NBFY2} and \eqref{NBFY}
imply \eqref{NMZA}.
We obtain the condition (i). 
Therefore, the condition (i) is replaced by the linear constraint (i').
Furthermore, when we impose the condition (i), the operator
\begin{align}
\rho_A(Y) := \frac{1}{d_B}\Tr_{RB}[ Y(|0\rangle \langle 0|\otimes I_{AB})]
\label{rhoA def}
\end{align}
is a density matrix, where $d_B$ is the dimension of ${\cal H}_B$.
Given conditions (i') and (ii) along with the objective function $\Tr[Y(G \otimes T)]$, we can define our new conic programs. 
Namely for $k=1,2,3,4$, we define 
\begin{align}
J_k:= \min_{Y \in {\cal S}_{BA}^k}
\{\Tr [Y (G \otimes T)]|Y \hbox{ satisfies }
\hbox{(i'), (ii).} \}.\notag
\end{align}
For $k=5$, 
we define ${\cal S}^5_{BA}(T):=
{\cal S}({\cal R}_C\otimes {\cal H}_{AB})_{\rm P}\cap {\cal B}''_T$,
and define the conic linear programming.
\begin{align}
J_5:= \min_{Y \in {\cal S}_{BA}^5(T)}
\{\Tr [Y (G \otimes T)]|Y \hbox{ satisfies }
\hbox{(i'), (ii).} \}.\notag
\end{align}
When we need to clarify the dependence of the channel model $(T,(F_j)_j)$,
$J_k$ is denoted by $J_k[T,(F_j)_j]$. 

 \begin{figure}[htbp]  
    \centering
    \includegraphics[width=0.49\textwidth]{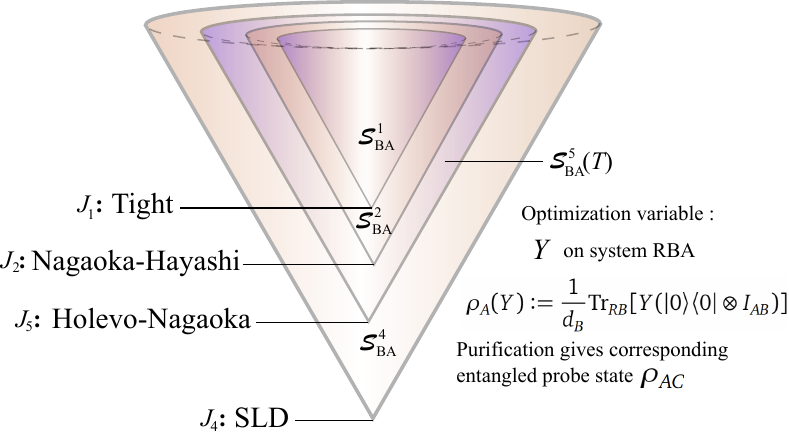}    
    \caption{.   
     \label{fig:cones} We depict the relationship between $J_1,J_2,J_5,J_4$ 
     and their associated cones $\mathcal S^1_{BA},\mathcal S^2_{BA},\mathcal S^5_{BA}(T),\mathcal S^4_{BA}.$ The optimal values of the these optimizations, which are minimizations, are the tight bound, the Nagaoka-Hayashi bound, 
     the Holevo-Nagaoka bound, and the SLD bound respectively. We pictorally illustrate that $J_1 \ge J_2 \ge J_5 \ge J_4$.
     We denote the optimization variable for all of these minimizations as $Y$, which is a matrix supported on $\mathcal R_C \otimes \mathcal H_B \otimes \mathcal H_A$.
From the optimal solution $Y^*$ of any of these conic programmings, we can derive a corresponding $\rho_A(Y^*)$, whose purification to system AC yields the corresponding optimal probe state.}  
\end{figure}

The main result of our paper is the following theorem.
\begin{theorem}\Label{NZD}
For $k=1,2,3,4,5$, we have
\begin{align}
J_k=\bar S_k. \Label{MAIN}
\end{align}
\end{theorem}
This theorem allows calculation of precision bounds for channel estimation optimized over probe states as given by $\bar S_k$ using the conic programs that correspond to $J_k$. 
Since the cones considered in $J_k$ are analogous to the cones considered in $S_k$ \cite{HO}, we know how to solve $J_k$ numerically. 
Furthermore for $k=2,3,4,5$, we can solve $J_k$ via semidefinite programming (SDP).
In contrast, the calculation of $J_1$ requires the minimization over a certain separable cone on $\mathcal R_C \otimes (\mathcal H_B \otimes \mathcal H_A)$. Such a type conic programming is discussed in Section IV of \cite{HO}.

In addition, 
in the single parameter case, i.e., $d=1$,
the SLD bound is attainable. 
\eqref{IN3} implies
\begin{align}
J_1 = J_2 = J_3 =J_5= J_4.\label{IN4}
\end{align}

First, 
$J_1$ has the operational meaning of the precision limit under the scenario shown in Fig. \ref{fig1:operational-meaning}.
We can choose our input state for each single input system 
individually.
Our measurement can be done over the joint system of
the single output system and the ancilla system of the single input.
In addition, we allow adaptive improvement for the choice of the input state and the measurement.

Second, $J_5$ has the operational meaning of the precision limit under the scenario shown in Fig. \ref{fig2:operational-meaning}.
Again, we can choose our input state for each single input system 
individually.
Our measurement can be done over the joint system of
all output systems and all ancilla systems of the inputs.
However, when our measurement can be done over the joint system of
the initial $k_1$ output systems and the initial $k_1$ ancilla systems of the inputs,
we can adaptively choose the next $k_2$ input states
depending on the above measurement outcome.
Then, depending on the above measurement outcome,
our measurement can be done over the joint system of
the next $k_2$ output systems and the next $k_2$ ancilla systems of the inputs.
Such an adaptive improvement over several rounds is allowed.

Now we sketch Theorem \ref{NZD}'s proof. 
First we can establish an upper bound on $J_k$ in terms of $\bar S_k$ in the following lemma. 
\begin{lemma}\Label{L1}
For $k=1,2,3,4,5$, we have
\begin{align}
J_k\le \bar S_k.\Label{MAIN1}
\end{align}
\end{lemma}
To prove Lemma \ref{L1}, we show 
that given any solution $\rho_{AC}$ and $X$ to $\bar S_k$, we can also construct a corresponding solution $Y$ for $J_k$ with the same value for the objective function. 
We prove the inequality opposite to \eqref{MAIN1},
based on the discussion in Section \ref{sec:estimation}, 
we rewrite 
the constraints for the completeness condition and the locally unbiased condition as
\begin{align}
& \Tr_R [X (|0\rangle \langle 0|\otimes I_{BC})] 
=I_{BC} \Label{NB2}\\
&\frac{1}{2}\Tr [(I_A \otimes X ) 
((|0\rangle \langle j'|+|j'\rangle \langle 0|) \otimes F_j \otimes I_C)
(I_{RB}\otimes \rho_{AC})]
=\delta_{j,j'}. \Label{NB3}
\end{align}
Then, we show the following lemma.

\begin{lemma}\Label{L3}
For $k=1,2,3,4,5$,
we choose $Y \in {\cal S}_{BA}^k$ satisfying the conditions (i), (ii), and $\rho_A:=\rho_A(Y)>0$.
We diagonalize $\rho_A$ as
$\sum_{j=1}^{d_A} s_j |\phi_j\rangle \langle \phi_j|$.
We choose an orthonormal basis 
$\{\psi_j\}$ of ${\cal H}_C$.
We define a unitary map
$U:\phi_j \mapsto \psi_j $ from $\cH_A$ to $\cH_C$.
We choose $\rho_{AC}$ as the pure state
$\sum_{j=1}^{d_A}\sqrt{s_j}|\phi_j,\psi_j\rangle$, which is a purification of 
$\rho_A$.
Then, we have
\begin{align}
&\Tr [Y (G \otimes T)]  \ge \bar S_k.\Label{ZIT}
\end{align}

In addition,
we choose $X$
as 
\begin{align}
X:= 
(U\otimes I_{RB}) (\rho_A^{-1/2}\otimes I_{RB}) Y(\rho_A^{-1/2} \otimes I_{RB} )(U^\dagger \otimes I_{RB}) .
\Label{SSD}
\end{align}

For $k=1,2,3,4$, 
$X$ is an element of ${\cal S}_{BC}^k $,
and satisfies
the conditions \eqref{NB2} and \eqref{NB3},
and the following relation.
\begin{align}
\Tr [Y (G \otimes T) ]
=\Tr[
\Tr_{C}[ (I_{RB} \otimes \rho_{AC}) (I_A \otimes X) ]
 ( G \otimes T)]
.\Label{MXI}
\end{align}

For $k=5$,
$X$ is an element of ${\cal S}_{BC}^4 $,
and satisfies the conditions \eqref{NB2} and \eqref{NB3},
\eqref{MXI}, and 
\begin{align}
\Tr[ (I_A \otimes X ) 
((|i\rangle \langle j'|-|j'\rangle \langle i|) \otimes T \otimes I_C)
(I_{RB}\otimes \rho_{AC})]
=0.\Label{MXI2}
\end{align}
\end{lemma}

Lemma \ref{L3} assumes that $Y$ satisfies the full rank condition 
for $\rho_A(Y)$.
If the minimization in $J_k$ is achieved by an operator $Y$ to satisfy this condition, 
we obtain Theorem \ref{NZD}.
However, there is a possibility that 
the optimal operator $Y$ does not satisfy this condition.
To cover such a possibility, we need the following technical lemma.

\begin{lemma}\Label{L2}
We have
\begin{align}
J_k= \inf_{Y \in \mathcal S^k_{BA}}
\left\{\Tr [Y (G \otimes T)]\left|
\begin{array}{l}
Y \hbox{ satisfies }
\hbox{(i), (ii)
and, }\\
\rho_A(Y)>0, i.e.,\\
 \rho_A(Y) \hbox{ is full rank.}
\end{array}
\right.\right \}\Label{MMFB}
\end{align}
for $k=1,2,3,4$, and 
\begin{align}
J_5= \inf_{Y \in \mathcal S^5_{BA}(T)}
\left\{\Tr [Y (G \otimes T)]\left|
\begin{array}{l}
Y \hbox{ satisfies }
\hbox{(i), (ii)
and, }\\
\rho_A(Y) \hbox{ is full rank.}
\end{array}
\right.\right
 \}\Label{MMFB5}
\end{align}
\end{lemma}

Notice that both sides in the definition \eqref{SSD}
act on
${\cal H}_{R}\otimes {\cal H}_{B}\otimes {\cal H}_{C}$
because $U$ maps $\cH_A$ to $\cH_C$.
The combination of Lemma \ref{L2} and 
\eqref{ZIT} in Lemma \ref{L3} yields
the $\ge$ part of \eqref{MAIN}
while Lemma \ref{L1} shows
the $\le$ part of \eqref{MAIN}.
Therefore, the combination of Lemmas \ref{L1},
\ref{L3}, and \ref{L2} shows Theorem \ref{NZD}.
In addition, Lemmas \ref{L1},
\ref{L3}, and \ref{L2} are shown in Appendix \ref{sec:proofs}.

The pure state $\rho_{AC}$ given in Lemma \ref{L3} 
is the input state to achieves 
the value $\Tr [Y (G \otimes T)]$ in the respective bounds.
The operator $X$ given in Lemma \ref{L3} 
achieves 
the value $\Tr [Y (G \otimes T)]$ in the respective 
conic linear programming.

\section{Case with maximally entangled input state}
\label{sec:maxent}
\subsection{Notations}
The channel model is given as
the pair $(T^{\otimes n},(F_j^{(n)})_j)$.
Fixing the input state to be the maximally entangled state 
$|\Phi\rangle \langle \Phi|$, for $k=1,2,3,4,5$,
gives the following inequalities.
\begin{align}
J_k[T,(F_j)_j] &\le \bar{S}_k[T,(F_j)_j,|\Phi\rangle \langle \Phi|] \Label{NBR1}.
\end{align}
The aim of this section is to derive a necessarily and sufficient condition for the equality of the above inequality for $k=4,5$.

For this aim, we discuss when a maximally entangled input state
realizes the minimum value in the respective bounds.
In this section, we simplify 
$J_k[T,(F_j)_j] $ and 
$ \bar{S}_k[T,(F_j)_j,|\Phi\rangle \langle \Phi|] $
to
$J_k$ and $ \bar{S}_k[|\Phi\rangle \langle \Phi|] $, respectively.
We employ the normalized Choi state 
$T_N:= \frac{1}{d_A}T$ on ${\cal H}_{AB}$.
Also, we employ $F_{j,N}:= \frac{1}{d_A}F_j$.
Then, we impose the following conditions to $Y_N:= d_A Y$.

\begin{description}
\item[(i-N)]
Given fixed $Y_N$ on $\mathcal H_R \otimes \mathcal H_B \otimes \mathcal H_A$, there exists a state $\rho_A$ on ${\cal H}_A$ such that 
\begin{align}
\Tr_{R}[ Y_N (|0\rangle \langle 0|\otimes I_{AB})]
=
I_B \otimes d_A \rho_A\Label{NMZA-N}.
\end{align}
\item[(ii-N)]
\begin{align}
\frac{1}{2}\Tr[ Y_N
( (|0\rangle \langle j'|+|j'\rangle \langle 0|) \otimes F_{N,j})] =\delta_{j,j'}\notag
\end{align}
for $j,j'=1, \ldots, d$.
\end{description}
The condition (i-N) is equivalent to the following linear constraint.
\begin{description}
\item[(i'-N)]
Let $\{|b\rangle\}_{b=1}^{d_B}$ be any orthonormal basis of ${\cal H}_B$.
For $b \in \{1,\ldots, d_B-1\}$ and
$b' \in \{2,\ldots, d_B-1\}$ with $b>b'$
we have
\begin{align}
&\Tr_{RB}[ Y_N (|0\rangle \langle 0|\otimes 
I_A
\otimes |b \rangle \langle b'|)]
=0 ,\Label{NBFY2-N}\\
&\Tr_{RB} [Y_N (|0\rangle \langle 0|\otimes 
I_A
\otimes (|b\rangle \langle b|-|b+1\rangle \langle b+1|)
]
= 0
\Label{NBFY-N}
\end{align}
as operators on $\cH_A$.
Also, 
\begin{align}
\Tr[ Y_N (|0\rangle \langle 0|\otimes I_A\otimes |1\rangle \langle 1|)]=d_A.\Label{ZKT-N}
\end{align}
\end{description} 

Then, $J_k$ for $k=1,2,3,4$ is rewritten as
\begin{align}
J_k= \min_{Y_N \in {\cal S}_{BA}^k}
\{\Tr [Y_N (G \otimes T_N)]|Y_N \hbox{ satisfies }
\hbox{(i'-N), (ii-N).} \}.\Label{NVF1}
\end{align}
$J_5$ is rewritten as
\begin{align}
J_5= \min_{Y \in {\cal S}_{BA}^5(T_N)}
\{\Tr [Y_N (G \otimes T_N)]|Y_N \hbox{ satisfies }
\hbox{(i'-N), (ii-N).} \}.\Label{NVF2}
\end{align}

When the input state is fixed as a maximally entangled state $|\Phi\rangle$, the bound is 
$\bar{S}_k[|\Phi\rangle\langle \Phi|]$.
Then, we have $T_N= \Tr_A[ (T\otimes I_C) (I_B \otimes |\Phi\rangle\langle \Phi|)]$.
Hence, by replacing the system $\cH_C$ by $\cH_A$, we rewrite 
$\bar{S}_k[|\Phi\rangle\langle \Phi|]$ as
\begin{align}
&\bar{S}_k[|\Phi\rangle\langle \Phi|] \nonumber \\
=&
 \min_{Y_N \in {\cal S}_{BA}^k}
\{\Tr [Y_N (G \otimes T_N)]|Y_N \hbox{ satisfies }
\hbox{(i''), (ii-N).} \} \Label{MIN1}\\
=&{S}_k[T_{N},(F_{j,N})_j], \notag\\
&\bar{S}_5[|\Phi\rangle\langle \Phi|]\notag\\
=&
 \min_{Y_N \in {\cal S}_{BA}^5(T_N)}
\{\Tr [Y_N (G \otimes T_N)]|Y_N \hbox{ satisfies }
\hbox{(i''), (ii-N).} \} \Label{MIN2}\\
=&{S}_5[T_{N},(F_{j,N})_j],\notag
\end{align}
for $k=1,2,3,4$,
where the condition (i'') is defined as
\begin{description}
\item[(i'')]
\begin{align}
\Tr_{R}[ Y_N (|0\rangle \langle 0|\otimes I_{AB})]
=
I_{AB}\Label{NMTT}.
\end{align}
\end{description}

Therefore, we find that the difference between 
$J_k$ and $\bar{S}_k[|\Phi\rangle\langle \Phi|]$
is characterized by the difference between the conditions
(i'-N) and (i'').
To discuss this difference, 
we focus on
${S}_k[T_{N},(F_{j,N})_j]$, i.e., the model $(T_{N},(F_{j,N})_j)$
in the following discussion.

\subsection{Equality condition for $k=4$}
We choose $L^i_*:=L^i[T_N, (F_{j,N})_j]$, and define
$\vec{L}_*=(L^i_*)_i$.
For a vector of Hermitian matrices
$\vec{Z}=(Z^j)_j$, we define the operator
\begin{align}
W_{\rm SLD}(T_N,\vec{Z}):=\sum_{1\le i,j\le d} G_{j,i} Z^i T_N Z^j.\notag
\end{align}
We define the subset ${\cal K} \subset {\cal B}_{sa}({\cal H}_{AB})$ as
\begin{align}
{\cal K}:= \{X \in {\cal B}_{sa}({\cal H}_{AB})|
\Tr_{B}X \hbox{ is a constant times of }I_A.\}\notag
\end{align}
This subset relates to dual matrix-valued variable that corresponds to the condition 
(i'-N) in the primal problem. We supply the precise description of the dual problem in 
Appendix \ref{B1}.

 Then, we have the following theorem.
\begin{theorem}\Label{NNT} 
The following conditions are equivalent.
\begin{description}
\item[(A1)]
\begin{align}
J_4=S_4[|\Phi\rangle\langle \Phi|].\notag
\end{align}
\item[(A2)]
$\Tr_{B}W_{\rm SLD}(T_N,\vec{L}_*)$ belongs to ${\cal K}$.
\end{description}
\end{theorem}

\subsection{Equality condition for $k=5$}
Given the weight matrix $G$, 
for simplicity, we choose the new parameter 
$\tilde{\theta}:= \sqrt{G}\theta$.
By using the estimator $\hat{\theta}$ of the parameter $\theta$,
the new parameter's estimator
$\tilde{\theta}$ is given as
$\sqrt{G} \hat{\theta}$.
Hence, by using the covariance matrix $V_\theta[\hat{\Pi}]$
of the parameter $\theta$,
the covariance matrix 
$\tilde{V}_{\tilde{\theta}}[\hat{\Pi}]$
of the new parameter
$\tilde{\theta}$ is given as
$\sqrt{G}V_\theta[\hat{\Pi}]\sqrt{G}$.
That is, we have
\begin{align}
\Tr G V_\theta[\hat{\Pi}]= 
\Tr \tilde{V}_{\tilde{\theta}}[\hat{\Pi}].\notag
\end{align}
In other words,
under the new parameter $\tilde{\theta}$,
the weight matrix is given as 
the identity matrix $I$ 
so that the analysis on 
the weight matrix $I$
can recover the case with 
a general weight matrix $G$.
Therefore, without loss of generality,
we can assume that
the weight matrix is the identity matrix $I$.


We define $\vec{Z}_*$ as
\begin{align}
\vec{Z}_*:= \argmin_{\vec{Z}}
\{\Tr \Pi(\vec{Z}) G \otimes T_N| \Tr D_j Z^i=\delta_j^i\},\notag
\end{align}
where $\Pi(\vec{Z})^{i,j}= Z^iZ^j$, i.e.,
$\Pi(\vec{Z})= 
\sum_{1\le i,j\le d} |i\rangle\langle j| \otimes (Z^i)^\dagger Z^j
=(\sum_{1=1}^d |i\rangle \otimes (Z^i)^\dagger)
(\sum_{1=1}^d \langle i| \otimes Z^i)
$. 
We choose $Z^0_*:=I_{AB}$.
We define the matrix 
$V_*^{i,j}:=\Tr Z^i_* (Z^i_*)^\dagger T_N$.
We define $C_*:= \Im V_*|\Im V_*|^{-1}$.


Next, we define
\begin{align}
W_{HN}(T_n,\vec{Z}):= 
\Big(\sum_{i=1}^d Z^i T_N Z^i\Big)-
\sum_{1\le i,j\le d} \sqrt{-1}C_*^{i,j}
Z^j T_N Z^i \notag
\end{align}
which corresponds to a dual matrix-valued variable.
\begin{theorem}\Label{NNT2}
The following conditions are equivalent.
\begin{description}
\item[(B1)]
\begin{align}
J_5=S_5[|\Phi\rangle\langle \Phi|].\notag
\end{align}
\item[(B2)]
$\Tr_{B}W_{HN}(T_n,\vec{Z}_*)$ belongs to ${\cal K}$.
\end{description}
\end{theorem}

\subsection{Examples}
\subsubsection{One-parameter case}
We choose $\cH_A$ and $\cH_B$ as two dimensional systems 
spanned by $\ket{0},\ket{1}$.
We define the channel $\Lambda_0$ as a depolarizing channel:
\begin{align}
\Lambda_{0,p}(\rho):=(1-p)\rho+p \rho_{{\rm mix},B},\notag
\end{align}
where $\rho_{{\rm mix},B}$ is the completely mixed state on the output system $\cH_B$.
The channel $\Lambda_\theta$
is given as 
$\Gamma_{\theta,p}(\rho):=U_\theta \Lambda_{0,p}(\rho)
U_\theta^\dagger$,
where 
$U_\theta := \exp (i \theta \sigma_1)$
and 
$\sigma_1= \ket{0}\bra{1}+ \ket{1}\bra{0}$.
Then, the following theorem holds.
\begin{theorem}
\label{thm:1param}
When $G$ is $1$,
we have the following relation
\begin{align}
&J_k[
(\Gamma_{\theta,p}\otimes \iota)(|I\rangle \langle I|) ,
(
\frac{d}{d\theta}\Gamma_{\theta,p}\otimes \iota)(|I\rangle \langle I|) 
] \nonumber \\
=& \bar{S}_k[
(\Gamma_{\theta,p}\otimes \iota)(|I\rangle \langle I|) ,
(
\frac{d}{d\theta}\Gamma_{\theta,p}\otimes \iota)(|I\rangle \langle I|) ,
|\Phi\rangle \langle \Phi|] \nonumber\\
=& \frac{2-p}{8(1-p)^2}
 \Label{NBR7}
\end{align}
for $k=1,2,3,4,5$.
\end{theorem}

\begin{proof}
Due to \eqref{IN4},
it is sufficient to show \eqref{NBR7} with $k=4$.
Also, due to the symmetry, it is sufficient to show the case with $\theta=0$.
Then, $T_N= 
(1-p)|\Phi\rangle \langle \Phi|+ p \rho_{{\rm mix},AB}$
and $F_{1,N}= (1-p) [i \sigma_1 , |\Phi\rangle \langle \Phi|]$, where
$\rho_{{\rm mix},AB}$ is the completely mixed state on $\cH_{AB}$.
Since 
\begin{align}
[\sigma_1 , |\Phi\rangle \langle \Phi|] \circ |\Phi\rangle \langle \Phi|
&=\frac{1}{2}[\sigma_1 , |\Phi\rangle \langle \Phi|]\notag \\
[\sigma_1 , |\Phi\rangle \langle \Phi|] \circ \rho_{{\rm mix},AB}
&=\frac{1}{4}[\sigma_1 , |\Phi\rangle \langle \Phi|] ,\notag
\end{align}
using $ \frac{1}{2}(1-p)+\frac{1}{4}p=\frac{1}{4}(2-p)$,
we have $L_1= 4 \frac{1-p}{2-p}
[i \sigma_1 , |\Phi\rangle \langle \Phi|]$.
Since $\sigma_1^2=I$,
the SLD Fisher information in this model is 
\begin{align}
&\Tr L_1^2T_N=
\Tr 
(4 \frac{1-p}{2-p})^2
([i\sigma_1 , |\Phi\rangle \langle \Phi|])^2 T_N \nonumber\\
=&
16 \frac{(1-p)^2}{(2-p)^2} \Tr
(\sigma_1  |\Phi\rangle\langle \Phi| \sigma_1
+ |\Phi\rangle\langle \Phi| ) T_N \nonumber\\
=&
16 \frac{(1-p)^2}{(2-p)^2} \Tr
\Big((1-p)|\Phi\rangle\langle \Phi|
+\frac{p}{4} 
(\sigma_1  |\Phi\rangle\langle \Phi| \sigma_1
+ |\Phi\rangle\langle \Phi| ) \Big)\nonumber\\
=&
16 \frac{(1-p)^2}{(2-p)^2} 
((1-p)+\frac{p}{2} )
=
16 \frac{(1-p)^2}{(2-p)^2}  (1-\frac{1}{2}p) \nonumber\\
=&
 \frac{8(1-p)^2}{2-p}.\notag
\end{align}
Hence, we obtain the second equation in 
\eqref{NBR7} with $k=4$.

Then, 
$L^1_*
=\frac{2-p}{8(1-p)^2}\cdot 4 \frac{1-p}{2-p}[i \sigma_1 , |\Phi\rangle \langle \Phi|]
=\frac{1}{2(1-p)} [i \sigma_1 , |\Phi\rangle \langle \Phi|]$.
Since
\begin{align}
[i\sigma_1 , |\Phi\rangle \langle \Phi|]
|\Phi\rangle \langle \Phi|[ i\sigma_1 , |\Phi\rangle \langle \Phi|]
= &\sigma_1 |\Phi\rangle \langle \Phi|\sigma_1 \notag \\
[i\sigma_1 , |\Phi\rangle \langle \Phi|] 
\rho_{{\rm mix},AB}[ i\sigma_1 , |\Phi\rangle \langle \Phi|]
=& \frac{1}{4}(
\sigma_1  |\Phi\rangle\langle \Phi| \sigma_1
+ |\Phi\rangle\langle \Phi|),\notag
\end{align}
we have
\begin{align}
&W_{\rm SLD}(T_N,\vec{L}_*)= L^1_* T_N L^1_*\nonumber\\
=&
\frac{1}{4(1-p)^2}
((1-p)\sigma_1 |\Phi\rangle \langle \Phi|\sigma_1 
+\frac{p}{4}(\sigma_1  |\Phi\rangle\langle \Phi| \sigma_1
+ |\Phi\rangle\langle \Phi|)\nonumber\\
=&
\frac{1}{4(1-p)^2}
((1-\frac{3p}{4})\sigma_1 |\Phi\rangle \langle \Phi|\sigma_1 
+\frac{p}{4} |\Phi\rangle\langle \Phi|).\notag
\end{align}
Since 
$\Tr_B \sigma_1 |\Phi\rangle \langle \Phi|\sigma_1
=\Tr_B  |\Phi\rangle \langle \Phi|
=\frac{1}{2}I_A$, we have
\begin{align}
&\Tr_B 
\frac{1}{4(1-p)^2}
((1-\frac{3p}{4})\sigma_1 |\Phi\rangle \langle \Phi|\sigma_1 
+\frac{p}{4} |\Phi\rangle\langle \Phi|) \nonumber\\
=&
\frac{1}{4(1-p)^2}
((1-\frac{3p}{4})+\frac{p}{4})
\frac{1}{2}I_A
=
\frac{2-p}{16(1-p)^2} I_A.\notag
\end{align}
Thus, $ W_{\rm SLD}(T_N,\vec{L}_*)$ belongs to ${\cal K}$.
Theorem \ref{NNT} guarantees the first equation in \eqref{NBR7} for $k=4$.
\end{proof}

\subsubsection{Generalized Pauli channel}\label{Pauli}
We consider the generalized Pauli channel on the system $
\cH=\cH_A=\cH_B$, which is spanned by $\{|a\rangle\}_{a \in \mathbb{Z}_d}$.
We define the operators $\mathsf{W}(a,b)$ for 
$a,b\in\mathbb{Z}_d$ as the following unitary matrices on $\cH$; 
\begin{align}
\mathsf{X}(a) &\coloneqq \sum_{j\in\mathbb{Z}_d} |j+a\rangle \langle j |, \quad 
\mathsf{Z}(b) \coloneqq \sum_{j\in\mathbb{Z}_d} \omega^{bj} |j\rangle \langle j |,\\
\mathsf{W}(a,b) &\coloneqq \mathsf{X}(a)\mathsf{Z}(b),
\end{align}
where $\omega \coloneqq \exp({2\pi i/d})$.
We introduce a distribution family $p_\theta$ over $\mathbb{Z}_d^2$.
Then, we define the family of channels $\{\Lambda_\theta\}$ as
\begin{align}
\Lambda_\theta(\rho):=
\sum_{(a,b) \in \mathbb{Z}_d^2}
p_\theta(a,b) \mathsf{W}(a,b)\rho \mathsf{W}(a,b)^\dagger.
\end{align}

We denote the Fisher information of the distribution family $\{P_\theta\}$.
Then, as shown in \cite{GP1,GP2}, 
we have the following theorem.

\begin{theorem}
\label{thm:Pauli}
We have the following relations
\begin{align}
&J_4[
(\Lambda_{\theta}\otimes \iota)(|I\rangle \langle I|) ,
(\frac{\partial }{\partial \theta^l}\Lambda_{\theta}\otimes \iota)(|I\rangle \langle I|) )_l
] \nonumber \\
=& \bar{S}_4[
(\Lambda_{\theta}\otimes \iota)(|I\rangle \langle I|) ,
(\frac{\partial }{\partial \theta^l}\Lambda_{\theta}\otimes \iota)(|I\rangle \langle I|) )_l,
|\Phi\rangle \langle \Phi|] \nonumber \\
=&
\Tr G J_\theta^{-1}
\Label{NBR8} .
\end{align}
\end{theorem}

\begin{proof}
When the input state is the maximally entangled state 
$|\Phi\rangle \langle \Phi|$, the output state is 
\begin{align}
\sum_{(a,b) \in \mathbb{Z}_d^2}
p_\theta(a,b)  \mathsf{W}(a,b)|\Phi\rangle\langle \Phi| \mathsf{W}(a,b)^\dagger.
\end{align}
Since $\{\mathsf{W}(a,b)|\Phi\rangle\}_{a,b}$ forms  
an orthogonal basis on ${\cal H}^{\otimes 2}$, 
the state family can be considered as distribution family $\{p_\theta\}$.
Then, we obtain the second equation.

Next, we show the first equation.
We denote the logarithmic likelihood derivative for the $j$-th parameter
of $\{p_\theta\}$ by $l_{\theta,j}$. We define the function 
$l_\theta^j$ as $l_\theta^j:= \sum_{j'} (J_\theta^{-1})^{j,j'} l_{\theta,j'}$.
Then, we have
$L_*^j=
\sum_{(a,b) \in \mathbb{Z}_d^2}
l_\theta^j(a,b) 
\mathsf{W}(a,b) |\Phi\rangle\langle \Phi| \mathsf{W}(a,b)^\dagger$.
Hence,
\begin{align}
&\Tr_{B}W_{\rm SLD}(T_N,\vec{L}_*)\notag \\
=&\Tr_{B} \sum_{(a,b) \in \mathbb{Z}_d^2}
\Big(\sum_{1\le j',j\le d} G_{j,j'} 
l_\theta^{j'}(a,b) p_\theta (a,b)l_\theta^j(a,b)\Big) \notag \\
&\cdot \mathsf{W}(a,b) |\Phi\rangle\langle \Phi| \mathsf{W}(a,b)^\dagger \notag \\
=& \sum_{(a,b) \in \mathbb{Z}_d^2}
\sum_{1\le j',j\le d} G_{j,j'} 
l_\theta^{j'}(a,b) p_\theta (a,b)l_\theta^j(a,b)
\rho_{mix},
\end{align}
where $\rho_{mix}$ is the completely mixed state on $\cH_C$
because $\Tr_{B}\mathsf{W}(a,b) |\Phi\rangle\langle \Phi| \mathsf{W}(a,b)^\dagger=\Tr_A |\Phi\rangle\langle \Phi| =\rho_{mix}$.
Hence, 
$\Tr_{B}W_{\rm SLD}(T_N,\vec{L}_*)$ belongs to ${\cal K}$.
Theorem \ref{NNT} guarantees the first equation.
\end{proof}

\subsubsection{Spin $j$ representation of SU(2)}\label{BVR}
Next, we consider spin $j$ representation of SU(2)
over the Hilbert space $\cH_j$.
Here, $\sigma_{1,j}$, $\sigma_{2,j}$, and $\sigma_{3,j}$
are defined as the spin $j$ representations of the generators of SU(2)
on $\cH_j$.
We set $\cH_A$ and $\cH_B$ to be $\cH_j$.
We define the channel $\Lambda_0$ as a depolarizing channel:
\begin{align}
\Lambda_{0,p}(\rho):=(1-p)\rho+p \rho_{{\rm mix},B}.\notag
\end{align}
The channel $\Lambda_\theta$
is given as 
$\Lambda_{\theta,p}(\rho):=U_\theta \Lambda_{0,p}(\rho)
U_\theta^\dagger$,
where 
$U_\theta := \exp (i \sum_{k=1}^3 \theta^k \sigma_k)$.

Then, the following theorem holds.

\begin{theorem}
\label{thm:SU2-maxent}
When the weight matrix $G$ is chosen to be $I$,
we have the following relations
\begin{align}
&J_4[
(\Lambda_{\theta,p}\otimes \iota)(|I\rangle \langle I|) ,
(\frac{\partial }{\partial \theta^l}\Lambda_{\theta,p}\otimes \iota)(|I\rangle \langle I|) )_{l=1,2,3}
] \nonumber \\
=& \bar{S}_4[
(\Lambda_{\theta,p}\otimes \iota)(|I\rangle \langle I|) ,
(\frac{\partial }{\partial \theta^l}\Lambda_{\theta,p}\otimes \iota)(|I\rangle \langle I|) 
)_{l=1,2,3},
|\Phi\rangle \langle \Phi|] \nonumber \\
=&
\frac{9(1-\frac{4j^2+4j-1}{(2j+1)^2} p)}{8j(j+1)(1-p)^2}
\Label{NBR8} .
\end{align}
In particular, for $p=0$, we have
\begin{align}
&J_4[
(\Lambda_{\theta,0}\otimes \iota)(|I\rangle \langle I|) ,
(\frac{\partial }{\partial \theta^l}\Lambda_{\theta,0}\otimes \iota)(|I\rangle \langle I|) )_{l=1,2,3}
] \nonumber \\
=&J_1[
(\Lambda_{\theta,0}\otimes \iota)(|I\rangle \langle I|) ,
(\frac{\partial }{\partial \theta^l}\Lambda_{\theta,0}\otimes \iota)(|I\rangle \langle I|) )_{l=1,2,3}
] \nonumber \\
=& \bar{S}_1[
(\Lambda_{\theta,0}\otimes \iota)(|I\rangle \langle I|) ,
(\frac{\partial }{\partial \theta^l}\Lambda_{\theta,0}\otimes \iota)(|I\rangle \langle I|) 
)_{l=1,2,3},
|\Phi\rangle \langle \Phi|] \nonumber\\
=&
\frac{9}{8j(j+1)}
\Label{NBR10} .
\end{align}
\end{theorem}
The paper \cite{Imai} studied a similar problem in the $n$-copy setting of the estimation of SU($D$).
It maximized the trace of SLD Fisher information matrix by varying the input state.
This maximization 
is a different problem from our problem.

\begin{proof}
First, we show \eqref{NBR8}.
Due to the symmetry, it is sufficient to show the case with $\theta=0$.
Then, $T_N= 
(1-p)|\Phi\rangle \langle \Phi|+ p \rho_{{\rm mix},AB}$
and $F_{l,N}= (1-p) [i \sigma_{l,j} , |\Phi\rangle \langle \Phi|]$, where
$\rho_{{\rm mix},AB}$ is the completely mixed state on $\cH_{AB}$.
Since 
\begin{align}
[\sigma_{l,j} , |\Phi\rangle \langle \Phi|] \circ |\Phi\rangle \langle \Phi|
&=\frac{1}{2}[\sigma_{l,j} , |\Phi\rangle \langle \Phi|] \notag\\
[\sigma_{l,j} , |\Phi\rangle \langle \Phi|] \circ \rho_{{\rm mix},AB}
&=\frac{1}{(2j+1)^2}[\sigma_{l,j} , |\Phi\rangle \langle \Phi|] ,\notag
\end{align}
using $ \frac{1}{2}(1-p)+\frac{1}{(2j+1)^2}p=\frac{1}{2}-\frac{4j^2+4j-1}{2(2j+1)^2} p$,
we have $L_l= c_{p,j}
[i \sigma_{l,j} , |\Phi\rangle \langle \Phi|]$ with 
$c_{p,j}:=\frac{1-p}{\frac{1}{2}-\frac{4j^2+4j-1}{2(2j+1)^2} p}$.
Since 
$ \langle \Phi| \sigma_{l',j} \sigma_{l,j} |\Phi\rangle=\frac{2j(j+1)}{3}\delta_{l,l'}$
and
$ \langle \Phi| \sigma_{l',j} |\Phi\rangle=0$,
the SLD Fisher information matrix in this model is 
\begin{align}
&\Tr L_{l}L_{l'}  T_N=
\Tr 
c_{p,j}^2
[i \sigma_{l,j} , |\Phi\rangle \langle \Phi|] [i \sigma_{l',j} , |\Phi\rangle \langle \Phi|]
 T_N \notag\\
=&
c_{p,j}^2 \Tr
(\sigma_{l,j}  |\Phi\rangle\langle \Phi| \sigma_{l',j}
+ \frac{2j(j+1)}{3}\delta_{l,l'}|\Phi\rangle\langle \Phi| ) T_N \notag\\
=&
c_{p,j}^2 \Tr
\Big((1-p)\frac{2j(j+1)}{3}\delta_{l,l'}|\Phi\rangle\langle \Phi|\notag \\
&+\frac{p}{(2j+1)^2} 
(\sigma_{l,j}  |\Phi\rangle\langle \Phi| \sigma_{l',j} 
+\frac{2j(j+1)}{3}\delta_{l,l'} |\Phi\rangle\langle \Phi| ) \Big)\notag\\
=&\tilde{c}_{p,j}
\delta_{l,l'},\notag
\end{align}
where
$\tilde{c}_{p,j}:=c_{p,j}^2 \frac{2j(j+1)}{3}((1-p)+\frac{2p}{(2j+1)^2} )
$.
Thus, the SLD Fisher information matrix is
$\tilde{c}_{p,j}I$.
Hence, the SLD bound $\bar{S}_4$ is 
$\tr \tilde{c}_{p,j}^{-1}I=3/\tilde{c}_{p,j}$.
Since 
we have
\begin{align}
&\tilde{c}_{p,j}
=
(\frac{1-p}{\frac{1}{2}-\frac{4j^2+4j-1}{2(2j+1)^2} p})^2
 \frac{2j(j+1)}{3}
(1-\frac{4j^2+4j-1}{(2j+1)^2} p) \nonumber \\
=&
\frac{4(1-p)^2}{1-\frac{4j^2+4j-1}{(2j+1)^2} p}
 \frac{2j(j+1)}{3}
=
\frac{8j(j+1)(1-p)^2}{3(1-\frac{4j^2+4j-1}{(2j+1)^2} p)}.\notag
\end{align}
Hence, we obtain the second equation in \eqref{NBR8}.

Then, 
$L^l_*
=\frac{c_{p,j}}{\tilde{c}_{p,j}}
[i \sigma_{l,j} , |\Phi\rangle \langle \Phi|]$.
Since
\begin{align}
&[i\sigma_{l,j} , |\Phi\rangle \langle \Phi|]
|\Phi\rangle \langle \Phi|[ i\sigma_{l',j} , |\Phi\rangle \langle \Phi|]
= \sigma_{l,j} |\Phi\rangle \langle \Phi|\sigma_{l',j} , \\
&[i\sigma_{l,j} , |\Phi\rangle \langle \Phi|] 
\rho_{{\rm mix},AB}[ i\sigma_{l',j} , |\Phi\rangle \langle \Phi|]\nonumber \\
=& \frac{1}{(2j+1)^2}(
\sigma_{l,j}  |\Phi\rangle\langle \Phi| \sigma_{l',j}
+\frac{2j(j+1)}{3}\delta_{l,l'} |\Phi\rangle\langle \Phi|),\notag
\end{align}
we have
\begin{align}
&W_{\rm SLD}(T_N,\vec{L}_*)= \sum_{l=1}^3L^l_* T_N L^l_*\nonumber \\
=&
(\frac{c_{p,j}}{\tilde{c}_{p,j}})^2\sum_{l=1}^3
\Big((1-p)\sigma_{l,j} |\Phi\rangle \langle \Phi|\sigma_{l,j} 
\nonumber \\
&+\frac{p}{(2j+1)^2}(\sigma_{l,j}  |\Phi\rangle\langle \Phi| \sigma_{l,j}
+\frac{2j(j+1)}{3} |\Phi\rangle\langle \Phi|\Big)\nonumber \\
=&
(\frac{c_{p,j}}{\tilde{c}_{p,j}})^2
\Big((1-p+\frac{p}{(2j+1)^2})
(\sum_{l=1}^3 \sigma_{l,j}  |\Phi\rangle\langle \Phi| \sigma_{l,j})
\nonumber \\
&+  \frac{3p}{(2j+1)^2} \frac{2j(j+1)}{3}|\Phi\rangle\langle \Phi|\Big).\notag
\end{align}
Since $\sum_{l=1}^3 \sigma_{l,j}^2$ is a Casimir operator, it is 
$2j(j+1)I $ on $\cH_j$. Thus,
\begin{align}
&\Tr_B 
\sum_{l=1}^3 \sigma_{l,j} |\Phi\rangle \langle \Phi|\sigma_{l,j}
=
\Tr_B 
\sum_{l=1}^3 \sigma_{l,j}^2 |\Phi\rangle \langle \Phi| \nonumber \\
=&
\Tr_B 2j(j+1)
|\Phi\rangle \langle \Phi|= \frac{2j(j+1)}{2j+1}I_A.\notag
\end{align}
Since 
$\Tr_B |\Phi\rangle \langle \Phi|
=\frac{1}{2j+1}I_A$, we have
\begin{align}
&\Tr_B 
W_{\rm SLD}(T_N,\vec{L}_*) \nonumber \\
=&
(\frac{c_{p,j}}{\tilde{c}_{p,j}})^2
\Big((1-p+\frac{p}{(2j+1)^2})
  \frac{2j(j+1)}{2j+1}
\nonumber \\
&+ \frac{3p}{(2j+1)^2} \frac{2j(j+1)}{3} \frac{1}{2j+1}\Big) I_A \nonumber \\
=&
(\frac{c_{p,j}}{\tilde{c}_{p,j}})^2
\Big((1-p+\frac{2p}{(2j+1)^2})
  \frac{2j(j+1)}{2j+1}\Big) I_A .\notag
\end{align}
Thus, $ W_{\rm SLD}(T_N,\vec{L}_*)$ belongs to ${\cal K}$.
Theorem \ref{NNT} guarantees the first equation of \eqref{NBR8}.

Next, with $p=0$,
we show 
\begin{align}
& \bar{S}_1[
(\Lambda_{\theta,0}\otimes \iota)(|I\rangle \langle I|) ,
(\frac{\partial }{\partial \theta^l}\Lambda_{\theta,0}\otimes \iota)(|I\rangle \langle I|) 
)_{l=1,2,3},
|\Phi\rangle \langle \Phi|] \nonumber  \\
=& \bar{S}_4[
(\Lambda_{\theta,0}\otimes \iota)(|I\rangle \langle I|) ,
(\frac{\partial }{\partial \theta^l}\Lambda_{\theta,0}\otimes \iota)(|I\rangle \langle I|) 
)_{l=1,2,3},
|\Phi\rangle \langle \Phi|] \Label{NBR11} .
\end{align}
Then, $T_N= 
|\Phi\rangle \langle \Phi|$
and $F_{l,N}= 
(i \sigma_{l,j}) |\Phi\rangle \langle \Phi|
+ |\Phi\rangle \langle \Phi|(i \sigma_{l,j})^\dagger$.
Since
$\langle \Phi|(i \sigma_{l,j})^\dagger
(i \sigma_{l',j}) |\Phi\rangle=
\frac{2j(j+1)}{3}\delta_{l,l'} $,
the vectors 
$((i \sigma_{l,j}) |\Phi\rangle)_{l=1,2,3}$
are orthogonal. 
Thus, Theorem 3 of \cite{matsumoto2002new} guarantees that
the SLD bound is attainable, i.e., \eqref{NBR11}.
Since $\tilde{S}_1 \ge J_1 \ge J_4$, 
the combination of \eqref{NBR8} and \eqref{NBR11}
implies the first and second equations in
\eqref{NBR10}.
The third equation in \eqref{NBR10} follows from \eqref{NBR8}. 
\end{proof}

 \section{Application to field sensing}
 \label{sec:field sensing}

\subsection{Model for field sensing}
The canonical example that is widely considered in quantum metrology is `field sensing', where a classical field interacts with an ensemble of qubits. When a 3D classical field interacts identically with $n$ qubits, we can write the interaction Hamiltonian as 
\begin{align}
H_\theta = \theta_1 E^1  + \theta_2 E^2 + \theta_3 E^3\notag
\end{align}
where $\theta = (\theta_1, \theta_2, \theta_3)$ is a real vector that is proportional to the 3D field that we wish to estimate, and $E^1, E^2, E^3$ are angular momentum operators defined on $n$-qubits,
given by
\begin{align}
E^1 = \frac{1}{2} ( \tau_{1}^{(1)} + \dots + \tau_{1}^{(n)} ),\notag\\
E^2 = \frac{1}{2} ( \tau_{2}^{(1)} + \dots + \tau_{2}^{(n)}  ),\notag\\
E^3 = \frac{1}{2} ( \tau_{3}^{(1)} + \dots + \tau_{3}^{(n)} ),\notag
\end{align}
$\tau_1 = |0\>\<1|+|1\>\<0|,\tau_3 = |0\>\<0| -|1\>\<1|$ are Pauli matrices that apply the bit-flip and phase-flip on a qubit, 
$\tau_2  = i \tau_1 \tau_3$, and $\tau_{j}^{(k)}$ represents an $n$-qubit Pauli matrix that applies $\tau_{j}$ on the $k$th qubit, and identity operations everywhere else.

Amplitude damping operators $A_j$, which model energy loss, apply $\gamma |0\>\<1|$ on the $j$th qubit and the identity operator on other qubits. 
These operators arise because of a linear interaction between individual qubits and a Markovian zero-temperature bath\cite{CLY97}.  
Namely, 
\begin{align}
A_{j,\gamma} &=  \gamma I^{\otimes j-1} \otimes |0\>\<1|  \otimes I^{\otimes (n-j)}.\notag
\end{align}
The collective amplitude damping operator, also considered in \cite[Eq. (7)]{PhysRevA.58.3491}, models collective energy loss, specifically because of a collective linear interaction between all $n$ qubits and a Markovian zero-temperature bath, and is given by 
$\tilde{A}_\gamma:=\sum_{j=1}^nA_{j,\gamma}$.

We model the evolution of an initial probe state using the master equation
\begin{align}
\frac{d\rho}{dt} = \mathcal L_{\theta,\gamma} (\rho), 
\label{master equation}
\end{align}
where $t$ denotes time, and the operator $\mathcal L_\theta $ can be written as a linear operator, which is 
\begin{align}
 \mathcal L_{\theta,\gamma}(\rho) 
= 
-i (H_\theta \rho - \rho H_\theta)
+ 
\tilde A_\gamma \rho \tilde A_\gamma^\dagger 
- \frac{1}{2} (\tilde A_\gamma^\dagger \tilde A_\gamma \rho 
+ \rho \tilde A_\gamma^\dagger \tilde A_\gamma ).\label{BNJ}
\end{align}

For a non-negative evolution time $s$, let $\tilde{\rho}_s$ denote the solution to the master equation \eqref{master equation}. In particular, we can write 
\begin{align}
\tilde{\rho}_s  = e^{s\mathcal L_{\theta,\gamma} }(\tilde{\rho}_0).\notag
\end{align}
This means that we can write $\tilde{\rho}_s$ as the Taylor series
\begin{align}
\tilde{\rho}_s  =\tilde{\rho}_0 + \sum_{k=1}^\infty  \frac{ s^k (\mathcal L_{\theta,\gamma})^k (\tilde{\rho}_0 )}{k!}.\notag
\end{align}
In our application, we set the evolution time as $s=1.$
Hence, the channel that we consider in our channel estimation problem is 
\begin{align}
    \Lambda_{\theta,\gamma}(\rho)
    = e^{\mathcal L_{\theta,\gamma}}(\rho) = 
    \rho + \sum_{k=1}^\infty  \frac{\mathcal L_{\theta,\gamma} ^k (\rho )}{k!}.\Label{NMR5}
\end{align}
Using $\Lambda_{\theta,\gamma}$, 
we can calculate the corresponding Choi matrix 
$ T_{\theta,\gamma}$, 
and its derivatives about the true paramater $\theta_0$ are
\begin{align}
    F_{j,\gamma} &= \frac{\partial}{\partial \theta^j} T_{\theta,\gamma} |_{\theta = \theta_0}\notag
\end{align}
for $j=1,2,3$.


\subsection{Numerical results} 
 \begin{figure*}[htbp]  
    \centering
    \includegraphics[width=0.98\textwidth]{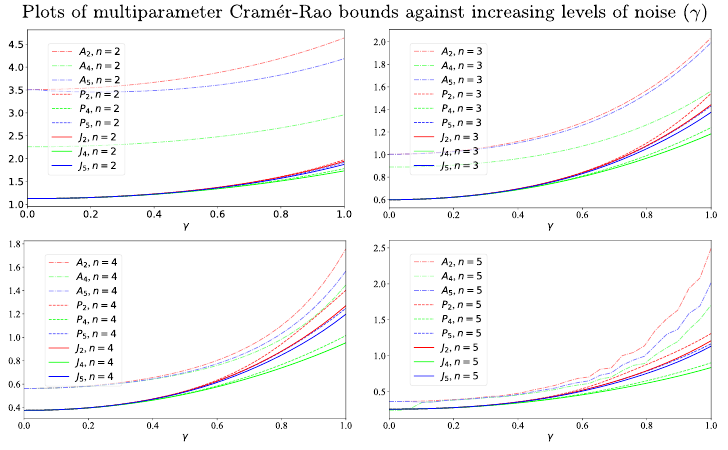}    
    \caption{We numerically calculate the CR-bounds for field sensing when $\theta_1 = \theta_2 = \theta_3 = 0$ for fixed number of qubits $n=2,3,4,5$.
 Here, $J_k \coloneqq J_{k,\gamma}  $ denote the CR-bounds using the optimal probe state with ancilla assistance. 
 In contrast, $P_k \coloneqq S_{k,\gamma}^{\rm sym}$ denote the CR-bounds using the maximally entangled probe state on the symmetric subspace.
We also use $A_k \coloneqq  S_{k,\gamma}^{3D}$ for the CR-type bounds using the 3D-GHZ state that do not require ancilla resistance.
 The vertical axis represents the CR-type lower bound on the sum of the variances of the field parameters $\theta_1, \theta_2 $ and $\theta_3$, which corresponds to a weight matrix $G$ equal to the identity matrix.
 The horizontal axis denotes the amount of amplitude damping noise $\gamma$.   
     \label{fig:fixn}  }  
\end{figure*}

\begin{figure*}[htbp]  
    \centering 
        \includegraphics[width=\textwidth]{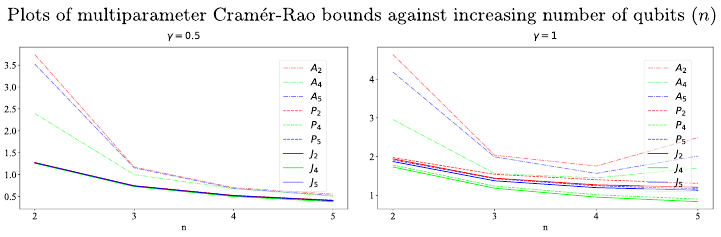}    
    \caption{We represent the data in Figure \ref{fig:fixn} differently, plotting the CR-bounds for fixed values of $\gamma$.
 The horizontal axis denotes the number of qubits $n$ in the field sensing problem.
    }  
      \label{fig:APJgam}
\end{figure*}

Here, we set the true parameter as $\theta_0 = (0,0,0)$ and the number of qubits as $n=2,3,4,5$. We investigate our channel estimation problem when $\gamma$ varies from 0 to 1.
In our calculations, we set the weight matrix $G$ as the identity matrix $I$.

We numerically evaluate the channel estimation precision bounds 
$J_{2,\gamma},J_{4,\gamma},J_{5,\gamma}$, where
\begin{align}
J_{k,\gamma} &= J_k[T_{\theta_0,\gamma}, 
(F_{1,\gamma},F_{2,\gamma},F_{3,\gamma})] .\notag
\end{align}
The precision bounds $J_{2,\gamma},J_{4,\gamma},J_{5,\gamma}$ are given by the optimal values of semidefinite programs with corresponding optimal solutions given by 
$Y_{2,\gamma},Y_{4,\gamma},Y_{5,\gamma}$,
respectively.
Using the optimal solutions 
$Y_{2,\gamma},Y_{4,\gamma},Y_{5,\gamma}$,
we calculate the corresponding probe states on $\mathcal H_A$ given by 
$\rho_A(Y_{2,\gamma}) ,
\rho_A(Y_{4,\gamma}) ,
\rho_A(Y_{5,\gamma}) $ respectively.

Let us denote $|\Phi\>$ as the maximally entangled state on the symmetric subspace,
and consider the precision bounds that correspond to using $|\Phi\>$ as the input probe state for the channel estimation problem. 
These precision bounds are 
\begin{align}
 S_{k,\gamma}^{\rm sym} 
 : =& S_k \Big[  
 \Lambda_{\theta_0,\gamma} \otimes \iota_C(  |\Phi\>\<\Phi|),
\nonumber \\   
& \quad \Big(
 \frac{\partial }{\partial \theta^j} \Lambda_{\theta,\gamma}  \otimes \iota_C(  |\Phi\>\<\Phi|)\Big|_{\theta = \theta_0}  
 \Big)_{j=1,2,3}\Big]   .\notag
\end{align}

We also numerically evaluate $S^{\rm sym}_{2,\gamma},
S^{\rm sym}_{4,\gamma},
S^{\rm sym}_{5,\gamma}$.
When we have no access to an ancillary system $\mathcal{H}_C$, we may consider using the 3D-GHZ state \cite{BaumgratzDatta2016PRL}
\begin{align}
|\psi_{\text{3D GHZ}}\> = \frac{\sum_{j=0}^1 |j\>^{\otimes n}   + |+\>^{\otimes n} + |-\>^{\otimes n} + |+i\>^{\otimes n} + |-i\>^{\otimes n}}{N},\notag
\end{align}
where $|\pm\> = \frac{|0\> \pm |1\>}{\sqrt 2}$ and $|\pm i\>= \frac{|0\> +\pm i  |1\>}{\sqrt 2}$, and $N$ is the appropriate normalization factor.
We denote the corresponding precision bounds for using the 3D-GHZ state without ancillary assistance as
\begin{align}
S_{k,\gamma}^{3D} = S_k[   \Lambda_{\theta_0}(\rho_{\rm 3D GHZ} ),
\frac{\partial }{\partial \theta^j}   \Lambda_{\theta,\gamma}(\rho_{\rm 3D GHZ} )|_{\theta = \theta_0}
 ] .\notag
\end{align}
We also numerically calculate $S_{2,\gamma}^{3D}, 
S_{4,\gamma}^{3D}, S_{5,\gamma}^{3D}$.

We numerically find that when $\gamma = 0$, we have 
\begin{align}
    J_{2,0}=J_{4,0}=J_{5,0} = S^{\rm sym}_{2,0} = S^{\rm sym}_{4,0} 
    = S^{\rm sym}_{5,0} 
    = \frac{9}{n(n+2)}.\notag
\end{align}
In fact, the analysis in Subsection \ref{BVR} showed that
the same equality when the input state is limited to a state on 
the symmetric subspace.
Our numerical analysis suggests that the support of 
the optimal input state is limited to the symmetric subspace.

Furthermore, when we solve the semidefinite programs corresponding to 
$J_{2,\gamma},J_{4,\gamma},J_{5,\gamma}$,
we find that the corresponding optimal solutions 
$Y_{2,\gamma},Y_{4,\gamma},Y_{5,\gamma}$
have corresponding 
density matrices 
$\rho_A(Y_{2,\gamma}),\rho_A(Y_{4,\gamma}),
\rho_A(Y_{5,\gamma})$ 
that are very close to the completely mixed state.

When $\gamma > 0$ we numerically find that 
\begin{align}
    S^{3D}_2 > S^{\rm sym}_2  > J_2, \quad 
    S^{3D}_4 > S^{\rm sym}_4 > J_4, \quad 
    S^{3D}_5 > S^{\rm sym}_5 > J_5. \label{PJ-gap}
\end{align}

In Figure \ref{fig:fixn}, for different fixed values of $n=2,3,4,5$, we plot the precision bounds $S^{3D}_{k,\gamma}, S^{\rm sym}_{k,\gamma}$ and 
$J_{k,\gamma}$ on the vertical axis and $\gamma \in [0,1]$ on the horizontal axies.
In Figure \ref{fig:APJgam}, 
for fixed $\gamma  = 0.5$ or $\gamma = 1$,
we plot the precision bounds $A_k,P_k$ and $J_k$ on the vertical axis, and the number of qubits $n$ on the horizontal axis.

In this way, we confirm the suboptimality of the maximally entangled state by our numerical evaluation of $\rho_A(Y)$ using the optimal solution $Y$ for the conic programs 
$J_{2,\gamma}, J_{4,\gamma}$ and $J_{5,\gamma}$ for $\gamma$.
In particular, while $\rho_A(Y_{2,\gamma})$, 
$\rho_A(Y_{4,\gamma})$, and $\rho_A(Y_{5,\gamma})$ 
are still a diagonal matrices when $\gamma > 0$, they are not completely mixed states. 
Hence, their purifications cannot be the maximally entangled state.
We can furthermore see the suboptimality of the maximally entangled state because $S^{\rm sym}_{k,\gamma} \ge J_{k,\gamma}$ 
for $k=2,4,5$ when $\gamma > 0$.
Curiously, from Figure \ref{fig:fixn},
the maximally entangled state 
on the symmetric subspace
is nonetheless still quite close to optimal. 

Motivating the above analysis, we prove the following theorem. 
When the support of the input state is limited to the symmetric subspace
in the channel $ \Lambda_{\theta}$,
we denote the obtained channel by $\Lambda_{\theta,\gamma}^{\rm sym}$.
Using this limitation for the inputs, we define 
$T_{\theta,\gamma}^{\rm sym}$ and $F_{j,\gamma}^{\rm sym}$ in the same way.

\begin{theorem}\label{THNM}
For $\gamma >0$, we have 
\begin{align}
J_4[T_{\theta_0,\gamma}^{\rm sym}, &
(F_{j,\gamma}^{\rm sym})_{j=1,2,3}] \nonumber \\
<& S^{\rm sym}_{k,\gamma} \nonumber \\
=& S_4 \Big[  
 \Lambda_{\theta_0,\gamma}^{\rm sym} \otimes \iota_C(  |\Phi\>\<\Phi|),
\nonumber \\   
& \quad \Big(
 \frac{\partial }{\partial \theta^j} \Lambda_{\theta,\gamma} ^{\rm sym} \otimes \iota_C(  |\Phi\>\<\Phi|)\Big|_{\theta = \theta_0}  
 \Big)_{j=1,2,3}\Big]   .\notag
\end{align}
\end{theorem}
Since the relation 
$J_4[T_{\theta_0,\gamma}, (F_{j,\gamma})_{j=1,2,3}] 
\le J_4[T_{\theta_0,\gamma}^{\rm sym}, 
(F_{j,\gamma}^{\rm sym})_{j=1,2,3}] 
$ holds,
the above theorem means that 
the maximally entangled state on the symmetric subspace
is not the optimal input state.

\subsection{Methodology of numerical analysis}
Now we use $\rho' = \Lambda_{\theta,\gamma}^{\approx}(\rho)$ to approximate $\Lambda_{\theta,\gamma}(\rho)$ according to the following procedure:
Given any input state $\rho$, we define 
\begin{align}
\bar m_*(\rho ) = \min( 100, \min \{j : \| ({ \mathcal L}_{\theta,\gamma}^j    (\rho)/j!\| 
< 10^{-12}\}).\notag
\end{align}
Then we define 
\begin{align}
  \Lambda_{\theta,\gamma}^{\approx}
(\rho)
:= 
    \sum_{j=0}^{ m_*(\rho)}  \frac{ { \mathcal L}_{\theta,\gamma} ^j    (\rho )}{j!}.\notag
\end{align}
On the state $\rho_{AC}$, we also define 
\begin{align}
  \Lambda_{\theta,\gamma}^{\approx}
(\rho_{AC})
:= 
    \sum_{j=0}^{m_*(\rho_{AC})}  
    \frac{(\mathcal L_{\theta,\gamma} ^j\otimes \iota_C) (\rho_{AC} )}{j!}.\notag
\end{align}

We obtain approximations of $T_{\theta,\gamma}$ with
\begin{align}
  T_{\theta,\gamma}^{\approx}
  =  \Lambda_{\theta,\gamma}^{\approx}( (n+1) |\Phi\>\<\Phi| ).   \notag
\end{align}
We also obtain approximations of $F_1, F_2,F_3$ according to the formula
\begin{align}
F_{1,\gamma}^\approx &= 
\frac{T^{\approx}_{ (10^{-12},0,0),\gamma}
- T^{\approx}_{(0,0,0) ,\gamma}}{10^{-12}},\notag \\    
F_{2,\gamma}^\approx &= 
\frac{T^{\approx}_{ (0, 10^{-12},0) ,\gamma}
- T^{\approx}_{(0,0,0) ,\gamma}}{10^{-12}}, \notag\\    
F_{3,\gamma}^\approx &= 
\frac{T^{\approx}_{ (0,0,10^{-12}),\gamma}
- T^{\approx}_{(0,0,0) ,\gamma}}{10^{-12}}. \notag  
\end{align}
With the above, we approximate $J_{2,\gamma}, J_{4,\gamma},
J_{5,\gamma}$ with 
\begin{align}
J_{k,\gamma}^\approx
    = J_k[T^{\approx}_{(0,0,0),\gamma}, 
    (F_{j,\gamma}^\approx)_{j=1,2,3} ] .\notag
\end{align}
We also approximate
$P_2, P_4, P_5$ with 
\begin{align}
S_{k,\gamma}^{{\rm sym},\approx}
    = S_{k,\gamma}
    [T^{\approx}_{(0,0,0)/(n+1),\gamma}, 
    (F_{j,\gamma}^\approx/(n+1))_{j=1,2,3} ] .\notag
\end{align}

We approximate 
$\Lambda_{(0,0,0)}(\rho_{\rm 3DGHZ})$ according to formula
\begin{align}
Q_0 = \Lambda^\approx_{(0,0,0),\gamma}(\rho_{\rm 3DGHZ})\notag
\end{align}
and approximate $\frac{\partial}{\partial \theta^j}\Lambda_{\theta}(\rho_{\rm 3DGHZ})|_{\theta = (0,0,0),\gamma}$ according to 
\begin{align}
Q_1^\approx &= \frac{\Lambda^\approx_{(10^{-12},0,0),\gamma}(\rho_{\rm 3DGHZ})
 -
 \Lambda^\approx_{(0,0,0),\gamma}
 (\rho_{\rm 3DGHZ})}{10^{-12}}, \notag\\
 Q_2^\approx &= \frac{\Lambda^\approx_{(0,10^{-12},0),\gamma}(\rho_{\rm 3DGHZ})
 -
 \Lambda^\approx_{(0,0,0),\gamma}
 (\rho_{\rm 3DGHZ})}{10^{-12}}, \notag\\
 Q_3^\approx &= \frac{\Lambda^\approx_{(0,0,10^{-12}),\gamma}
 (\rho_{\rm 3DGHZ})
 -
 \Lambda^\approx_{(0,0,0),\gamma}(\rho_{\rm 3DGHZ})}{10^{-12}}.\notag
\end{align}
This allows us to approximate 
$S^{3D}_{2,\gamma},S^{3D}_{4,\gamma},S^{3D}_{5,\gamma}$ according to 
\begin{align}
    S_{k,\gamma}^{3D,\approx} 
    = S_k[ Q_0, (Q_1,Q_2,Q_3)].\notag
\end{align}   
In our numerical evaluations of 
$S_{k,\gamma}^{3D,\approx}, S_{k,\gamma}^{{\rm sym},\approx}, J_{k,\gamma}^\approx$,
we evaluate the semidefinite programs according to MatLab code given in Appendix \ref{sec:SDPs}.

\subsection{Proof of Theorem \ref{THNM}}
Next, we consider $\theta= (0,0,0)$, $\gamma > 0$, and where we fix the input probe state as the maximally entangled state $|\Phi\>\<\Phi|$.
In this case, after the channel acts on our input probe state, 
we have the output state 
$\rho_\gamma  :=(\Lambda_{0,\gamma} \otimes \iota_C )(|\Phi\>\<\Phi|)$.
Letting $F^u := E^u \otimes I_C$,
we define
\begin{align}
L_1 &:= L_1[\rho_\gamma , ( i [\rho_\gamma, F^1], i [\rho_\gamma, F^2],i [\rho_\gamma, F^3]) ]\notag\\
L_2 &:= L_2[\rho_\gamma , ( i [\rho_\gamma, F^1], i [\rho_\gamma, F^2],i [\rho_\gamma, F^3]) ].\notag
\end{align}
Then, we have the following lemma.
\begin{lemma}
\label{lem:trL1rhoL2-trL2rhoL1}
    \begin{align}
&\tr[L_1 \rho_\gamma L_2]
-
\tr[L_2 \rho_\gamma L_1]  \notag\\
&=
\frac{- 2 i n(n+2)}{3} (1 + n(n+2)/3)\gamma^2 
 + O(\gamma^4).\label{commute 12}
\end{align}
\end{lemma} 
Since $\tr[L_1 \rho_\gamma L_2]
-
\tr[L_2 \rho_\gamma L_1]$ is not equal to zero, this proves that the SLD bound $S^{\rm sym}_{k,\gamma} $ is not a tight bound.
That is, we obtain Theorem \ref{THNM}.

In the following, we show Lemma \ref{lem:trL1rhoL2-trL2rhoL1}. 
For this aim, we handle the symmetric subspace.
The symmetric subspace is described by using the Dicke basis $\{|D^n_0\> , \dots, |D^n_n\>\}$ where the Dicke states
\begin{align}
|D^{n}_{w}\> :=
\sum_{\substack{
(x_1,\dots, x_n) \in \{0,1\}^n \\
x_1 + \dots +  x_n = w
}} \frac{1}{\sqrt{\binom n w}} (|x_1\> \otimes \dots \otimes |x_n\> ) \notag
\end{align}
are uniform superpositions of all computational states in $\{|0\>,|1\>\}^{\otimes n}$ of constant Hamming weight $w$ with unit norm. 
Then, 
the maximally entangled state $|\Phi\>$ on the symmetric subspace has the form
$|\Phi\>= \frac{1}{\sqrt{n+1}}\sum_{u=0}^n |D^n_u\>\otimes |D^n_u\>$, and
we introduce two orthogonal vectors 
\begin{align}
 |\phi\> := &
( (I_A  - \frac{1}{2}  {\tilde{A}_\gamma} ^\dagger{\tilde{A}_\gamma})\otimes I_C) |\Phi\>\notag\\
=&    \frac{1}{\sqrt{n+1}}\sum_{u=0}^n  (I_A  - \frac{1}{2}  {\tilde{A}_\gamma} ^\dagger{\tilde{A}_\gamma}) |D^n_u\>\otimes |D^n_u\> ,\Label{BYU1}\\
 |\epsilon\> :=& 
( {\tilde{A}_\gamma}\otimes I_C)   |\Phi\>
= 
   \frac{1}{\sqrt{n+1}}\sum_{u=0}^n  {\tilde{A}_\gamma}  |D^n_u\>\otimes |D^n_u\>.
\Label{BYU2}
\end{align}
The vectors $ |\phi\>$ and $|\epsilon\>$ are orthogonal because of the following reason.
First, according to \eqref{tildeAgam sym}, when $u=1,\dots, n$, 
${\tilde{A}_\gamma}|D^n_u\> = \gamma \sqrt{(n-u+1)u} |D^n_{u-1}\>$, and when $u =0$, ${\tilde{A}_\gamma}|D^n_u\>=0,$
and hence ${\tilde{A}_\gamma}|D^n_u\> = a_u|D^n_{u-1}\>$ for some real number $a_u$.
Second, according to \eqref{tildeAgam sym}
$ (I_A  - \frac{1}{2}  {\tilde{A}_\gamma} ^\dagger{\tilde{A}_\gamma}) |D^n_u\> =  
b_u |D^n_u\>$ for some real number $b_u$.
Hence $\<\phi|\epsilon\> = \frac{1}{n+1}\sum_{u=1}^n a_u b_u \<D^n_{u-1}|D^n_u\>$. Since $|D^n_u\>$ and $|D^n_{u-1}\>$ are pairwise orthogonal for $u=1,\dots,n$, it follows that $\<\phi|\epsilon\>=0$, and hence $ |\phi\>$ and $|\epsilon\>$ are orthogonal.

It is easy to see that
\begin{align}
\lambda_1 :=&\<\phi|\phi\> 
=
\frac{1}{n+1} \sum_{w=0}^n (  1- \gamma^2 (n-w+1)w/2 )^2\label{pp}
\notag\\
=&1- \frac{\gamma^2}{6}n(n+2) + O(\gamma^4),\\
\lambda_2 :=&\<\epsilon|\epsilon\> 
= \frac{\gamma^2}{n+1} \sum_{w=0}^{n-1} (n-w)(w+1)
= \frac{\gamma^2}{6}n(n+2).\label{ee}
\end{align} 

Using 
$|\phi_1\> := |\phi\>/\sqrt{\lambda_1}$ and
$|\phi_2\> := |\epsilon\>/\sqrt{\lambda_2}$, 
we define
\begin{align}
 \bar \rho_\gamma := |\phi\>\<\phi| + |\epsilon\>\<\epsilon|  
 = \lambda_1 |\phi_1\>\<\phi_1| + \lambda_2|\phi_2\>\<\phi_2| . \Label{BYU3}
\end{align} 
Hence, 
the two largest eigenvalues of  $\bar\rho_\gamma$ are $\lambda_1$ and 
$\lambda_2$,
which are given by $1- \frac{\gamma^2}{6}n(n+2)+ O(\gamma^4)$ 
and $ \frac{\gamma^2}{6}n(n+2)$ respectively.
Using these, we define two quantities
\begin{align}
f(1,2) :=&
\frac
{ \lambda_1 -\lambda_2}
{ \lambda_1 + \lambda_2}
= 1- \frac{\gamma^2}{3}n(n+2) + O(\gamma^4),\\
g: =& f(1,2)^2-1.\notag
\end{align}
Then, we have
\begin{align}
g =
- \frac{2 \gamma^2 n(n+2)}{3} + O(\gamma^4)
 \label{gval}.
\end{align}

Introducing two operators
\begin{align}
\bar L_1 &:= L_1[\bar \rho_\gamma , ( i [\bar \rho_\gamma, F^1], i [\bar \rho_\gamma, F^2],i [\bar \rho_\gamma, F^3]) ]\notag\\
\bar L_2 &:= L_2[\bar \rho_\gamma , ( i [\bar \rho_\gamma, F^1], i [\bar \rho_\gamma, F^2],i [\bar \rho_\gamma, F^3]) ],\notag
\end{align}
we have the following two lemmas.
\begin{lemma}\label{L12}
Let $u,v = 1,2,3$.
We have
\begin{align}
 \tr (\bar L_u \bar \rho_\gamma \bar L_v ) 
=&
- 4\sum_{l=1}^2
\<\phi_{l}|{F^u}  |\phi_{l}\>
\<\phi_l|{F^v}|\phi_l\>  \notag\\
&+
4 g \lambda_2
\<\phi_{1}|{F^u}  |\phi_{2}\>
\<\phi_2|{F^v}|\phi_1\>  \notag\\
&+
4 g \lambda_1
\<\phi_{2}|{F^u}  |\phi_{1}\>
\<\phi_1|{F^v}|\phi_2\>  \notag\\
&+
4 \sum_{l=1}^2  \lambda_l \<\phi_{l}|{F^v} {F^u}|\phi_l\>.
\label{full-expression}
\end{align} 
\end{lemma}

\begin{lemma}
Let $\mu := n(n+2)/12$. Then
 \label{lem:E1E2-cross-products}
\begin{align}
\<\phi|F^1 |\epsilon\> &=  \mu \gamma +  O(\gamma^3),
 \label{p1e}
 \\
\<\phi|F^2 |\epsilon\> &=  i\mu \gamma +  O(\gamma^3), 
 \label{p2e}
 \\
   \<\phi| F^3 |\phi\>   &= \mu \gamma^2 ,\label{p3p}\\
      \<\epsilon| F^3 |\epsilon\>   &= \mu  \gamma^2   + O(\gamma^4) \label{e3e}.
\end{align}
\end{lemma}
Lemma \ref{L12} is shown in Appendix \ref{app:SLD calculations},
and Lemma \ref{lem:E1E2-cross-products} is shown in Appendix \ref{app:prove lemma 12}.

Using \eqref{full-expression} with $(u,v)=(1,2)$, 
we find that  
\begin{align}
&\tr[\bar L_1\bar \rho_\gamma \bar L_2]
-\tr[\bar L_2\bar \rho_\gamma \bar L_1]\notag\\
\stackrel{(a)}{=}&
4 g \lambda_2
\<\phi_{1}|{F^1}  |\phi_{2}\>
\<\phi_2|{F^2}|\phi_1\>   +
4 g \lambda_1
\<\phi_{2}|{F^1}  |\phi_{1}\>
\<\phi_1|{F^2}|\phi_2\>  \notag\\
&-4 g \lambda_2
\<\phi_{1}|{F^2}  |\phi_{2}\>
\<\phi_2|{F^1}|\phi_1\>  
-4 g \lambda_1
\<\phi_{2}|{F^2}  |\phi_{1}\>
\<\phi_1|{F^1}|\phi_2\>  \notag\\
&+
4 \sum_{l=1}^2  \lambda_l \<\phi_{l}|{F^2} {F^1}|\phi_l\>
-
4 \sum_{l=1}^2  \lambda_l \<\phi_{l}|{F^1} {F^2}|\phi_l\>
\notag\\
\stackrel{(b)}{=}&
4 g \lambda_2
\<\phi_{1}|{F^1}  |\phi_{2}\>
\<\phi_2|{F^2}|\phi_1\>   + 4g
 \lambda_1
\<\phi_{2}|{F^1}  |\phi_{1}\>
\<\phi_1|{F^2}|\phi_2\> \notag\\
&-4 g \lambda_2
\<\phi_{1}|{F^2}  |\phi_{2}\>
\<\phi_2|{F^1}|\phi_1\>  
-4 g \lambda_1
\<\phi_{2}|{F^2}  |\phi_{1}\>
\<\phi_1|{F^1}|\phi_2\>  \notag\\
&-
4 i ( \lambda_1 \< \phi_1 | F^3  |\phi_1 \> 
+ \lambda_2 \< \phi_2 | F^3  |\phi_2 \>)\notag\\
\stackrel{(c)}{=}& 
8 g (\lambda_2
\<\phi_{1}|{F^1}  |\phi_{2}\>
\<\phi_2|{F^2}|\phi_1\>   +
 \lambda_1
\<\phi_{2}|{F^1}  |\phi_{1}\>
\<\phi_1|{F^2}|\phi_2\>) \notag\\
&-
4 i ( \lambda_1 \< \phi_1 | E^3  |\phi_1 \> 
+ \lambda_2 \< \phi_2 | E^3  |\phi_2 \>).\label{BNGH3}
\end{align}
In $(a)$ we use  
$\<\phi_1|F^2 |\phi_1\> = 0$ and $\<\phi_2|F^2 |\phi_2\> = 0$ so that the first summation in \eqref{full-expression} vanishes.
In $(b)$ we use
$\< \phi_j |( F^2 F^1 - F^1 F^2 )  |\phi_j \>
 = -i \< \phi_j |F^3  |\phi_j \>$.
In $(c)$ we use the following. First $i\<\phi_1|F^2 |\phi_2\>\in \mathbb R$ implies that 
$\<\phi_1|F^2 |\phi_2\> = - \<\phi_2|F^2 |\phi_1\>$. 
Second, $\<\phi_1|F^1 |\phi_2\>\in \mathbb R$
implies that $\<\phi_1|F^1 |\phi_2\> = \<\phi_2|F^1 |\phi_1\>$.
Hence 
$\<\phi_2|F^1 |\phi_1\> \<\phi_1|F^2 |\phi_2\> = 
-\<\phi_1|F^1 |\phi_2\> \<\phi_2|F^2 |\phi_1\>$, which implies
that $\<\phi_2|F^1 |\phi_1\> \<\phi_1|F^2 |\phi_2\>-  \<\phi_1|F^1 |\phi_2\> \<\phi_2|F^2 |\phi_1\> = 2\<\phi_2|F^1 |\phi_1\> \<\phi_1|F^2 |\phi_2\>$, which implies $(c)$.

Then, we have   
\begin{align}
 &\tr[\bar L_1 \bar \rho_\gamma \bar  L_2]
-
\tr[\bar L_2  \bar \rho_\gamma \bar L_1]\notag\\
\stackrel{(a)}{=}&
8 g (
\<\phi|{F^1}  |\epsilon\>
\<\epsilon|{F^2}|\phi\> / \lambda_1   +
\<\epsilon|{F^1}  |\phi\>
\<\phi|{F^2}|\epsilon\>/\lambda_2) \notag\\
&-4i(\<\phi|{F^3}|\phi\> + \<\epsilon|{F^3}|\epsilon\>)\notag\\
\stackrel{(b)}{=}&
8g (-i\mu^2\gamma^2 + O(\gamma^4) )
+ 
8g (( i\mu^2\gamma^2+ O(\gamma^4)  )  /  (2\mu\gamma^2) )\notag\\
&-
4 i ( 2 \mu \gamma^2  + O(\gamma^4))\notag \\
\stackrel{(c)}{=}& 
4i \mu g  - 8 i \mu \gamma^2 + O(\gamma^4) \notag\\
\stackrel{(d)}{=}&
(-32 i \mu^2 - 8i \mu)\gamma^2 + O(\gamma^4) \notag\\
\stackrel{(e)}{=}&
-8 i (4\mu^2 + \mu) \gamma^2 + O(\gamma^4) \notag\\
\stackrel{(f)}{=}&
\frac{- 2 i n(n+2)}{3} (1 + n(n+2)/3)\gamma^2 + O(\gamma^4) \label{bar commutator}.
\end{align} 
In $(a)$, we use $|\phi\>=|\phi_1\>/\sqrt{\lambda_1}, |\epsilon\> =|\phi_2\>/\sqrt{\lambda_2}$, and
\eqref{BNGH3}.
In $(b)$, we use Lemma \ref{lem:E1E2-cross-products}, and the fact that $\lambda_1 = \<\phi|\phi\>, \lambda_2 = \<\epsilon|\epsilon\>$ with \eqref{pp} and \eqref{ee}.
In $(c)$, we collect terms of leading order in $\gamma$.
In $(d)$, we use \eqref{gval}, namely, $g= -8\mu \gamma^2 +O(\gamma^4)$.
In $(e)$, we substitute the value of $\mu = n(n+2)/12$.

In Appendix \ref{app:commutator perturbation}, we prove the following lemma.
\begin{lemma}
\label{lem:commutator perturbation}
Let $u,v = 1,2,3$. Then
\begin{align}
\tr[L_u \rho_\gamma L_v]  -
\tr[\bar L_u \bar \rho_\gamma \bar L_v]   = O(\gamma^4).\label{EE160}
\end{align}
\end{lemma}
Then from \eqref{bar commutator} and Lemma \ref{lem:commutator perturbation}, we can get Lemma \ref{lem:trL1rhoL2-trL2rhoL1}.

\section{Discussions}
\label{sec:discussions}

We have shown that the channel estimation problems that correspond to $S_j$ admit corresponding conic optimizations $J_j$ that allow us to find the corresponding Cram\'er-Rao bound for the optimal entangled state. We illustrate the power of our framework with theoretical analysis on the scenario when the maximally entangled state is the optimal probe state, and also with numerical analysis for the often studied field sensing problem using quantum probe states.

We believe that the theory we develop here has applications that extend beyond the problem of field sensing using probe states. In fact, any problem where we estimate multiple incompatible parameters embedded within a quantum channel using entangled probe states stands to benefit from our theory. This encompasses many applications beyond that of quantum field sensing, such as in quantum imaging \cite{moreau2019imaging}.

In a recent paper \cite{ouyang2022quantum}, it was shown that field sensing using quantum probe states in the face of a linear rate of errors can approach the Heisenberg limit if we use quantum error correction on appropriate permutation-invariant codes \cite{ouyang2014permutation,ouyang2021permutation}. However, the corresponding question of what can be done in the setting of using entangled probe states was not considered, and remains an interesting open problem.

Next we like to discuss the efficiency of evaluating the various bounds $J_j$.
When $j=2,\dots, 5$, the optimizations for $J_j$ are efficiently solvable by SDPs. 
When $j=1$, although the conic programing $S_1$ 
cannot be considered as a SDP, 
it can be approximately calculated with SDP 
by employing the concept of
symmetric extension.
Originally, symmetric extension was introduced to considering the membership problem for the separability \cite{DPS1,DPS2}.
We consider the systems 
${\mathbb C^{d+1}}^{\otimes n} \otimes {\cal H}$,
${\mathbb C^{d+1}} \otimes {\cal H}^{\otimes n}$
and define 
\begin{align}
{\mathcal S}^1_{n}:=&
\{X \in
{\cal B}({\mathbb C^{d+1}}^{\otimes n},{\cal H})
| \Tr_{j^c} X = \Tr_{1^c} X  , X\ge 0
\}\notag  \\
\tilde{\mathcal S}^1_{n}:=&
\{X \in
{\cal B}({\mathbb C^{d+1}},{\cal H}^{\otimes n})
| \Tr_{j^c} X = \Tr_{1^c} X  , X\ge 0
\} ,\notag
\end{align}
where $\Tr_{j^c}$ expresses the partial trace 
except for the $j$-th system on 
${\mathbb C^{d+1}}^{\otimes n}$ or ${\cal H}^{\otimes n}$.
The minimizer $X_*$ of $S_{1}$ has a symmetric extension $X_{*,n}$
that belongs to
${\mathcal S}^1_{n}$ and $\tilde{\mathcal S}^1_{n}$
and satisfies $\Tr_{1^c} X_{*,n}= X_*$.
Due to the condition $\Tr_{1^c} X_{*,n}= X_* \in S_{1}$, 
we have the following lower bounds of $S_{1}$:
\begin{align}
S_{1,n}&:= \min_{X \in \mathcal S^1_{n}}
\{\Tr ( G \otimes \rho) \Tr_{1^c} X |
\Tr_{1^c} \hbox{ satisfies }\eqref{c1},\eqref{c2}.
\}\Label{o1-T1} \\
\tilde{S}_{1,n}&:= \min_{X \in \tilde{\mathcal S}^1_{n}}
\{\Tr ( G \otimes \rho) \Tr_{1^c} X |
\Tr_{1^c} \hbox{ satisfies }\eqref{c1},\eqref{c2}.
\}\Label{o1-T2}
\end{align} 
The above quantities can be calculated by SDP.
Since an element $X$ of ${\mathcal S}^1_{n}$ or $\tilde{\mathcal S}^1_{n}$ satisfies $\Tr_{1^c} X \in \mathcal S^2$,
we have
\begin{align}
S_{1} \ge S_{1,n+1} \ge S_{1,n}\ge S_{2},\quad
S_{1} \ge \tilde{S}_{1,n+1} \ge \tilde{S}_{1,n}\ge S_{2}.\notag
\end{align}

Also, 
\cite{DPS2} showed that 
for any non-separable state $\rho$, there exists an integer $n$
such that $\rho$ does not belong to $S_{1,n}$.
The speed of the convergence is studied in \cite{NOP,Fawzi}.
Therefore, we have
\begin{align}
S_{1} = \lim_{n\to \infty} {S}_{1,n}
=\lim_{n\to \infty}  \tilde{S}_{1,n}.\notag
\end{align}
In addition, 
this discussion can be applied to 
the approximate calculation of $J_1$ by SDP.

As an application of our theory, we considered the canonical problem of field sensing. We expect that our theory can lend insight to many other problems that can be phrased naturally in the channel estimation framework, such as for other applications in quantum imaging and sensing.

\section*{Acknowledgement}
MH is supported in part by the National Natural Science Foundation of China (Grant No. 62171212).
YO acknowledges support from EPSRC (Grant No. EP/W028115/1).

\bibliography{ref2}{}

\begin{thebibliography}{51}%
\makeatletter
\providecommand \@ifxundefined [1]{%
 \@ifx{#1\undefined}
}%
\providecommand \@ifnum [1]{%
 \ifnum #1\expandafter \@firstoftwo
 \else \expandafter \@secondoftwo
 \fi
}%
\providecommand \@ifx [1]{%
 \ifx #1\expandafter \@firstoftwo
 \else \expandafter \@secondoftwo
 \fi
}%
\providecommand \natexlab [1]{#1}%
\providecommand \enquote  [1]{``#1''}%
\providecommand \bibnamefont  [1]{#1}%
\providecommand \bibfnamefont [1]{#1}%
\providecommand \citenamefont [1]{#1}%
\providecommand \href@noop [0]{\@secondoftwo}%
\providecommand \href [0]{\begingroup \@sanitize@url \@href}%
\providecommand \@href[1]{\@@startlink{#1}\@@href}%
\providecommand \@@href[1]{\endgroup#1\@@endlink}%
\providecommand \@sanitize@url [0]{\catcode `\\12\catcode `\$12\catcode
  `\&12\catcode `\#12\catcode `\^12\catcode `\_12\catcode `\%12\relax}%
\providecommand \@@startlink[1]{}%
\providecommand \@@endlink[0]{}%
\providecommand \url  [0]{\begingroup\@sanitize@url \@url }%
\providecommand \@url [1]{\endgroup\@href {#1}{\urlprefix }}%
\providecommand \urlprefix  [0]{URL }%
\providecommand \Eprint [0]{\href }%
\providecommand \doibase [0]{https://doi.org/}%
\providecommand \selectlanguage [0]{\@gobble}%
\providecommand \bibinfo  [0]{\@secondoftwo}%
\providecommand \bibfield  [0]{\@secondoftwo}%
\providecommand \translation [1]{[#1]}%
\providecommand \BibitemOpen [0]{}%
\providecommand \bibitemStop [0]{}%
\providecommand \bibitemNoStop [0]{.\EOS\space}%
\providecommand \EOS [0]{\spacefactor3000\relax}%
\providecommand \BibitemShut  [1]{\csname bibitem#1\endcsname}%
\let\auto@bib@innerbib\@empty
\bibitem [{\citenamefont {Escher}\ \emph {et~al.}(2011)\citenamefont {Escher},
  \citenamefont {de~Matos~Filho},\ and\ \citenamefont
  {Davidovich}}]{escher2011general}%
  \BibitemOpen
  \bibfield  {author} {\bibinfo {author} {\bibfnamefont {B.}~\bibnamefont
  {Escher}}, \bibinfo {author} {\bibfnamefont {R.~L.}\ \bibnamefont
  {de~Matos~Filho}},\ and\ \bibinfo {author} {\bibfnamefont {L.}~\bibnamefont
  {Davidovich}},\ }\bibfield  {title} {\bibinfo {title} {General framework for
  estimating the ultimate precision limit in noisy quantum-enhanced
  metrology},\ }\href {https://doi.org/10.1038/nphys1958} {\bibfield  {journal}
  {\bibinfo  {journal} {Nature Physics}\ }\textbf {\bibinfo {volume} {7}},\
  \bibinfo {pages} {406} (\bibinfo {year} {2011})}\BibitemShut {NoStop}%
\bibitem [{\citenamefont {Hayashi}(2011)}]{Hayashi11}%
  \BibitemOpen
  \bibfield  {author} {\bibinfo {author} {\bibfnamefont {M.}~\bibnamefont
  {Hayashi}},\ }\bibfield  {title} {\bibinfo {title} {Comparison between the
  {C}ramer-{R}ao and the mini-max approaches in quantum channel estimation},\
  }\href {https://doi.org/10.1007/s00220-011-1239-4} {\bibfield  {journal}
  {\bibinfo  {journal} {Commun. Math. Phys.}\ }\textbf {\bibinfo {volume}
  {304}},\ \bibinfo {pages} {689 } (\bibinfo {year} {2011})}\BibitemShut
  {NoStop}%
\bibitem [{\citenamefont {Pirandola}\ \emph {et~al.}(2019)\citenamefont
  {Pirandola}, \citenamefont {Laurenza}, \citenamefont {Lupo},\ and\
  \citenamefont {Pereira}}]{pirandola2019fundamental}%
  \BibitemOpen
  \bibfield  {author} {\bibinfo {author} {\bibfnamefont {S.}~\bibnamefont
  {Pirandola}}, \bibinfo {author} {\bibfnamefont {R.}~\bibnamefont {Laurenza}},
  \bibinfo {author} {\bibfnamefont {C.}~\bibnamefont {Lupo}},\ and\ \bibinfo
  {author} {\bibfnamefont {J.~L.}\ \bibnamefont {Pereira}},\ }\bibfield
  {title} {\bibinfo {title} {Fundamental limits to quantum channel
  discrimination},\ }\href {https://doi.org/10.1038/s41534-019-0162-y}
  {\bibfield  {journal} {\bibinfo  {journal} {npj Quantum Information}\
  }\textbf {\bibinfo {volume} {5}},\ \bibinfo {pages} {50} (\bibinfo {year}
  {2019})}\BibitemShut {NoStop}%
\bibitem [{\citenamefont {Zhou}\ and\ \citenamefont
  {Jiang}(2021)}]{zhou2021asymptotic}%
  \BibitemOpen
  \bibfield  {author} {\bibinfo {author} {\bibfnamefont {S.}~\bibnamefont
  {Zhou}}\ and\ \bibinfo {author} {\bibfnamefont {L.}~\bibnamefont {Jiang}},\
  }\bibfield  {title} {\bibinfo {title} {Asymptotic theory of quantum channel
  estimation},\ }\href {https://doi.org/10.1103/PRXQuantum.2.010343} {\bibfield
   {journal} {\bibinfo  {journal} {PRX Quantum}\ }\textbf {\bibinfo {volume}
  {2}},\ \bibinfo {pages} {010343} (\bibinfo {year} {2021})}\BibitemShut
  {NoStop}%
\bibitem [{\citenamefont {Helstrom}(1967)}]{HELSTROM1967101}%
  \BibitemOpen
  \bibfield  {author} {\bibinfo {author} {\bibfnamefont {C.}~\bibnamefont
  {Helstrom}},\ }\bibfield  {title} {\bibinfo {title} {Minimum mean-squared
  error of estimates in quantum statistics},\ }\href
  {https://doi.org/https://doi.org/10.1016/0375-9601(67)90366-0} {\bibfield
  {journal} {\bibinfo  {journal} {Physics Letters A}\ }\textbf {\bibinfo
  {volume} {25}},\ \bibinfo {pages} {101} (\bibinfo {year} {1967})}\BibitemShut
  {NoStop}%
\bibitem [{\citenamefont {Helstrom}(1976)}]{helstrom}%
  \BibitemOpen
  \bibfield  {author} {\bibinfo {author} {\bibfnamefont {C.~W.}\ \bibnamefont
  {Helstrom}},\ }\href {https://doi.org/10.1007/BF01007479} {\emph {\bibinfo
  {title} {Quantum detection and estimation theory}}}\ (\bibinfo  {publisher}
  {Academic press},\ \bibinfo {year} {1976})\BibitemShut {NoStop}%
\bibitem [{\citenamefont {Holevo}(2011)}]{holevo}%
  \BibitemOpen
  \bibfield  {author} {\bibinfo {author} {\bibfnamefont {A.~S.}\ \bibnamefont
  {Holevo}},\ }\href {https://doi.org/10.1007/978-88-7642-378-9} {\emph
  {\bibinfo {title} {Probabilistic and statistical aspects of quantum
  theory}}}\ (\bibinfo  {publisher} {Edizioni della Normale},\ \bibinfo {year}
  {2011})\BibitemShut {NoStop}%
\bibitem [{\citenamefont {Nagaoka}(1989{\natexlab{a}})}]{nagaoka89}%
  \BibitemOpen
  \bibfield  {author} {\bibinfo {author} {\bibfnamefont {H.}~\bibnamefont
  {Nagaoka}},\ }\bibfield  {title} {\bibinfo {title} {A new approach to
  {C}ram\'er-{R}ao bounds for quantum state estimation},\ }\href@noop {}
  {\bibfield  {journal} {\bibinfo  {journal} {IEICE Tech Report}\ }\textbf
  {\bibinfo {volume} {IT 89-42}},\ \bibinfo {pages} {9} (\bibinfo {year}
  {1989}{\natexlab{a}})},\ \bibinfo {note} {(Reprinted in
  \cite{hayashi})}\BibitemShut {NoStop}%
\bibitem [{\citenamefont {Hayashi}\ and\ \citenamefont
  {Matsumoto}(2008{\natexlab{a}})}]{HM08}%
  \BibitemOpen
  \bibfield  {author} {\bibinfo {author} {\bibfnamefont {M.}~\bibnamefont
  {Hayashi}}\ and\ \bibinfo {author} {\bibfnamefont {K.}~\bibnamefont
  {Matsumoto}},\ }\bibfield  {title} {\bibinfo {title} {Asymptotic performance
  of optimal state estimation in qubit system},\ }\href
  {https://doi.org/10.1063/1.2988130} {\bibfield  {journal} {\bibinfo
  {journal} {Journal of Mathematical Physics}\ }\textbf {\bibinfo {volume}
  {49}},\ \bibinfo {pages} {102101} (\bibinfo {year}
  {2008}{\natexlab{a}})}\BibitemShut {NoStop}%
\bibitem [{\citenamefont {Hayashi}(2005)}]{hayashi}%
  \BibitemOpen
  \bibinfo {editor} {\bibfnamefont {M.}~\bibnamefont {Hayashi}},\ ed.,\ \href
  {https://doi.org/10.1142/5630} {\emph {\bibinfo {title} {Asymptotic theory of
  quantum statistical inference: selected papers}}}\ (\bibinfo  {publisher}
  {World Scientific},\ \bibinfo {year} {2005})\BibitemShut {NoStop}%
\bibitem [{\citenamefont {Hayashi}\ and\ \citenamefont
  {Matsumoto}(2008{\natexlab{b}})}]{HM}%
  \BibitemOpen
  \bibfield  {author} {\bibinfo {author} {\bibfnamefont {M.}~\bibnamefont
  {Hayashi}}\ and\ \bibinfo {author} {\bibfnamefont {K.}~\bibnamefont
  {Matsumoto}},\ }\bibfield  {title} {\bibinfo {title} {Asymptotic performance
  of optimal state estimation in qubit system},\ }\href
  {https://doi.org/10.1063/1.2988130} {\bibfield  {journal} {\bibinfo
  {journal} {Journal of Mathematical Physics}\ }\textbf {\bibinfo {volume}
  {49}},\ \bibinfo {pages} {102101} (\bibinfo {year}
  {2008}{\natexlab{b}})}\BibitemShut {NoStop}%
\bibitem [{\citenamefont {Demkowicz-Dobrza{\'n}ski}\ \emph
  {et~al.}(2020)\citenamefont {Demkowicz-Dobrza{\'n}ski}, \citenamefont
  {G{\'o}recki},\ and\ \citenamefont {Gu{\c{t}}{\u{a}}}}]{D_and_D_2020}%
  \BibitemOpen
  \bibfield  {author} {\bibinfo {author} {\bibfnamefont {R.}~\bibnamefont
  {Demkowicz-Dobrza{\'n}ski}}, \bibinfo {author} {\bibfnamefont
  {W.}~\bibnamefont {G{\'o}recki}},\ and\ \bibinfo {author} {\bibfnamefont
  {M.}~\bibnamefont {Gu{\c{t}}{\u{a}}}},\ }\bibfield  {title} {\bibinfo {title}
  {Multi-parameter estimation beyond quantum fisher information},\ }\href
  {https://doi.org/10.1088/1751-8121/ab8ef3} {\bibfield  {journal} {\bibinfo
  {journal} {Journal of Physics A: Mathematical and Theoretical}\ }\textbf
  {\bibinfo {volume} {53}},\ \bibinfo {pages} {363001} (\bibinfo {year}
  {2020})}\BibitemShut {NoStop}%
\bibitem [{\citenamefont {Sidhu}\ and\ \citenamefont
  {Kok}(2020)}]{sidhu2020geometric}%
  \BibitemOpen
  \bibfield  {author} {\bibinfo {author} {\bibfnamefont {J.~S.}\ \bibnamefont
  {Sidhu}}\ and\ \bibinfo {author} {\bibfnamefont {P.}~\bibnamefont {Kok}},\
  }\bibfield  {title} {\bibinfo {title} {Geometric perspective on quantum
  parameter estimation},\ }\href {https://doi.org/10.1116/1.5119961} {\bibfield
   {journal} {\bibinfo  {journal} {AVS Quantum Science}\ }\textbf {\bibinfo
  {volume} {2}},\ \bibinfo {pages} {014701} (\bibinfo {year}
  {2020})}\BibitemShut {NoStop}%
\bibitem [{\citenamefont {Hayashi}\ and\ \citenamefont {Ouyang}(2023)}]{HO}%
  \BibitemOpen
  \bibfield  {author} {\bibinfo {author} {\bibfnamefont {M.}~\bibnamefont
  {Hayashi}}\ and\ \bibinfo {author} {\bibfnamefont {Y.}~\bibnamefont
  {Ouyang}},\ }\bibfield  {title} {\bibinfo {title} {Tight {C}ram{\'{e}}r-{R}ao
  type bounds for multiparameter quantum metrology through conic programming},\
  }\href {https://doi.org/10.22331/q-2023-08-29-1094} {\bibfield  {journal}
  {\bibinfo  {journal} {{Quantum}}\ }\textbf {\bibinfo {volume} {7}},\ \bibinfo
  {pages} {1094} (\bibinfo {year} {2023})}\BibitemShut {NoStop}%
\bibitem [{\citenamefont {Albarelli}\ \emph {et~al.}(2019)\citenamefont
  {Albarelli}, \citenamefont {Friel},\ and\ \citenamefont
  {Datta}}]{Albarelli2019_PRL}%
  \BibitemOpen
  \bibfield  {author} {\bibinfo {author} {\bibfnamefont {F.}~\bibnamefont
  {Albarelli}}, \bibinfo {author} {\bibfnamefont {J.~F.}\ \bibnamefont
  {Friel}},\ and\ \bibinfo {author} {\bibfnamefont {A.}~\bibnamefont {Datta}},\
  }\bibfield  {title} {\bibinfo {title} {Evaluating the {H}olevo
  {C}ram{\'e}r-{R}ao bound for multiparameter quantum metrology},\ }\href
  {https://doi.org/10.1103/PhysRevLett.123.200503} {\bibfield  {journal}
  {\bibinfo  {journal} {Phys. Rev. Lett.}\ }\textbf {\bibinfo {volume} {123}},\
  \bibinfo {pages} {200503} (\bibinfo {year} {2019})}\BibitemShut {NoStop}%
\bibitem [{\citenamefont {Sidhu}\ \emph {et~al.}(2021)\citenamefont {Sidhu},
  \citenamefont {Ouyang}, \citenamefont {Campbell},\ and\ \citenamefont
  {Kok}}]{sidhu2020tight}%
  \BibitemOpen
  \bibfield  {author} {\bibinfo {author} {\bibfnamefont {J.~S.}\ \bibnamefont
  {Sidhu}}, \bibinfo {author} {\bibfnamefont {Y.}~\bibnamefont {Ouyang}},
  \bibinfo {author} {\bibfnamefont {E.~T.}\ \bibnamefont {Campbell}},\ and\
  \bibinfo {author} {\bibfnamefont {P.}~\bibnamefont {Kok}},\ }\bibfield
  {title} {\bibinfo {title} {Tight bounds on the simultaneous estimation of
  incompatible parameters},\ }\href
  {https://doi.org/10.1103/PhysRevX.11.011028} {\bibfield  {journal} {\bibinfo
  {journal} {Phys. Rev. X}\ }\textbf {\bibinfo {volume} {11}},\ \bibinfo
  {pages} {011028} (\bibinfo {year} {2021})}\BibitemShut {NoStop}%
\bibitem [{\citenamefont {Hayashi}(1997)}]{Haya}%
  \BibitemOpen
  \bibfield  {author} {\bibinfo {author} {\bibfnamefont {M.}~\bibnamefont
  {Hayashi}},\ }\href@noop {} {\bibinfo {title} {A linear programming approach
  to attainable {C}ram\'{e}r-{R}ao type bounds and randomness condition}}
  (\bibinfo {year} {1997}),\ \Eprint {https://arxiv.org/abs/quant-ph/9704044}
  {arXiv:quant-ph/9704044 [quant-ph]} \BibitemShut {NoStop}%
\bibitem [{\citenamefont {Nagaoka}(2005)}]{Nagaoka-generalization}%
  \BibitemOpen
  \bibfield  {author} {\bibinfo {author} {\bibfnamefont {H.}~\bibnamefont
  {Nagaoka}},\ }\bibfield  {title} {\bibinfo {title} {A generalization of the
  simultaneous diagonalization of hermitian matrices and its relation to
  quantum estimation theory},\ }in\ \href
  {https://doi.org/10.1142/9789812563071_0012} {\emph {\bibinfo {booktitle}
  {Asymptotic theory of quantum statistical inference: selected papers}}},\
  \bibinfo {editor} {edited by\ \bibinfo {editor} {\bibfnamefont
  {M.}~\bibnamefont {Hayashi}}}\ (\bibinfo  {publisher} {World Scientific},\
  \bibinfo {year} {2005})\ p.\ \bibinfo {pages} {133–149}\BibitemShut
  {NoStop}%
\bibitem [{\citenamefont {Hayashi}(1999)}]{Hayanoncomm}%
  \BibitemOpen
  \bibfield  {author} {\bibinfo {author} {\bibfnamefont {M.}~\bibnamefont
  {Hayashi}},\ }\bibfield  {title} {\bibinfo {title} {On simultaneous
  measurement of noncommutative observables (in {J}apanese)},\ }\href
  {https://www.kurims.kyoto-u.ac.jp/~kyodo/kokyuroku/contents/pdf/1099-8.pdf}
  {\bibfield  {journal} {\bibinfo  {journal} {Surikaisekikenkyusho {(RIMS)}
  Kokyuroku (Development of Infinite-Dimensional Noncommutative Analysis)}\ ,\
  \bibinfo {pages} {96}} (\bibinfo {year} {1999})}\BibitemShut {NoStop}%
\bibitem [{\citenamefont {Conlon}\ \emph {et~al.}(2021)\citenamefont {Conlon},
  \citenamefont {Suzuki}, \citenamefont {Lam},\ and\ \citenamefont
  {Assad}}]{CSLA}%
  \BibitemOpen
  \bibfield  {author} {\bibinfo {author} {\bibfnamefont {L.~O.}\ \bibnamefont
  {Conlon}}, \bibinfo {author} {\bibfnamefont {J.}~\bibnamefont {Suzuki}},
  \bibinfo {author} {\bibfnamefont {P.~K.}\ \bibnamefont {Lam}},\ and\ \bibinfo
  {author} {\bibfnamefont {S.~M.}\ \bibnamefont {Assad}},\ }\bibfield  {title}
  {\bibinfo {title} {Efficient computation of the {N}agaoka--{H}ayashi bound
  for multiparameter estimation with separable measurements},\ }\href
  {https://doi.org/10.1038/s41534-021-00414-1} {\bibfield  {journal} {\bibinfo
  {journal} {npj Quantum Information}\ }\textbf {\bibinfo {volume} {7}},\
  \bibinfo {pages} {1} (\bibinfo {year} {2021})}\BibitemShut {NoStop}%
\bibitem [{\citenamefont {Liu}\ \emph {et~al.}(2023)\citenamefont {Liu},
  \citenamefont {Hu}, \citenamefont {Yuan},\ and\ \citenamefont {Yang}}]{Liu}%
  \BibitemOpen
  \bibfield  {author} {\bibinfo {author} {\bibfnamefont {Q.}~\bibnamefont
  {Liu}}, \bibinfo {author} {\bibfnamefont {Z.}~\bibnamefont {Hu}}, \bibinfo
  {author} {\bibfnamefont {H.}~\bibnamefont {Yuan}},\ and\ \bibinfo {author}
  {\bibfnamefont {Y.}~\bibnamefont {Yang}},\ }\bibfield  {title} {\bibinfo
  {title} {Optimal strategies of quantum metrology with a strict hierarchy},\
  }\href {https://doi.org/10.1103/PhysRevLett.130.070803} {\bibfield  {journal}
  {\bibinfo  {journal} {Phys. Rev. Lett.}\ }\textbf {\bibinfo {volume} {130}},\
  \bibinfo {pages} {070803} (\bibinfo {year} {2023})}\BibitemShut {NoStop}%
\bibitem [{\citenamefont {Altherr}\ and\ \citenamefont {Yang}(2021)}]{Altherr}%
  \BibitemOpen
  \bibfield  {author} {\bibinfo {author} {\bibfnamefont {A.}~\bibnamefont
  {Altherr}}\ and\ \bibinfo {author} {\bibfnamefont {Y.}~\bibnamefont {Yang}},\
  }\bibfield  {title} {\bibinfo {title} {Quantum metrology for non-markovian
  processes},\ }\href {https://doi.org/10.1103/PhysRevLett.127.060501}
  {\bibfield  {journal} {\bibinfo  {journal} {Phys. Rev. Lett.}\ }\textbf
  {\bibinfo {volume} {127}},\ \bibinfo {pages} {060501} (\bibinfo {year}
  {2021})}\BibitemShut {NoStop}%
\bibitem [{\citenamefont {G{\'o}recki}\ \emph {et~al.}(2020)\citenamefont
  {G{\'o}recki}, \citenamefont {Zhou}, \citenamefont {Jiang},\ and\
  \citenamefont {Demkowicz-Dobrza{\'n}ski}}]{gorecki2020quantum}%
  \BibitemOpen
  \bibfield  {author} {\bibinfo {author} {\bibfnamefont {W.}~\bibnamefont
  {G{\'o}recki}}, \bibinfo {author} {\bibfnamefont {S.}~\bibnamefont {Zhou}},
  \bibinfo {author} {\bibfnamefont {L.}~\bibnamefont {Jiang}},\ and\ \bibinfo
  {author} {\bibfnamefont {R.}~\bibnamefont {Demkowicz-Dobrza{\'n}ski}},\
  }\bibfield  {title} {\bibinfo {title} {Optimal probes and error-correction
  schemes in multi-parameter quantum metrology},\ }\href
  {https://doi.org/10.22331/q-2020-07-02-288} {\bibfield  {journal} {\bibinfo
  {journal} {Quantum}\ }\textbf {\bibinfo {volume} {4}},\ \bibinfo {pages}
  {288} (\bibinfo {year} {2020})}\BibitemShut {NoStop}%
\bibitem [{\citenamefont {Friel}\ \emph {et~al.}(2020)\citenamefont {Friel},
  \citenamefont {Palittapongarnpim}, \citenamefont {Albarelli},\ and\
  \citenamefont {Datta}}]{friel2020attainability}%
  \BibitemOpen
  \bibfield  {author} {\bibinfo {author} {\bibfnamefont {J.}~\bibnamefont
  {Friel}}, \bibinfo {author} {\bibfnamefont {P.}~\bibnamefont
  {Palittapongarnpim}}, \bibinfo {author} {\bibfnamefont {F.}~\bibnamefont
  {Albarelli}},\ and\ \bibinfo {author} {\bibfnamefont {A.}~\bibnamefont
  {Datta}},\ }\bibfield  {title} {\bibinfo {title} {Attainability of the
  {H}olevo-{C}ram\'er-{R}ao bound for two-qubit 3d magnetometry},\ }\bibfield
  {journal} {\bibinfo  {journal} {arXiv preprint arXiv:2008.01502}\ }\href
  {https://doi.org/10.48550/arXiv.2008.01502} {10.48550/arXiv.2008.01502}
  (\bibinfo {year} {2020})\BibitemShut {NoStop}%
\bibitem [{\citenamefont {Ac\'{\i}n}(2001)}]{acin-PhysRevLett.87.177901}%
  \BibitemOpen
  \bibfield  {author} {\bibinfo {author} {\bibfnamefont {A.}~\bibnamefont
  {Ac\'{\i}n}},\ }\bibfield  {title} {\bibinfo {title} {Statistical
  distinguishability between unitary operations},\ }\href
  {https://doi.org/10.1103/PhysRevLett.87.177901} {\bibfield  {journal}
  {\bibinfo  {journal} {Phys. Rev. Lett.}\ }\textbf {\bibinfo {volume} {87}},\
  \bibinfo {pages} {177901} (\bibinfo {year} {2001})}\BibitemShut {NoStop}%
\bibitem [{\citenamefont {Sacchi}(2005)}]{sacchi-PhysRevA.71.062340}%
  \BibitemOpen
  \bibfield  {author} {\bibinfo {author} {\bibfnamefont {M.~F.}\ \bibnamefont
  {Sacchi}},\ }\bibfield  {title} {\bibinfo {title} {Optimal discrimination of
  quantum operations},\ }\href {https://doi.org/10.1103/PhysRevA.71.062340}
  {\bibfield  {journal} {\bibinfo  {journal} {Phys. Rev. A}\ }\textbf {\bibinfo
  {volume} {71}},\ \bibinfo {pages} {062340} (\bibinfo {year}
  {2005})}\BibitemShut {NoStop}%
\bibitem [{\citenamefont {Hayashi}(2009)}]{hay2009}%
  \BibitemOpen
  \bibfield  {author} {\bibinfo {author} {\bibfnamefont {M.}~\bibnamefont
  {Hayashi}},\ }\bibfield  {title} {\bibinfo {title} {Discrimination of two
  channels by adaptive methods and its application to quantum system},\ }\href
  {https://doi.org/10.1109/TIT.2009.2023726} {\bibfield  {journal} {\bibinfo
  {journal} {IEEE Transactions on Information Theory}\ }\textbf {\bibinfo
  {volume} {55}},\ \bibinfo {pages} {3807} (\bibinfo {year}
  {2009})}\BibitemShut {NoStop}%
\bibitem [{\citenamefont {Zhuang}\ and\ \citenamefont
  {Pirandola}(2020)}]{zhuang-PhysRevLett.125.080505}%
  \BibitemOpen
  \bibfield  {author} {\bibinfo {author} {\bibfnamefont {Q.}~\bibnamefont
  {Zhuang}}\ and\ \bibinfo {author} {\bibfnamefont {S.}~\bibnamefont
  {Pirandola}},\ }\bibfield  {title} {\bibinfo {title} {Ultimate limits for
  multiple quantum channel discrimination},\ }\href
  {https://doi.org/10.1103/PhysRevLett.125.080505} {\bibfield  {journal}
  {\bibinfo  {journal} {Phys. Rev. Lett.}\ }\textbf {\bibinfo {volume} {125}},\
  \bibinfo {pages} {080505} (\bibinfo {year} {2020})}\BibitemShut {NoStop}%
\bibitem [{\citenamefont {Wilde}\ \emph {et~al.}(2020)\citenamefont {Wilde},
  \citenamefont {Berta}, \citenamefont {Hirche},\ and\ \citenamefont
  {Kaur}}]{wilde2020amortized}%
  \BibitemOpen
  \bibfield  {author} {\bibinfo {author} {\bibfnamefont {M.~M.}\ \bibnamefont
  {Wilde}}, \bibinfo {author} {\bibfnamefont {M.}~\bibnamefont {Berta}},
  \bibinfo {author} {\bibfnamefont {C.}~\bibnamefont {Hirche}},\ and\ \bibinfo
  {author} {\bibfnamefont {E.}~\bibnamefont {Kaur}},\ }\bibfield  {title}
  {\bibinfo {title} {Amortized channel divergence for asymptotic quantum
  channel discrimination},\ }\href {https://doi.org/10.1007/s11005-020-01297-7}
  {\bibfield  {journal} {\bibinfo  {journal} {Letters in Mathematical Physics}\
  }\textbf {\bibinfo {volume} {110}},\ \bibinfo {pages} {2277} (\bibinfo {year}
  {2020})}\BibitemShut {NoStop}%
\bibitem [{\citenamefont {Nakahira}\ and\ \citenamefont
  {Kato}(2021)}]{Nakahira-and-Kato-2021}%
  \BibitemOpen
  \bibfield  {author} {\bibinfo {author} {\bibfnamefont {K.}~\bibnamefont
  {Nakahira}}\ and\ \bibinfo {author} {\bibfnamefont {K.}~\bibnamefont
  {Kato}},\ }\bibfield  {title} {\bibinfo {title} {Generalized quantum process
  discrimination problems},\ }\href
  {https://doi.org/10.1103/PhysRevA.103.062606} {\bibfield  {journal} {\bibinfo
   {journal} {Phys. Rev. A}\ }\textbf {\bibinfo {volume} {103}},\ \bibinfo
  {pages} {062606} (\bibinfo {year} {2021})}\BibitemShut {NoStop}%
\bibitem [{\citenamefont {Moreau}\ \emph {et~al.}(2019)\citenamefont {Moreau},
  \citenamefont {Toninelli}, \citenamefont {Gregory},\ and\ \citenamefont
  {Padgett}}]{moreau2019imaging}%
  \BibitemOpen
  \bibfield  {author} {\bibinfo {author} {\bibfnamefont {P.-A.}\ \bibnamefont
  {Moreau}}, \bibinfo {author} {\bibfnamefont {E.}~\bibnamefont {Toninelli}},
  \bibinfo {author} {\bibfnamefont {T.}~\bibnamefont {Gregory}},\ and\ \bibinfo
  {author} {\bibfnamefont {M.~J.}\ \bibnamefont {Padgett}},\ }\bibfield
  {title} {\bibinfo {title} {Imaging with quantum states of light},\ }\href
  {https://doi.org/10.1038/s42254-019-0056-0} {\bibfield  {journal} {\bibinfo
  {journal} {Nature Reviews Physics}\ }\textbf {\bibinfo {volume} {1}},\
  \bibinfo {pages} {367} (\bibinfo {year} {2019})}\BibitemShut {NoStop}%
\bibitem [{\citenamefont {T{\'o}th}\ and\ \citenamefont
  {Apellaniz}(2014)}]{toth2014quantum}%
  \BibitemOpen
  \bibfield  {author} {\bibinfo {author} {\bibfnamefont {G.}~\bibnamefont
  {T{\'o}th}}\ and\ \bibinfo {author} {\bibfnamefont {I.}~\bibnamefont
  {Apellaniz}},\ }\bibfield  {title} {\bibinfo {title} {Quantum metrology from
  a quantum information science perspective},\ }\href
  {https://doi.org/10.1088/1751-8113/47/42/424006} {\bibfield  {journal}
  {\bibinfo  {journal} {Journal of Physics A: Mathematical and Theoretical}\
  }\textbf {\bibinfo {volume} {47}},\ \bibinfo {pages} {424006} (\bibinfo
  {year} {2014})}\BibitemShut {NoStop}%
\bibitem [{\citenamefont {Duan}\ and\ \citenamefont
  {Guo}(1998)}]{PhysRevA.58.3491}%
  \BibitemOpen
  \bibfield  {author} {\bibinfo {author} {\bibfnamefont {L.-M.}\ \bibnamefont
  {Duan}}\ and\ \bibinfo {author} {\bibfnamefont {G.-C.}\ \bibnamefont {Guo}},\
  }\bibfield  {title} {\bibinfo {title} {Optimal quantum codes for preventing
  collective amplitude damping},\ }\href
  {https://doi.org/10.1103/PhysRevA.58.3491} {\bibfield  {journal} {\bibinfo
  {journal} {Phys. Rev. A}\ }\textbf {\bibinfo {volume} {58}},\ \bibinfo
  {pages} {3491} (\bibinfo {year} {1998})}\BibitemShut {NoStop}%
\bibitem [{\citenamefont {Nagaoka}(1989{\natexlab{b}})}]{Nagaoka}%
  \BibitemOpen
  \bibfield  {author} {\bibinfo {author} {\bibfnamefont {H.}~\bibnamefont
  {Nagaoka}},\ }\bibfield  {title} {\bibinfo {title} {On the parameter
  estimation problem for quantum statistical models},\ }\href@noop {}
  {\bibfield  {journal} {\bibinfo  {journal} {In Proceeding of 12th Symposium
  on Information Theory and Its Applications (SITA), Inuyama, Tottori, Japan,
  December 6–9}\ ,\ \bibinfo {pages} {577}} (\bibinfo {year}
  {1989}{\natexlab{b}})}\BibitemShut {NoStop}%
\bibitem [{\citenamefont {Choi}(1975)}]{choi1975completely}%
  \BibitemOpen
  \bibfield  {author} {\bibinfo {author} {\bibfnamefont {M.-D.}\ \bibnamefont
  {Choi}},\ }\bibfield  {title} {\bibinfo {title} {Completely positive linear
  maps on complex matrices},\ }\href
  {https://doi.org/10.1016/0024-3795(75)90075-0} {\bibfield  {journal}
  {\bibinfo  {journal} {Linear algebra and its applications}\ }\textbf
  {\bibinfo {volume} {10}},\ \bibinfo {pages} {285} (\bibinfo {year}
  {1975})}\BibitemShut {NoStop}%
\bibitem [{\citenamefont {Fujiwara}\ and\ \citenamefont {Imai}(2003)}]{GP1}%
  \BibitemOpen
  \bibfield  {author} {\bibinfo {author} {\bibfnamefont {A.}~\bibnamefont
  {Fujiwara}}\ and\ \bibinfo {author} {\bibfnamefont {H.}~\bibnamefont
  {Imai}},\ }\bibfield  {title} {\bibinfo {title} {Quantum parameter estimation
  of a generalized pauli channel},\ }\href
  {https://doi.org/10.1088/0305-4470/36/29/314} {\bibfield  {journal} {\bibinfo
   {journal} {Journal of Physics A: Mathematical and General}\ }\textbf
  {\bibinfo {volume} {36}},\ \bibinfo {pages} {8093} (\bibinfo {year}
  {2003})}\BibitemShut {NoStop}%
\bibitem [{\citenamefont {Hayashi}(2003)}]{GP2}%
  \BibitemOpen
  \bibfield  {author} {\bibinfo {author} {\bibfnamefont {M.}~\bibnamefont
  {Hayashi}},\ }\href@noop {} {\bibinfo {title} {Private communication to {A}.
  {F}ujiwara}} (\bibinfo {year} {2003})\BibitemShut {NoStop}%
\bibitem [{\citenamefont {Imai}\ and\ \citenamefont {Fujiwara}(2007)}]{Imai}%
  \BibitemOpen
  \bibfield  {author} {\bibinfo {author} {\bibfnamefont {H.}~\bibnamefont
  {Imai}}\ and\ \bibinfo {author} {\bibfnamefont {A.}~\bibnamefont
  {Fujiwara}},\ }\bibfield  {title} {\bibinfo {title} {Geometry of optimal
  estimation scheme for su(d) channels},\ }\href
  {https://doi.org/10.1088/1751-8113/40/16/009} {\bibfield  {journal} {\bibinfo
   {journal} {Journal of Physics A: Mathematical and Theoretical}\ }\textbf
  {\bibinfo {volume} {40}},\ \bibinfo {pages} {4391} (\bibinfo {year}
  {2007})}\BibitemShut {NoStop}%
\bibitem [{\citenamefont {Matsumoto}(2002)}]{matsumoto2002new}%
  \BibitemOpen
  \bibfield  {author} {\bibinfo {author} {\bibfnamefont {K.}~\bibnamefont
  {Matsumoto}},\ }\bibfield  {title} {\bibinfo {title} {A new approach to the
  {C}ram{\'e}r-{R}ao-type bound of the pure-state model},\ }\href
  {https://doi.org/10.1088/0305-4470/35/13/307} {\bibfield  {journal} {\bibinfo
   {journal} {Journal of Physics A: Mathematical and General}\ }\textbf
  {\bibinfo {volume} {35}},\ \bibinfo {pages} {3111} (\bibinfo {year}
  {2002})}\BibitemShut {NoStop}%
\bibitem [{\citenamefont {Chuang}\ \emph {et~al.}(1997)\citenamefont {Chuang},
  \citenamefont {Leung},\ and\ \citenamefont {Yamamoto}}]{CLY97}%
  \BibitemOpen
  \bibfield  {author} {\bibinfo {author} {\bibfnamefont {I.~L.}\ \bibnamefont
  {Chuang}}, \bibinfo {author} {\bibfnamefont {D.~W.}\ \bibnamefont {Leung}},\
  and\ \bibinfo {author} {\bibfnamefont {Y.}~\bibnamefont {Yamamoto}},\
  }\bibfield  {title} {\bibinfo {title} {Bosonic quantum codes for amplitude
  damping},\ }\href {https://doi.org/10.1103/PhysRevA.56.1114} {\bibfield
  {journal} {\bibinfo  {journal} {Phys. Rev. A}\ }\textbf {\bibinfo {volume}
  {56}},\ \bibinfo {pages} {1114} (\bibinfo {year} {1997})}\BibitemShut
  {NoStop}%
\bibitem [{\citenamefont {Baumgratz}\ and\ \citenamefont
  {Datta}(2016)}]{BaumgratzDatta2016PRL}%
  \BibitemOpen
  \bibfield  {author} {\bibinfo {author} {\bibfnamefont {T.}~\bibnamefont
  {Baumgratz}}\ and\ \bibinfo {author} {\bibfnamefont {A.}~\bibnamefont
  {Datta}},\ }\bibfield  {title} {\bibinfo {title} {Quantum enhanced estimation
  of a multidimensional field},\ }\href
  {https://doi.org/10.1103/PhysRevLett.116.030801} {\bibfield  {journal}
  {\bibinfo  {journal} {Phys. Rev. Lett.}\ }\textbf {\bibinfo {volume} {116}},\
  \bibinfo {pages} {030801} (\bibinfo {year} {2016})}\BibitemShut {NoStop}%
\bibitem [{\citenamefont {Ouyang}\ and\ \citenamefont
  {Brennen}(2022)}]{ouyang2022quantum}%
  \BibitemOpen
  \bibfield  {author} {\bibinfo {author} {\bibfnamefont {Y.}~\bibnamefont
  {Ouyang}}\ and\ \bibinfo {author} {\bibfnamefont {G.~K.}\ \bibnamefont
  {Brennen}},\ }\bibfield  {title} {\bibinfo {title} {Quantum error correction
  on symmetric quantum sensors},\ }\bibfield  {journal} {\bibinfo  {journal}
  {arXiv preprint arXiv:2212.06285}\ }\href
  {https://doi.org/10.48550/arXiv.2212.06285} {10.48550/arXiv.2212.06285}
  (\bibinfo {year} {2022})\BibitemShut {NoStop}%
\bibitem [{\citenamefont {Ouyang}(2014)}]{ouyang2014permutation}%
  \BibitemOpen
  \bibfield  {author} {\bibinfo {author} {\bibfnamefont {Y.}~\bibnamefont
  {Ouyang}},\ }\bibfield  {title} {\bibinfo {title} {{P}ermutation-invariant
  quantum codes},\ }\href {https://doi.org/10.1103/PhysRevA.90.062317}
  {\bibfield  {journal} {\bibinfo  {journal} {Physical Review A}\ }\textbf
  {\bibinfo {volume} {90}},\ \bibinfo {pages} {062317} (\bibinfo {year}
  {2014})},\ \Eprint {https://arxiv.org/abs/1302.3247} {1302.3247} \BibitemShut
  {NoStop}%
\bibitem [{\citenamefont {Ouyang}(2021)}]{ouyang2021permutation}%
  \BibitemOpen
  \bibfield  {author} {\bibinfo {author} {\bibfnamefont {Y.}~\bibnamefont
  {Ouyang}},\ }\bibfield  {title} {\bibinfo {title} {Permutation-invariant
  quantum coding for quantum deletion channels},\ }in\ \href
  {https://doi.org/10.1109/ISIT45174.2021.9518078} {\emph {\bibinfo {booktitle}
  {2021 IEEE International Symposium on Information Theory (ISIT)}}}\ (\bibinfo
  {year} {2021})\ pp.\ \bibinfo {pages} {1499--1503}\BibitemShut {NoStop}%
\bibitem [{\citenamefont {Doherty}\ \emph {et~al.}(2002)\citenamefont
  {Doherty}, \citenamefont {Parrilo},\ and\ \citenamefont {Spedalieri}}]{DPS1}%
  \BibitemOpen
  \bibfield  {author} {\bibinfo {author} {\bibfnamefont {A.~C.}\ \bibnamefont
  {Doherty}}, \bibinfo {author} {\bibfnamefont {P.~A.}\ \bibnamefont
  {Parrilo}},\ and\ \bibinfo {author} {\bibfnamefont {F.~M.}\ \bibnamefont
  {Spedalieri}},\ }\bibfield  {title} {\bibinfo {title} {Distinguishing
  separable and entangled states},\ }\href
  {https://doi.org/10.1103/PhysRevLett.88.187904} {\bibfield  {journal}
  {\bibinfo  {journal} {Phys. Rev. Lett.}\ }\textbf {\bibinfo {volume} {88}},\
  \bibinfo {pages} {187904} (\bibinfo {year} {2002})}\BibitemShut {NoStop}%
\bibitem [{\citenamefont {Doherty}\ \emph {et~al.}(2004)\citenamefont
  {Doherty}, \citenamefont {Parrilo},\ and\ \citenamefont {Spedalieri}}]{DPS2}%
  \BibitemOpen
  \bibfield  {author} {\bibinfo {author} {\bibfnamefont {A.~C.}\ \bibnamefont
  {Doherty}}, \bibinfo {author} {\bibfnamefont {P.~A.}\ \bibnamefont
  {Parrilo}},\ and\ \bibinfo {author} {\bibfnamefont {F.~M.}\ \bibnamefont
  {Spedalieri}},\ }\bibfield  {title} {\bibinfo {title} {Complete family of
  separability criteria},\ }\href {https://doi.org/10.1103/PhysRevA.69.022308}
  {\bibfield  {journal} {\bibinfo  {journal} {Phys. Rev. A}\ }\textbf {\bibinfo
  {volume} {69}},\ \bibinfo {pages} {022308} (\bibinfo {year}
  {2004})}\BibitemShut {NoStop}%
\bibitem [{\citenamefont {Navascu\'es}\ \emph {et~al.}(2009)\citenamefont
  {Navascu\'es}, \citenamefont {Owari},\ and\ \citenamefont {Plenio}}]{NOP}%
  \BibitemOpen
  \bibfield  {author} {\bibinfo {author} {\bibfnamefont {M.}~\bibnamefont
  {Navascu\'es}}, \bibinfo {author} {\bibfnamefont {M.}~\bibnamefont {Owari}},\
  and\ \bibinfo {author} {\bibfnamefont {M.~B.}\ \bibnamefont {Plenio}},\
  }\bibfield  {title} {\bibinfo {title} {Power of symmetric extensions for
  entanglement detection},\ }\href {https://doi.org/10.1103/PhysRevA.80.052306}
  {\bibfield  {journal} {\bibinfo  {journal} {Phys. Rev. A}\ }\textbf {\bibinfo
  {volume} {80}},\ \bibinfo {pages} {052306} (\bibinfo {year}
  {2009})}\BibitemShut {NoStop}%
\bibitem [{\citenamefont {Fawzi}(2021)}]{Fawzi}%
  \BibitemOpen
  \bibfield  {author} {\bibinfo {author} {\bibfnamefont {H.}~\bibnamefont
  {Fawzi}},\ }\bibfield  {title} {\bibinfo {title} {The set of separable states
  has no finite semidefinite representation except in dimension $3 \times 2$},\
  }\href {https://doi.org/10.1007/s00220-021-04163-2} {\bibfield  {journal}
  {\bibinfo  {journal} {Communications in Mathematical Physics}\ }\textbf
  {\bibinfo {volume} {386}},\ \bibinfo {pages} {1319} (\bibinfo {year}
  {2021})}\BibitemShut {NoStop}%
\bibitem [{\citenamefont {Ouyang}\ \emph {et~al.}(2022)\citenamefont {Ouyang},
  \citenamefont {Shettell},\ and\ \citenamefont {Markham}}]{ouyang_RQMWSS}%
  \BibitemOpen
  \bibfield  {author} {\bibinfo {author} {\bibfnamefont {Y.}~\bibnamefont
  {Ouyang}}, \bibinfo {author} {\bibfnamefont {N.}~\bibnamefont {Shettell}},\
  and\ \bibinfo {author} {\bibfnamefont {D.}~\bibnamefont {Markham}},\
  }\bibfield  {title} {\bibinfo {title} {Robust quantum metrology with explicit
  symmetric states},\ }\href {https://doi.org/10.1109/TIT.2021.3132634}
  {\bibfield  {journal} {\bibinfo  {journal} {IEEE Transactions on Information
  Theory}\ }\textbf {\bibinfo {volume} {68}},\ \bibinfo {pages} {1809}
  (\bibinfo {year} {2022})}\BibitemShut {NoStop}%
\bibitem [{\citenamefont {Varga}(2004)}]{varga-GCT}%
  \BibitemOpen
  \bibfield  {author} {\bibinfo {author} {\bibfnamefont {R.~S.}\ \bibnamefont
  {Varga}},\ }\href {https://doi.org/10.1007/978-3-642-17798-9} {\emph
  {\bibinfo {title} {Ger\v{s}gorin and his circles}}},\ \bibinfo {edition}
  {1st}\ ed.\ (\bibinfo  {publisher} {Springer-Verlag},\ \bibinfo {year}
  {2004})\BibitemShut {NoStop}%
\bibitem [{\citenamefont {Fan}\ \emph {et~al.}(2018)\citenamefont {Fan},
  \citenamefont {Wang},\ and\ \citenamefont {Zhong}}]{fan2018eigenvector}%
  \BibitemOpen
  \bibfield  {author} {\bibinfo {author} {\bibfnamefont {J.}~\bibnamefont
  {Fan}}, \bibinfo {author} {\bibfnamefont {W.}~\bibnamefont {Wang}},\ and\
  \bibinfo {author} {\bibfnamefont {Y.}~\bibnamefont {Zhong}},\ }\bibfield
  {title} {\bibinfo {title} {An $\ell_\infty$ eigenvector perturbation bound
  and its application to robust covariance estimation},\ }\href
  {http://jmlr.org/papers/v18/16-140.html} {\bibfield  {journal} {\bibinfo
  {journal} {Journal of Machine Learning Research}\ }\textbf {\bibinfo {volume}
  {18}},\ \bibinfo {pages} {1} (\bibinfo {year} {2018})}\BibitemShut {NoStop}%
\end{thebibliography}%

\appendix

\begin{widetext}
\section{Proofs of Lemmas \ref{L1}, \ref{L3}, and \ref{L2}}
\label{sec:proofs}
\subsection{Proof of Lemma \ref{L1}}
Using the conditions \eqref{NB2} and \eqref{NB3},
for $k=1,2,3,4$, we can write
\begin{align}
&S_k[\Tr_A (T\otimes I_C) (I_B \otimes \rho_{AC}),(\Tr_A (F_j\otimes I_C)( I_B \otimes \rho_{AC}))_j ]\notag\\
=&
 \min_{X \in {\cal S}_{BC}^k}
\left\{\Tr ( G \otimes T\otimes I_C) (I_{RB} \otimes \rho_{AC})
(I_A \otimes X) \left|
\begin{array}{l}
I_{BC}= \Tr_R X |0\rangle \langle 0|\otimes I_{BC}  \\
\frac{1}{2}\Tr (I_A \otimes X ) 
((|0\rangle \langle j'|+|j'\rangle \langle 0|) \otimes F_j \otimes I_C)
(I_{RB}\otimes \rho_{AC})
=\delta_{j,j'}
\end{array}
\right.\right\}.\notag
\end{align}
For $k=5$, we have
\begin{align}
&
S_5
[\Tr_A (T\otimes I_C) (I_B \otimes \rho_{AC}),
(\Tr_A (F_j\otimes I_C)( I_B \otimes \rho_{AC}))_j ] \notag \\
=& 
 \min_{X \in {\cal S}_{BC}^4}
\left\{\Tr ( G \otimes T\otimes I_C) (I_{RB} \otimes \rho_{AC})
(I_A \otimes X) \left|
\begin{array}{l}
I_{BC}= \Tr_R X |0\rangle \langle 0|\otimes I_{BC}  \\
\frac{1}{2}\Tr (I_A \otimes X ) 
((|0\rangle \langle j'|+|j'\rangle \langle 0|) \otimes F_j \otimes I_C)
(I_{RB}\otimes \rho_{AC})
=\delta_{j,j'} \\
\Tr (I_A \otimes X ) 
((|i\rangle \langle j'|-|j'\rangle \langle i|) \otimes T \otimes I_C)
(I_{RB}\otimes \rho_{AC})
=0
\end{array}
\right.\right\}.\notag
\end{align} 

For $k=1,2,3,4$, we have
\eqref{MAIN1} as
\begin{align}
&\min_{\rho_{AC} \in {\cal S}({\cal H}_A\otimes {\cal H}_C)}
S_k
[\Tr_A (T\otimes I_C) (I_B \otimes \rho),(\Tr_A (F_j\otimes I_C)( I_B \otimes \rho))_j ]\notag\\
=&\min_{\rho_{AC} \in {\cal S}({\cal H}_A\otimes {\cal H}_C)}
 \min_{X \in {\cal S}_{BC}^k}
\left\{\Tr ( G \otimes T\otimes I_C) (I_{RB} \otimes \rho_{AC})
(I_A \otimes X) \left|
\begin{array}{l}
I_{BC}= \Tr_R X |0\rangle \langle 0|\otimes I_{BC} \\
\frac{1}{2}\Tr (I_A \otimes X ) 
((|0\rangle \langle j'|+|j'\rangle \langle 0|) \otimes F_j \otimes I_C)
(I_{RB}\otimes \rho_{AC})
=\delta_{j,j'}
\end{array}
\right.\right\}
\notag\\
=&\min_{\rho_{AC} \in {\cal S}({\cal H}_A\otimes {\cal H}_C)}
 \min_{X \in {\cal S}_{BC}^k}
\left\{\Tr ( G \otimes T) 
\Tr_{C}[ (I_{RB} \otimes \rho_{AC}) (I_A \otimes X) ]
\left|
\begin{array}{l}
I_{BC}= \Tr_R X |0\rangle \langle 0|\otimes I_{BC}  \\
\frac{1}{2}\Tr 
((|0\rangle \langle j'|+|j'\rangle \langle 0|) \otimes F_j )
\Tr_C[(I_{RB}\otimes \rho_{AC})(I_A \otimes X )]
=\delta_{j,j'}
\end{array}
\right.\right\}
\Label{E35B}
\\
\stackrel{(a)}{\ge}
&\min_{\rho_{AC} \in {\cal S}({\cal H}_A\otimes {\cal H}_C)}
 \min_{X \in {\cal S}_{BC}^k}
\left\{
\Tr ( G \otimes T) \Tr_{C}[ (I_{RB} \otimes \rho_{AC}) (I_A \otimes X) ]
\left|
\begin{array}{l}
\Tr_R \Tr_C[(I_{RB}\otimes \rho_{AC})(I_A \otimes X )]
(|0\rangle \langle 0|\otimes I_{AB})
=(\Tr_C\rho_{AC})\otimes I_B
\\
\frac{1}{2}\Tr 
((|0\rangle \langle j'|+|j'\rangle \langle 0|) \otimes F_j )
\Tr_C[(I_{RB}\otimes \rho_{AC})(I_A \otimes X )]
=\delta_{j,j'}
\end{array}
\right.\right\}
\Label{E35A}
\\
\stackrel{(b)}{\ge} &
 \min_{Y \in \mathcal S^k_{BA}}
\{\Tr Y G \otimes T|
\hbox{(i), (ii) hold} \}=J_k.
\Label{E35}
\end{align} 
Here, step $(a)$ is shown as follows.
The condition 
$I_{BC}= \Tr_R X |0\rangle \langle 0|\otimes I_{BC}  $ implies
the condition
\begin{align}
&\Tr_R \Tr_C[(I_{RB}\otimes \rho_{AC})(I_A \otimes X )]
(|0\rangle \langle 0|\otimes I_{AB})
=\Tr_{R C}[(I_{RB}\otimes \rho_{AC})(I_A \otimes X )
(|0\rangle \langle 0|\otimes I_{ABC})]\notag\\
=&\Tr_{C}(I_{B}\otimes \rho_{AC})
\Tr_{R}[(I_A \otimes X ) (|0\rangle \langle 0|\otimes I_{ABC})]\notag\\
=&\Tr_{C}(I_{B}\otimes \rho_{AC})
(I_A \otimes \Tr_{R} [X  (|0\rangle \langle 0|\otimes I_{BC})])
=\Tr_{C}(I_{B}\otimes \rho_{AC})
(I_A \otimes I_{BC})
=(\Tr_C\rho_{AC})\otimes I_B.\notag
\end{align} 
Hence, 
a pair $(\rho_{AC},X)$ in \eqref{E35B}
satisfies the condition in \eqref{E35A}.
Thus, for a pair $(\rho_{AC},X)$ in \eqref{E35B}, we have
$\Tr ( G \otimes T) \Tr_{C}[ (I_{RB} \otimes \rho_{AC}) (I_A \otimes X) ]
\ge \eqref{E35A}$, which implies $(a)$.

Step $(b)$ is shown as follows.
For a pair $(\rho_{AC},X)$ in \eqref{E35A}, 
we choose $Y$ 
to be $\Tr_C[(I_{RB}\otimes \rho_{AC})(I_A \otimes X )]$.
Since
\begin{align}
\Tr_R \Tr_C[(I_{RB}\otimes \rho_{AC})(I_A \otimes X )]
(|0\rangle \langle 0|\otimes I_{AB})
=(\Tr_C\rho_{AC})\otimes I_B\Label{NBVF},
\end{align}
$\rho_A(Y)$ is calculated to be
$\Tr_C\rho_{AC}$.
Then, the condition \eqref{NBVF} implies the condition (i).
The condition $
\frac{1}{2}\Tr 
((|0\rangle \langle j'|+|j'\rangle \langle 0|) \otimes F_j )
\Tr_C[(I_{RB}\otimes \rho_{AC})(I_A \otimes X )]
=\delta_{j,j'}$ implies the condition (ii).
Also, we have
$\Tr ( G \otimes T) \Tr_{C}[ (I_{RB} \otimes \rho_{AC}) (I_A \otimes X) ]
\ge \eqref{E35}$.
The remaining issue is to show $Y \in \mathcal S^k_{BA}$
for $k=1,2,3,4$.
For $k=4$, it is sufficient to show $Y \ge 0$ and $Y \in {\cal B}_{ABC}''$.
Since $X \in {\cal B}_{BC}''$ and $X \ge 0$,
we have
 $\sqrt{I_{RB}\otimes \rho_{AC}}(I_A \otimes X )\sqrt{I_{RB}\otimes \rho_{AC}} \ge 0$ belongs to 
${\cal B}_{ABC}''$.
 Taking the trace on $C$, we obtain $Y \ge 0$ and $Y \in {\cal B}_{ABC}''$.
For other cases, we need to show additional conditions.
For $k=1$, we need to show that $Y \ge 0$
has a separable form with the bipartition $R$ and $AB$.
Since
$\sqrt{I_{RB}\otimes \rho_{AC}}(I_A \otimes X )\sqrt{I_{RB}\otimes \rho_{AC}}$ has a separable form with the bipartition $R$ and $ABC$, $Y $ satisfies this condition.
For $k=2$, we need to show that $Y \in {\cal B}_{AB}$.
Since $X \in {\cal B}_{BC}$ implies
$\sqrt{I_{RB}\otimes \rho_{AC}}(I_A \otimes X )\sqrt{I_{RB}\otimes \rho_{AC}}
\in {\cal B}_{ABC}$, $Y $ satisfies this condition.
For $k=3$, we need to show the following.
When we apply the partial transpose on the system $R$ to $Y$, it is positive semi-definite. 
Since $\sqrt{I_{RB}\otimes \rho_{AC}}(I_A \otimes X )\sqrt{I_{RB}\otimes \rho_{AC}}$ satisfies this property,
$Y $ satisfies this condition.
Therefore, we obtain $(b)$.

For $k=5$, we have \eqref{MAIN1} as
\begin{align}
&\min_{\rho_{AC} \in {\cal S}({\cal H}_A\otimes {\cal H}_C)}
S_5
[\Tr_A (T\otimes I_C) (I_B \otimes \rho),(\Tr_A (F_j\otimes I_C)( I_B \otimes \rho))_j ]\notag\\
=&\min_{\rho_{AC} \in {\cal S}({\cal H}_A\otimes {\cal H}_C)}
 \min_{X \in {\cal S}_{BC}^4}
\left\{\Tr ( G \otimes T\otimes I_C) (I_{RB} \otimes \rho_{AC})
(I_A \otimes X) \left|
\begin{array}{l}
I_{BC}= \Tr_R X |0\rangle \langle 0|\otimes I_{BC}  \\
\frac{1}{2}\Tr (I_A \otimes X ) 
((|0\rangle \langle j'|+|j'\rangle \langle 0|) \otimes F_j \otimes I_C)
(I_{RB}\otimes \rho_{AC})
=\delta_{j,j'} \\
\Tr (I_A \otimes X ) 
((|i\rangle \langle j'|-|j'\rangle \langle i|) \otimes T \otimes I_C)
(I_{RB}\otimes \rho_{AC})
=0
\end{array}
\right.\right\}
\notag\\
=&\min_{\rho_{AC} \in {\cal S}({\cal H}_A\otimes {\cal H}_C)}
 \min_{X \in {\cal S}_{BC}^4}
\left\{\Tr ( G \otimes T) 
\Tr_{C}[ (I_{RB} \otimes \rho_{AC}) (I_A \otimes X) ]
\left|
\begin{array}{l}
I_{BC}= \Tr_R X |0\rangle \langle 0|\otimes I_{BC}  \\
\frac{1}{2}\Tr 
((|0\rangle \langle j'|+|j'\rangle \langle 0|) \otimes F_j)
\Tr_C[(I_{RB}\otimes \rho_{AC})(I_A \otimes X )]
=\delta_{j,j'}\\
\Tr 
((|i\rangle \langle j'|-|j'\rangle \langle i|) \otimes T )
\Tr_C[(I_{RB}\otimes \rho_{AC})(I_A \otimes X )]
=0
\end{array}
\right.\right\}\Label{ZDY}
\\
\stackrel{(a)}{\ge}
&\min_{\rho_{AC} \in {\cal S}({\cal H}_A\otimes {\cal H}_C)}
 \min_{X \in {\cal S}_{BC}^4}
\left\{\Tr ( G \otimes T) 
\Tr_{C}[ (I_{RB} \otimes \rho_{AC}) (I_A \otimes X) ]
\left|
\begin{array}{l}
\Tr_R \Tr_C[(I_{RB}\otimes \rho_{AC})(I_A \otimes X )]
(|0\rangle \langle 0|\otimes I_{B})
=(\Tr_C\rho_AC)\otimes I_B
\\
\frac{1}{2}\Tr 
((|0\rangle \langle j'|+|j'\rangle \langle 0|) \otimes F_j )
\Tr_C[(I_{RB}\otimes \rho_{AC})(I_A \otimes X )]
=\delta_{j,j'}
\\
\Tr 
((|i\rangle \langle j'|-|j'\rangle \langle i|) \otimes T )
\Tr_C[(I_{RB}\otimes \rho_{AC})(I_A \otimes X )]
=0
\end{array}
\right.\right\}
\notag\\
\stackrel{(b)}{\ge} &
 \min_{Y \in \mathcal S^5_{BA}(T)}
\{\Tr Y G \otimes T|
\hbox{(i), (ii) hold} \}=J_5.
\end{align} 
Here, steps $(a)$ and $(b)$ can be shown in the same way as
\eqref{E35}.

\subsection{Proof of Lemma \ref{L2}}
Since the proof of Lemma \ref{L2} is easier than
the proof of Lemma \ref{L3}, we show Lemma \ref{L2} here.
For $k=1,2,3,4$,
given $Y \in \mathcal S^k_{BA}$ satisfying the conditions (i), (ii)
and $\epsilon>0$, we define 
$Y_\epsilon' \in \mathcal S^k_{BA}$ as
\begin{align}
Y_\epsilon':= 
(( \sqrt{1-\epsilon} |0\rangle \langle 0|
+\sqrt{1-\epsilon}^{-1}(I_R-|0\rangle \langle 0|))\otimes I_{AB})
Y
(( \sqrt{1-\epsilon} |0\rangle \langle 0|
+\sqrt{1-\epsilon}^{-1}(I_R-|0\rangle \langle 0|))\otimes I_{AB}).
\Label{XBO}
\end{align}
Then, we define 
$Y_\epsilon \in \mathcal S^k_{BA}$ as
\begin{align}
Y_\epsilon:= Y_\epsilon'
+\frac{\epsilon}{d_A}|0\rangle \langle 0|\otimes I_{AB}.\notag
\end{align}
Here, we relax the condition (i) as follows.
\begin{description}
\item[(i'')]
There exists a positive semi-definite operator $\rho_A(Y)$ on ${\cal H}_A$ such that 
the relation \begin{align}
\Tr_{R} Y (|0\rangle \langle 0|\otimes I_{AB})
=
I_B \otimes \rho_A(Y)\Label{NMZA8}
\end{align}
holds as operators on $\cH_A\otimes \cH_B$.
Here, we drop the condition $\Tr \rho_A(Y)=1$.
\end{description}

Since $Y_\epsilon'$ 
and
$\frac{1}{d_A}|0\rangle \langle 0|\otimes I_{AB}$
satisfy the condition (i''),
$Y_\epsilon$ also satisfies the condition (i'').
We can define $\rho_A(Y_\epsilon')$,
$\rho_A(\frac{1}{d_A}|0\rangle \langle 0|\otimes I_{AB})$,
and $\rho_A(Y_\epsilon)$.
Since we have
$\rho_A(Y_\epsilon')= (1-\epsilon)\rho_A(Y)$
and $\rho_A(\frac{1}{d_A}|0\rangle \langle 0|\otimes I_{AB})
=\frac{1}{d_A}I_A$,
we have
$\rho_A(Y_\epsilon)= (1-\epsilon)\rho_A(Y)
+\frac{\epsilon}{d_A}I_A>0$.

Also, due to the definition \eqref{XBO},
for $j>0$, 
$|0\rangle \langle j|$ component of 
$Y_\epsilon'$ and
$|j\rangle \langle 0|$ component of 
$Y_\epsilon'$ are the same as 
$|0\rangle \langle j|$ component of 
$Y$ and
$|j\rangle \langle 0|$ component of 
$Y$, respectively.
Thus,
$|0\rangle \langle j|$ component of 
$Y_\epsilon$ and
$|j\rangle \langle 0|$ component of 
$Y_\epsilon$ are the same as 
$|0\rangle \langle j|$ component of 
$Y$ and
$|j\rangle \langle 0|$ component of 
$Y$, respectively.

Therefore, $Y_\epsilon$
satisfies (i), (ii)
and, $\rho_A(Y_\epsilon)>0$.
Also, we have
$\lim_{\epsilon\to 0}
\Tr Y_\epsilon (G \otimes T)=
\Tr Y (G \otimes T)$.
Therefore, we obtain \eqref{MMFB}.

For $k=5$, in the same way as \eqref{MMFB}, we can show
\eqref{MMFB5}.

\subsection{Proof of Lemma \ref{L3}}
For $k=1,2,3,4$,
we choose $Y \in \mathcal S^k_{BA}$ satisfying the conditions (i), (ii), and $\rho_A:=\rho_A(Y)>0$.
We diagonalize $\rho_A$ as
$\sum_{j=1}^{d_A} s_j |\phi_j\rangle \langle \phi_j|$.
Hence, we have $s_j >0$.

We choose a 
CONS (complete orthonormal system) 
$\{\psi_j\}$ of ${\cal H}_C$.
We define a unitary map
$U:\phi_j \mapsto \psi_j $ from $\cH_A$ to $\cH_C$.
We define the matrix
$\rho_A^{-1/2}:= \sum_{j=1}^{d_A} s_j^{-1/2} |\phi_j\rangle \langle \phi_j|$.
In a way similar to the discussion after \eqref{NBVF},
we can show that 
$X:= 
(U\otimes I_{RB}) (\rho_A^{-1/2}\otimes I_{RB}) Y(\rho_A^{-1/2} \otimes I_{RB} )(U^\dagger \otimes I_{RB}) $ belongs to 
${\cal S}_{BC}^k $ for $k=1,2,3,4$.
Notice that both sides in this definition
act on
${\cal H}_{R}\otimes {\cal H}_{B}\otimes {\cal H}_{C}$
because $U$ maps $\cH_A$ to $\cH_C$.

We choose $\rho_{AC}$ as the pure state
$\sum_{j=1}^{d_A}\sqrt{s_j}|\phi_j,\psi_j\rangle$, which is a purification of 
$\rho_A$.

Then,
\begin{align}
&\Tr_R X (|0\rangle \langle 0|\otimes I_{BC})
=\Tr_R 
(U\otimes I_{RB}) (\rho_A^{-1/2}\otimes I_{RB}) Y(\rho_A^{-1/2} \otimes I_{RB} )(U^\dagger \otimes I_{RB}) 
 |0\rangle \langle 0|\otimes I_{BC} \notag\\
=& 
(U\otimes I_{B}) (\rho_A^{-1/2}\otimes I_{B}) 
\Tr_R [Y
 |0\rangle \langle 0|\otimes I_{BA}]
(\rho_A^{-1/2} \otimes I_{B} )(U^\dagger \otimes I_{B}) \notag\\
= &
(U\otimes I_{B}) (\rho_A^{-1/2}\otimes I_{B}) 
I_{B}\otimes \rho_A
(\rho_A^{-1/2} \otimes I_{B} )(U^\dagger \otimes I_{B}) \notag\\
= &
(U\otimes I_{B})
I_{BA}(U^\dagger \otimes I_{B}) =I_{BC}.\Label{NMJ1}
\end{align}

For a matrix $Z$ on ${\cal H}_A$, we have
\begin{align}
\Tr_C[ \rho_{AC}(I_A \otimes 
(U \rho_A^{-1/2} Z(\rho_A^{-1/2} U^\dagger))
 ] =Z\Label{MXAH}.
\end{align}
This can be shown as follows.
\begin{align}
&\Tr_C[ \rho_{AC}(I_A \otimes 
(U \rho_A^{-1/2} Z(\rho_A^{-1/2} U^\dagger))
 ] \notag\\
=&\Tr_C\Big[ 
\sum_{i=1}^{d_A}\sqrt{s_i}|\phi_i,\psi_i\rangle
\sum_{i'=1}^{d_A}\sqrt{s_{i'}}\langle \phi_{i'},\psi_{i'}|
\Big(I_A \otimes 
\Big(\sum_{j=1}^{d_A} s_j^{-1/2}|\psi_j\rangle \langle \phi_j|
Z
\sum_{j'=1}^{d_A} s_{j'}^{-1/2}|\psi_{j'}\rangle \langle \phi_{j'}|\Big)
\Big)\Big] \notag\\
= &
\Big(\sum_{j=1}^{d_A'} |\phi_j\rangle \langle \phi_j|
Z
\sum_{j'=1}^{d_A} |\phi_{j'}\rangle \langle \phi_{j'}|\Big)
 =Z.
\end{align}

We write the matrix $Y$ as
$\sum_{j,j',b,b'} |j,b\rangle \langle j',b'|  \otimes Y_{j,b,j',b'} $
by using matrices $Y_{j,b,j',b'}$ on ${\cal H}_A$.
Here, $\{|b\rangle\}$ is a CONS of ${\cal H}_B$.
Applying \eqref{MXAH} to the matrix components 
$Y_{j,b,j',b'}$,
we have
\begin{align}
&\Tr_C[(I_{RB}\otimes \rho_{AC})(I_A \otimes X )]\notag\\
\stackrel{(a)}{=} &
\Tr_C \big[(I_{RB}\otimes \rho_{AC})\big(I_A \otimes 
\big(
(U\otimes I_{RB}) (\rho_A^{-1/2}\otimes I_{RB}) Y(\rho_A^{-1/2} \otimes I_{RB} )(U^\dagger \otimes I_{RB}) 
\big) \big)
 \big] \notag\\
=&
\Tr_C \big[(I_{RB}\otimes \rho_{AC})
\big(I_A \otimes \big(
(U\otimes I_{RB}) (\rho_A^{-1/2}\otimes I_{RB}) 
(\sum_{j,j',b,b'} |j,b\rangle \langle j',b'|  \otimes Y_{j,b,j',b'}  )
(\rho_A^{-1/2} \otimes I_{RB} )(U^\dagger \otimes I_{RB}) \big)\big)
 \big] \notag\\
=&
\sum_{j,j',b,b'} |j,b\rangle \langle j',b'|  \otimes 
\Tr_C\big[\rho_{AC} \big(I_A \otimes 
\big(U \rho_A^{-1/2} Y_{j,b,j',b'}
\rho_A^{-1/2} U^\dagger \big)
\big) \big] \notag\\
\stackrel{(b)}{=} &
\sum_{j,j',b,b'} |j,b\rangle \langle j',b'|  \otimes Y_{j,b,j',b'} 
=Y,\Label{NBT}
\end{align}
where $(a)$ follows from the definition of $X$,
and 
$(b)$ follows from \eqref{MXAH}.

Thus, we have
\begin{align}
\frac{1}{2}\Tr 
((|0\rangle \langle j'|+|j'\rangle \langle 0|) \otimes F_j )
\Tr_C[(I_{RB}\otimes \rho_{AC})(I_A \otimes X )]
\stackrel{(a)}{=} 
\frac{1}{2}\Tr 
((|0\rangle \langle j'|+|j'\rangle \langle 0|) \otimes F_j )
Y
\stackrel{(b)}{=} \delta_{j,j'},\Label{NMJ2}
\end{align}
where
$(a)$ follows from \eqref{NBT}, and
$(b)$ follows from the fact that 
$Y$ satisfies the condition (ii).

Therefore, 
the relations \eqref{NMJ1}, \eqref{NBT}, and \eqref{NMJ2}
guarantee that
$X\in {\cal S}_{BC}^k $ satisfies
the conditions \eqref{NB2}, \eqref{NB3}, and 
\eqref{MXI}, respectively.
Thus,
for $k=1,2,3,4$,
we have
\begin{align}
&\Tr Y (G \otimes T) \notag\\
\stackrel{(a)}{\ge } &
\min_{\rho_{AC} \in {\cal S}({\cal H}_A\otimes {\cal H}_C)}
 \min_{X \in {\cal S}_{BC}^k}
\left\{\Tr ( G \otimes T) 
\Tr_{C}[ (I_{RB} \otimes \rho_{AC}) (I_A \otimes X) ]
\left|
\begin{array}{l}
 \Tr_R X (|0\rangle \langle 0|\otimes I_{BC}) =I_{BC} \\
\frac{1}{2}\Tr 
((|0\rangle \langle j'|+|j'\rangle \langle 0|) \otimes F_j \otimes I_C)
\Tr_C[(I_{RB}\otimes \rho_{AC})(I_A \otimes X )]
=\delta_{j,j'}
\end{array}
\right.\right\} \notag\\
\stackrel{(b)}{=} &
\min_{\rho_{AC} \in {\cal S}({\cal H}_A\otimes {\cal H}_C)}
S_k
[\Tr_A (T\otimes I_C) (I_B \otimes \rho),(\Tr_A (F_j\otimes I_C)( I_B \otimes \rho))_j ],
\end{align}
where 
$(a)$ follows from 
the fact that
$X\in {\cal S}_{BC}^k $ satisfies
the conditions \eqref{NB2}, \eqref{NB3}, and \eqref{MXI},
and $(b)$ follows from \eqref{E35B}.
Therefore, we obtain \eqref{ZIT}.

For $k=5$, we choose $Y \in \mathcal S^5_{BA}(T)$ satisfying the conditions (i), (ii) and $ \rho_A(Y)>0$.
We choose $\rho$ and $X$ in the same way as the above.
Since $Y \in \mathcal S^5_{BA}(T)$ satisfies the conditions (i), (ii),
the relations \eqref{NMJ1}, \eqref{NMJ2} hold.
Since $Y \in \mathcal S^5_{BA}(T)$, in the same way as \eqref{NMJ2}, we have
\begin{align}
 \Tr 
((|i\rangle \langle j'|-|j'\rangle \langle i|) \otimes T \otimes I_C)
\Tr_C[(I_{RB}\otimes \rho_{AC})(I_A \otimes X )]
=
 \Tr 
((|i\rangle \langle j'|-|j'\rangle \langle i|) \otimes T )
Y
=0 .\Label{NMJ3}
\end{align}

Therefore, the relations \eqref{NMJ1}, \eqref{NBT}, \eqref{NMJ2}, and \eqref{NMJ3} guarantee that
$X\in {\cal S}_{BC}^4 $
satisfies the conditions \eqref{NB2}, \eqref{NB3},
\eqref{MXI2}, and \eqref{MXI}, respectively.
Thus,
we have
\begin{align}
&\Tr Y (G \otimes T) \notag\\
\stackrel{(a)}{\ge} &
\min_{\rho_{AC} \in {\cal S}({\cal H}_A\otimes {\cal H}_C)}
 \min_{X \in {\cal S}_{BC}^4}
\left\{\Tr ( G \otimes T) 
\Tr_{C}[ (I_{RB} \otimes \rho_{AC}) (I_A \otimes X) ]
\left|
\begin{array}{l}
 \Tr_R X (|0\rangle \langle 0|\otimes I_{BC}) =I_{BC} \\
\frac{1}{2}\Tr 
((|0\rangle \langle j'|+|j'\rangle \langle 0|) \otimes F_j \otimes I_C)
\Tr_C[(I_{RB}\otimes \rho_{AC})(I_A \otimes X )]
=\delta_{j,j'}\\
\Tr 
((|i\rangle \langle j'|-|j'\rangle \langle i|) \otimes T )
\Tr_C[(I_{RB}\otimes \rho_{AC})(I_A \otimes X )]
=0.
\end{array}
\right.\right\}\notag\\
\stackrel{(b)}{=} &
\min_{\rho_{AC} \in {\cal S}({\cal H}_A\otimes {\cal H}_C)}
S_5
[\Tr_A (T\otimes I_C) (I_B \otimes \rho),(\Tr_A (F_j\otimes I_C)( I_B \otimes \rho))_j ],\notag
\end{align}
where $(a)$ follows from 
the fact that
$X\in {\cal S}_{BC}^k $ satisfies
the conditions \eqref{NB2}, \eqref{NB3}, \eqref{MXI2}, and \eqref{MXI},
and $(b)$
follows from \eqref{ZDY}.
Therefore, we obtain \eqref{ZIT}.


\section{Proofs of Theorems \ref{NNT} and \ref{NNT2}} 
\subsection{Proof of Theorem \ref{NNT}}\Label{B1}
We consider the dual problem for the conic programing   
\eqref{NVF1} and \eqref{NVF2}.
For $(A,B,W)\in \mathbb{R}^{d \times d} \times
{\cal T}_{sa}(\cH_{AB})^d\times 
{\cal T}_{sa}(\cH_{AB})$, 
we define
\begin{align}
&\Pi (A,B,W)\nonumber \\
:=&  G \otimes T_N
-\frac{1}{2} \Big( \sum_{1 \le i,j\le d} A_i^j 
(|0\rangle \langle i|+|i\rangle \langle 0|)\otimes F_{j,N}\Big)
+\sqrt{-1}\sum_{i=1}^d(|0\rangle \langle i|-|i\rangle \langle 0|)\otimes B_{i}
-|0\rangle \langle 0| \otimes W.\notag
\end{align}

We have
\begin{align}
J_4=\max_{(A,B,W)\in \mathbb{R}^{d \times d} \times
{\cal T}_{sa}(\cH_{AB})^d\times 
({\cal T}_{sa}(\cH_{AB})\cap {\cal K})
}
\Big\{ \Tr W+\sum_{j=1}A_j^j \Big|  \Pi (A,B,W) \ge 0
\Big\}.\notag
\end{align}
Here, 
the matrix $A \in \mathbb{R}^{d \times d}$ corresponds to the condition (ii-N) in the primal problem.
The matrices $(B^1, \ldots, B^d) \in
{\cal T}_{sa}(\cH_{AB})^d$ correspond to the condition 
$X_{k,0}\in {\cal B}_{sa}(\cH_{A,B})$ for $k=1, \ldots, d$
in the primal problem.
This condition is rewritten as $X_{k,0}^\dagger=X_{k,0}$.

The condition (i'-N) is composed of \eqref{NBFY2-N}, \eqref{NBFY-N}, and \eqref{ZKT-N}.
We denote the variable corresponding to the condition
\eqref{NBFY2-N} with $b \in \{1,\ldots, d_B-1\}$ and
$b' \in \{2,\ldots, d_B-1\}$ with $b>b'$
by the matrix $W_{b,b'}$ on $\cH_A$.
We denote the variable corresponding to the condition
\eqref{NBFY-N} with $b \in \{1,\ldots, d_B-1\}$
by the Hermitian matrix $W_{b}$.
We denote the variable corresponding to the condition
\eqref{ZKT-N} 
by the real number $w$.
Then, we choose the Hermitian matrix $W$ on 
$\cH_A\otimes \cH_B$ as
\begin{align}
W= \sum_{b>b'} 
W_{b,b'} \otimes |b\rangle \langle b'|
+
W_{b,b'}^\dagger \otimes |b'\rangle \langle b|
+\sum_{b=1}^{d_B-1}
W_{b}\otimes 
(|b\rangle \langle b|-|b+1\rangle \langle b+1|)
+wI \otimes |1\rangle \langle 1|.\notag
\end{align}
The above Hermitian matrix on 
$\cH_A\otimes \cH_B$ belongs to the subset 
${\cal K}$,
and any element of 
${\cal K}$ can be written as the above form.
Thus,
The matrix $W \in ({\cal T}_{sa}(\cH_{AB})\cap {\cal K})
$ corresponds to 
the condition (i'-N) in the primal problem.

On the other hand, 
we have
\begin{align}
&S_4[T_{N},(F_{j,N})_j]\nonumber\\
=&\max_{(A,B,W)\in \mathbb{R}^{d \times d} \times
{\cal T}_{sa}(\cH_{AB})^d\times 
{\cal T}_{sa}(\cH_{AB})
}
\Big\{ \Tr W+\sum_{j=1}A_j^j \Big|  \Pi (A,B,W) \ge 0
\Big\}.\Label{MAX1}
\end{align}
In $S_4$, the condition for $W$ is relaxed as
$W \in ({\cal T}_{sa}(\cH_{AB})\cap {\cal K})$ 
because 
the condition (i'-N) is changed to $
\Tr_{R}[ Y_N (|0\rangle \langle 0|\otimes I_{AB})]
=I_{AB}$.

Therefore, 
if and only if  the maximizer $(A_*,B_*,W_*)$ in \eqref{MAX1}
satisfies the condition $W_* \in {\cal K}$, we have
\begin{align}
S_4[T_{N},(F_{j,N})_j]=J_4.\Label{NZI}
\end{align}

\begin{lemma}\Label{LE1}
(1) The minimization \eqref{MIN1} 
is attained when 
$Y^{i,j}_N=X_*^{i,j}:=L^i_* L^j_* $, where $L^0_*$ is defined to be $I_{AB}$.
(2) The maximization \eqref{MAX1} 
is attained when 
$A=A_*:=-{G} J_{\rm SLD}^{-1}$,
$B_j=B_{*,j}:=-\frac{\sqrt{-1}}{2}\sum_{i=1}^d( {G} J_{\rm SLD}^{-1})_j^i [L_{*,i}, T_N ]$
and $W$ is $-W_{\rm SLD}(T_N,\vec{L}_*)$.
\end{lemma}
The combination of \eqref{NZI} and Lemma \ref{LE1}
implies Theorem \ref{NNT}.

\begin{proof}
We diagonalize the matrx $G$ as
$G_{i,j}=\sum_{j'=1}^d g_{j'} v_{j',i} v_{j',j}$,
where $v_{j',j}$ is an orthogonal matrix.
Then, we have
\begin{align}
& \Pi( A_*,\vec{B}_*, -W_{\rm SLD}(\vec{L}_*))\notag\\
=&
\sum_{1 \le i,j\le d}G_{i,j}  |i\rangle \langle j|\otimes T_N
-\sum_{j=1}^d G_{i,j} 
(|0\rangle \langle i| \otimes L_*^j T_N
+|i\rangle \langle 0| \otimes T_N L_*^j)
+\sum_{1 \le i,j \le d} G_{i,j}|i\rangle \langle j|\otimes L_*^i T_N L_*^j\notag\\
=&
\sum_{1 \le i,j\le d}G_{i,j}  |i\rangle \langle j|\otimes T_N
-\sum_{j=1}^d G_{i,j} 
(|0\rangle \langle j| \otimes L_*^i T_N
+|i\rangle \langle 0| \otimes T_N L_*^j)
+\sum_{1 \le i,j \le d} G_{i,j}|i\rangle \langle j|\otimes L_*^i T_N L_*^j\notag\\
=&\sum_{1 \le i,j\le d}G_{i,j}
\Big(
 (|0\rangle \langle 0| \otimes L_*^i-
|i\rangle \langle 0| \otimes I_{AB}\Big)
(|0\rangle \langle 0| \otimes T_N) 
\Big(
 (|0\rangle \langle 0| \otimes L_*^j-
|j\rangle \langle 0| \otimes I_{AB}\Big)^\dagger
\notag\\
=&
\sum_{j'=1}^d g_{j'}
\Big(\sum_{i=1}^d 
 v_{j',i}\Big(
 (|0\rangle \langle 0| \otimes L_*^i-
|i\rangle \langle 0| \otimes I_{AB}\Big)\Big)
(|0\rangle \langle 0| \otimes T_N) 
\Big(\sum_{j=1}^d 
 v_{j',j}\Big(
 (|0\rangle \langle 0| \otimes L_*^j-
|j\rangle \langle 0| \otimes I_{AB}\Big)\Big)^\dagger \notag\\
\ge & 0 \Label{BNJ1}.
\end{align}
Hence, the tuple $( A_*,\vec{B}_*, -W_{\rm SLD}(\vec{L}_*))$
belongs to the range of the
maximization given in the RHS of \eqref{MAX1}.

Since we have
\begin{align}
&\Tr X_* \Pi( A_*,\vec{B}_*, -W_{\rm SLD}(\vec{L}_*))\notag\\
=&
\Tr \Big[
\Big( |0\rangle \langle 0|\otimes I_{AB}+\sum_{j'=1}^d |j'\rangle \langle 0|\otimes L_*^{j'}
\Big)
\Big( |0\rangle \langle 0|\otimes I_{AB}+
\sum_{j'=1}^d |j'\rangle \langle 0|\otimes L_*^{j'} \Big)^\dagger \notag\\
&
\sum_{1 \le i,j\le d}G_{i,j}
\Big(
 (|0\rangle \langle 0| \otimes L_*^i-
|i\rangle \langle 0| \otimes I_{AB}\Big)
(|0\rangle \langle 0| \otimes T_N) 
\Big(
 (|0\rangle \langle 0| \otimes L_*^j-
|j\rangle \langle 0| \otimes I_{AB}\Big)^\dagger\Big]
\notag\\
=&
\sum_{1 \le i,j\le d}G_{i,j}
\Tr \Big[
\Big( |0\rangle \langle 0|\otimes I_{AB}+\sum_{j'=1}^d |0\rangle \langle j'|\otimes L_*^{j'}
\Big)
\Big(
 (|0\rangle \langle 0| \otimes L_*^i-
|i\rangle \langle 0| \otimes I_{AB}\Big)
(|0\rangle \langle 0| \otimes T_N) \notag\\
&\Big(
 (|0\rangle \langle 0| \otimes L_*^j-
|j\rangle \langle 0| \otimes I_{AB}\Big)^\dagger
\Big( |0\rangle \langle 0|\otimes I_{AB}+\sum_{j'=1}^d |j'\rangle \langle 0|\otimes L_*^{j}
\Big)\Big]
\notag\\
=&
\sum_{1 \le i,j\le d}G_{i,j}
\Tr 
\Big(
 (|0\rangle \langle 0| \otimes L_*^i-
|0\rangle \langle 0|  \otimes L_*^i\Big)
(|0\rangle \langle 0| \otimes T_N) 
\Big(
 (|0\rangle \langle 0| \otimes L_*^j-
|0\rangle \langle 0| \otimes L_*^i\Big)^\dagger
\notag\\
=&0,\Label{BNJ2}
\end{align}
we have
\begin{align}
\Tr [X_* (G \otimes T_N)]=
\Tr W_{\rm SLD}(\vec{L}_*)+\sum_{j=1}A_{*,j}^j.\notag
\end{align}
Since $X_*$ belongs to $\in {\cal S}_{BA}^4$,
we have
\begin{align}
\bar{S}_4[|\Phi\rangle\langle \Phi|]=S_4[T_{N},(F_{j,N})_j]
=
\Tr [X_* (G \otimes T_N)]=
\Tr W_{\rm SLD}(\vec{L}_*)+\sum_{j=1}A_{*,j}^j.\notag
\end{align}
Therefore, we obtain both desired statements.
\end{proof}

\subsection{Proof of Theorem \ref{NNT2}}
We denote the set of $d \times d$ anti-symmetric matrices
by $\mathbb{R}^{d \times d}_{AS}$.
For $(A,C,B,W)\in 
\mathbb{R}^{d \times d} \times
\mathbb{R}^{d \times d}_{AS} \times
{\cal T}_{sa}(\cH_{AB})^d\times 
{\cal T}_{sa}(\cH_{AB})$, 
we define
\begin{align}
\Pi (A,C,B,W):= \Pi (A,B,W) -\sqrt{-1} C \otimes T_N.\notag
\end{align}
Also, we have
\begin{align}
J_5=\max_{(A,C,B,W)\in \mathbb{R}^{d \times d} \times
\mathbb{R}^{d \times d}_{AS} \times
{\cal T}_{sa}(\cH_{AB})^d\times 
({\cal T}_{sa}(\cH_{AB})\cap {\cal K})
}
\Big\{ \Tr W+\sum_{j=1}A_j^j \Big|  \Pi (A,C,B,W) \ge 0
\Big\}.\notag
\end{align}
Here, 
the matrix $A \in \mathbb{R}^{d \times d}$,
the matrices $(B^1, \ldots, B^d) \in
{\cal T}_{sa}(\cH_{AB})^d$, and
the matrix $W \in ({\cal T}_{sa}(\cH_{AB})\cap {\cal K})
$ have the same meaning as in $J_4$.
The matrix $C \in \mathbb{R}^{d \times d}_{AS}$
corresponds to the condition \eqref{c3}
in the primal problem.

Also, we have
\begin{align}
&S_5[T_{N},(F_{j,N})_j] \notag\\
=&\max_{(A,C,B,W)\in \mathbb{R}^{d \times d} \times
\mathbb{R}^{d \times d}_{AS} \times
{\cal T}_{sa}(\cH_{AB})^d\times 
{\cal T}_{sa}(\cH_{AB})}
\Big\{ \Tr W+\sum_{j=1}A_j^j \Big|  \Pi (A,C,B,W) \ge 0
\Big\}.\Label{MAX2}
\end{align}
Then, if and only if the maximizer $(A_*,C_*,B_*,W_*)$ in \eqref{MAX2}
satisfies the condition $W_* \in {\cal K}$, we have
\begin{align}
S_5[T_{N},(F_{j,N})_j]=J_5.\Label{NZI2}
\end{align}

\begin{lemma}\Label{LE2}
(1) The minimization \eqref{MIN2} 
is attained when $X=X_*:= \Pi(\vec{Z}_*)
+(|\Im V_*|- ( \sqrt{-1}\Im V_*)\otimes I_{AB}$.

(2) $ (Z_*^i -\sqrt{-1} \sum_{j=1}^d C_{*}^{i,j} Z_*^j) \circ T_N$
is written as a linear sum of $F_{i',N}$ for $i=1, \ldots, d$.
That is, there exist real numbers $A_{*,i}^{i'}$ such that
$(Z_*^i -\sqrt{-1} \sum_{j=1}^d C_{*}^{i,j} Z_*^j) \circ T_N=
\sum_{i'=1}^d A_{*,i}^{i'}F_{i',N}$.

(3) The maximization \eqref{MAX2} 
is attained when 
$A=-A_*$,
$B_j=B_{*,j}:=-\frac{\sqrt{-1}}{2}\sum_{i=1}^d( A_*)_j^i [Z_{*,i}, T_N ]$,
$C=C_*$,
and $W=-W_{HN}(T_N,\vec{Z}_*)$.

\end{lemma}
The combination of \eqref{NZI2} and Lemma \ref{LE2}
implies Theorem \ref{NNT2}.

\begin{proof}
(1) is known.
We write the maximizer in \eqref{MAX2}
by $(A_\circ,C_\circ,B_\circ,W_\circ)$,
whose existence is guaranteed because of the finiteness of the dimension. 
Then, we have
\begin{align}
0=&\Tr (\Pi(\vec{Z}_*)+(|\Im V_*|- (\sqrt{-1} \Im V_*)\otimes I_{AB})  \Pi (A_\circ,C_\circ,B_\circ,W_\circ) \notag\\
=&\Tr \Pi(\vec{Z}_*)\Pi (A_\circ,C_\circ,B_\circ,W_\circ)
+\Tr (|\Im V_*|- (\sqrt{-1}\Im V_*))\otimes I_{AB} \Pi (A_\circ,C_\circ,B_\circ,W_\circ)\notag \\
\ge & \Tr (|\Im V_*|- (\sqrt{-1}\Im V_*))\otimes I_{AB} \Pi (A_\circ,C_\circ,B_\circ,W_\circ) \ge 0.\notag
\end{align}
Hence, we have
\begin{align}
0=&\Tr (|\Im V_*|- (\sqrt{-1}\Im V_*))\otimes I_{AB} \Pi (A_\circ,C_\circ,B_\circ,W_\circ) \notag\\
=&\Tr (|\Im V_*|- (\sqrt{-1}\Im V_*))\otimes I_{AB} (  I_R -\sqrt{-1} C_\circ) \otimes T_N 
=\Tr (|\Im V_*|- (\sqrt{-1}\Im V_*))(  I_R -\sqrt{-1} C_\circ) \notag\\ 
=&\Tr |\Im V_*| - \Tr  (\Im V_*) C_\circ.\notag
\end{align}
Hence, we have $C_\circ=C_*$.

Therefore, we have
\begin{align}
0=\Tr \Pi(\vec{Z}_*)\Pi (A_\circ,C_*,B_\circ,W_\circ),\notag
\end{align}
which implies 
\begin{align}
0=\Pi(\vec{Z}_*)\Pi (A_\circ,C_*,B_\circ,W_\circ)
=(\sum_{1=0}^d |i\rangle \otimes (Z_*^i)^\dagger)
(\sum_{1=0}^d \langle i| \otimes Z_*^i)
\Pi (A_\circ,C_*,B_\circ,W_\circ).\notag
\end{align}
Hence, we have
\begin{align}
0=
(\sum_{1=1}^d \langle i| \otimes Z^i)
\Pi (A_\circ,C_*,B_\circ,W_\circ).\notag
\end{align}

We choose the Hermitian matrix $R_i \in \cH_{AB}$
as $\sum_{j=1}^d A_{*,i}^j D_j+ \sqrt{-1}B_{\circ,i}$.
Hence, $R_i+R_i^\dagger$ can be written as a linear sum of
$D_1, \ldots, D_d$.
Hence, 
\begin{align}
\Pi (A_\circ,C_*,B_\circ,W_\circ)
=\sum_{1\le i,j\le d} (\delta^{i,j} -\sqrt{-1} C_*^{i,j}) 
|i\rangle \langle j| \otimes T_N
+ \sum_{j=1}^d  
(|0\rangle \langle j|\otimes R_j +|j\rangle \langle 0|\otimes R_j^\dagger)
-|0\rangle \langle 0|\otimes W_{\circ}.\notag
\end{align}
\begin{align}
\Pi(\vec{Z}_*)=\sum_{1\le i,j\le d}
|i\rangle \langle j|\notag
\end{align}
Hence,
\begin{align}
0=&\Big(\sum_{1=0}^d \langle i| \otimes Z_*^i\Big)
 \Pi (A_\circ,C_*,B_\circ,W_\circ)\\
=& \sum_{i,j} (\delta^{i,j} -\sqrt{-1} C_*^{i,j})  \langle j| \otimes 
 Z_*^i T_N
+ \sum_{j=1}^d \langle j|\otimes R_j 
+ \sum_{j=1}^d \langle 0|\otimes  Z_*^j R_j^\dagger)
- \langle 0|\otimes W_{\circ}.\notag
\end{align}
Thus,
\begin{align}
R_j& =- \sum_{i=1}^d (\delta^{i,j} -\sqrt{-1} C_*^{i,j})  \otimes 
 Z_*^i T_N \Label{MN4}\\
 W_{\circ}&=\sum_{j=1}^d  Z_*^j R_j^\dagger .\Label{MN5}
\end{align}
Thus, $\frac{1}{2}(R_i+R_i^\dagger)
=- \sum_{i=1}^d (\delta^{i,j} -\sqrt{-1} C_*^{i,j})  \otimes 
 Z_*^i \circ T_N$ can be written as a linear sum of
$D_1, \ldots, D_d$. Hence, the statement (2) is shown.
Hence, we have
$A=-A_*$ and
$B_j=B_{*,j}:=-\frac{\sqrt{-1}}{2}\sum_{i=1}^d( A_*)_j^i [Z_{*,i}, T_N ]$.
Combining 
\eqref{MN4} and \eqref{MN5}, we obtain 
$W_\circ=-W_{HN}(T_N,\vec{Z}_*)$.
\end{proof}
\end{widetext}

\section{Calculations for the field sensing problem}
\label{appendix:field-sensing}

\subsection{Restriction to the symmetric subspace}
 Because our numerical calculation can be restricted to the symmetric subspace, we calculate the projections of operators that we need to implement on the symmetric subspace. 
 We begin by considering the action of $E^1,E^2,E^3$ on the symmetric subspace.

Now let us denote $\Pi^{\rm sym} := \sum_{u=0}^n|D^n_u\>\<D^n_u|$.
We now show the following lemma.
 \begin{lemma}
\label{lem:Jkpi}
Let $r_w := \frac{1}{2} \sqrt{(n-w)(w+1)}$. Then
\begin{align}
\Pi^{\rm sym}{ E^3} \Pi^{\rm sym}&= \sum_{w=0}^{n} \left(\frac{n}{2} - w\right) |D^n_w\>\<D^n_w|, \label{E3 sym}\\
\Pi^{\rm sym} {E^1}\Pi^{\rm sym} &= \sum_{w=0}^{n-1}r_w
 \left(|D^n_w\>\<D^n_{w+1}| + |D^n_{w+1}\>\<D^n_{w}| \right) ,  \label{E1 sym}\\
\Pi^{\rm sym}{E^2}\Pi^{\rm sym} &= \sum_{w=0}^{n-1}r_w
 \left(-i |D^n_w\>\<D^n_{w+1}| +i |D^n_{w+1}\>\<D^n_{w}| \right). \label{E2 sym}
\end{align}
\end{lemma}

\begin{proof}[Proof of Lemma \ref{lem:Jkpi}]
From 
\cite[Lemma 6]{ouyang_RQMWSS} we know that 
\begin{align}
\<D^n_{w} |  {\tau_{3}^{(1)}} |  D^n_w\>  
&= \frac{1}{ \binom n w }
 K^{0}_{0}(0) K^{n}_{w}(1) , \notag\\
\<D^n_{w+1} | {\tau_{1}^{(1)}} |  D^n_w\>  
&= \frac{1}{\sqrt{ \binom n w \binom n {w+1}  }}
 K^{1}_{0}(1) K^{n-1}_{w}(0),\notag \\
\<D^n_{w+1} |  {\tau_{2}^{(1)}} |  D^n_w\>  
&= \frac{i}{\sqrt{ \binom n w \binom n {w+1}  }}
 K^{1}_{1}(0) K^{n-1}_{w}(0),\notag
\end{align}
where
 \begin{align}
K^n_k(z) =  \sum_{j = 0}^z \binom {z} {j} \binom{n-z}{k-j} (-1)^j\notag
\end{align}
is a Krawtchouk polynomial.
Making substitutions for the values of the Krawtchouk polynomials, we get
\begin{align}
\<D^n_{w} | {\tau_{3}^{(1)}} |  D^n_w\> &=  1 - \frac{2w}{n}, \notag\\
\<D^n_{w+1} | {\tau_{1}^{(1)}} |  D^n_w\>  
&= \frac{\sqrt{(n-w)(w+1)}}{n}, \notag\\
\<D^n_{w+1} | {\tau_{2}^{(1)}} |  D^n_w\>  
&=i  \frac{\sqrt{(n-w)(w+1)}}{n}.\notag
\end{align}
Hence it follows that 
\begin{align}
\<D^n_{w} | {E^3} |  D^n_w\> &=  \frac{n}{2} -w, \\
\<D^n_{w+1} | {E^1} |  D^n_w\> &= \frac{1}{2} \sqrt{(n-w)(w+1)}, \label{tau1proj} \\
\<D^n_{w+1} | {E^2} |  D^n_w\>  &= \frac{i}{2} \sqrt{(n-w)(w+1)}\label{tau2proj}.
\end{align}
Hence, the result follows.
\end{proof}
 
Next, we calculate the projections of $A_{j,\gamma}$ onto the symmetric subspace.
\begin{lemma}
\label{lem:Ljpi}
For all $j=1,\dots, n$, we have 
\begin{align}
\Pi^{\rm sym} A_{j,\gamma} \Pi^{\rm sym} = \gamma \frac{1}{n} \sum_{w=0}^{n-1}   \sqrt{(n-w)(w+1)} |D^n_w\>\<D^n_{w+1}|  .\notag
\end{align}
\end{lemma}
\begin{proof}
 Now consider $|0\>\<1| \otimes I^{\otimes n-1} = 
 (\tau_{1}^{(1)}- i \tau_{2}^{(1)})/2$.
 Hence from \eqref{tau1proj} and \eqref{tau2proj} we have
\begin{align}
\<D^n_{w+1} |( |0\>\<1| \otimes I^{\otimes n-1} )|  D^n_w\>  &= 0.\notag
\end{align}
Also since $|1\>\<0| \otimes I^{\otimes n-1} = (\tau_{1}^{(1)} + i \tau_{2}^{(1)})/2$,
we use \eqref{tau1proj} and \eqref{tau2proj} to find that
\begin{align}
\<D^n_{w+1} |( |1\>\<0| \otimes I^{\otimes n-1}) |  D^n_w\>  &= \sqrt{(n-w)(w+1)}/n.\notag
\end{align}
Hence the result follows. 
\end{proof} 
Hence the projection of $\tilde A_\gamma$ onto the symmetric subspace is given by
\begin{align}
  \gamma  \sum_{w=0}^{n-1}   \sqrt{(n-w)(w+1)} |D^n_w\>\<D^n_{w+1}|  .\label{tildeAgam sym}
\end{align}
We use the projections of $E^1,E^2,E^3$ and $\tilde A_\gamma$ on the symmetric subspace in order to do our numerical calculations entirely within the symmetric subspace.

\subsection{SLD equation}
\label{app:sld eq}

As a preparation of our proofs of 
Lemmas \ref{L12} and \ref{lem:commutator perturbation},
here, we solve the SLD equation for a density matrix
$\rho$ that need not have full rank.
This appendix addresses a general theory for 
SLD equations. 
\begin{lemma}
\label{lem:sld-unitary}
Let $\rho$ be a Hermitian matrix with spectral decomposition
$\rho = \sum_{k}\lambda_k |\phi_k\>\<\phi_k|$, such that for all $k$, we have $\lambda_{k}\ge \lambda_{k+1}$.
When $\lambda_k = 0$ for all $k > p$, then we can write 
When $\rho = \sum_{k=1}^{p} \lambda_k |\phi_k\>\<\phi_k| $, and the solution $L$ to 
\begin{align}
\frac{1}{2}(\rho L + L \rho ) 
= i \rho E -i E \rho .
\label{sld equation}
\end{align}
is given by
\begin{align}
L =  2i \sum_{k,l} f(k,l)
\<\phi_k|E|\phi_l\> |\phi_k\>\<\phi_l|\notag
\end{align}
where 
\begin{align}
f(k,l) := 0&, \quad k,l > p , \notag\\
f(k,l) := -1&, \quad k > p, l \le p,\notag\\
f(k,l) := 1&, \quad l > p, k\le p,\notag\\
f(k,l) :=  \frac{\lambda_k-\lambda_l}{\lambda_k + \lambda_l} &, \quad k,l \le p.\notag
\end{align}
\end{lemma}
\begin{proof}[Proof of Lemma \ref{lem:sld-unitary}]
Note that 
\begin{align}
\rho L = &
 2i \sum_{k=1}^p \sum_l f(k,l)
\<\phi_k|E|\phi_l\>\lambda_k |\phi_k\>\<\phi_l|
\notag\\
=&
 2i \sum_{k=1}^p \sum_{l=1}^p \frac{\lambda_k - \lambda_l}{\lambda_k + \lambda_l}
\<\phi_k|E|\phi_l\>\lambda_k |\phi_k\>\<\phi_l|
\notag\\
&- 2i \sum_{k=1}^p \sum_{l>p}  
\<\phi_k|E|\phi_l\>\lambda_k |\phi_k\>\<\phi_l|,
\\
L \rho =&
2i \sum_{l=1}^p \sum_k f(k,l)
\<\phi_k|E|\phi_l\> |\phi_k\>\<\phi_l|\lambda_l\notag\\
=&
2i \sum_{l=1}^p \sum_{k=1}^p \frac{\lambda_k - \lambda_l}{\lambda_k + \lambda_l}
\<\phi_k|E|\phi_l\> |\phi_k\>\<\phi_l|\lambda_l\notag\\
&+ 2i \sum_{l=1}^p \sum_{k>p}  
\<\phi_k|E|\phi_l\> |\phi_k\>\<\phi_l| \lambda_l.\notag
\end{align}
We hence get
\begin{align}
&\frac{1}{2}(\rho L + L \rho ) \notag\\
=&
 i \sum_{k=1}^p \sum_{l=1}^p (\lambda_k - \lambda_l )
\<\phi_k|E|\phi_l\>\lambda_k |\phi_k\>\<\phi_l|
\notag\\
&- i \sum_{k=1}^p \sum_{l>p}  
\<\phi_k|E|\phi_l\>\lambda_k |\phi_k\>\<\phi_l|\notag\\
&+ i \sum_{l=1}^p \sum_{k>p}  
\<\phi_k|E|\phi_l\>\lambda_l |\phi_k\>\<\phi_l|.\notag
\end{align}
Now, 
\begin{align}
i E \rho 
&=i \sum_{k}\sum_{l=1}^p |\phi_k\>\<\phi_k| E |\phi_l\>\<\phi_l|\lambda_l\\
i\rho  E
&= i \sum_{l}\sum_{k=1}^p \lambda_k |\phi_k\>\<\phi_k| E |\phi_l\>\<\phi_l|,\notag
\end{align}
and hence
\begin{align}
-i E \rho +i\rho  E 
=&
-i \sum_{k}\sum_{l=1}^p |\phi_k\>\<\phi_k| E|\phi_l\>\<\phi_l|\lambda_l \notag \\
&+ i \sum_{l}\sum_{k=1}^p \lambda_k |\phi_k\>\<\phi_k| E |\phi_l\>\<\phi_l|,\notag\\
=&
i \sum_{k=1}^p\sum_{l=1}^p (\lambda_k - \lambda_l ) \<\phi_k| E|\phi_l\>
|\phi_k\>\<\phi_l|\notag\\
&
-i \sum_{k>p}\sum_{l=1}^p |\phi_k\>\<\phi_k| E |\phi_l\>\<\phi_l|\lambda_l\\
&+ i \sum_{l>p}\sum_{k=1}^p \lambda_k |\phi_k\>\<\phi_k| E |\phi_l\>\<\phi_l|.\notag
\end{align}
This proves the lemma.
\end{proof}
%

Next we give a lemma that evaluates $\tr (L_j \rho L_k )$. We prove this using Lemma \ref{lem:sld-unitary}.
\begin{lemma}\label{lem:eval trace LjrhoLk}
Let $\rho = \sum_{l} \lambda_l |\phi_l\>\<\phi_l|$
where $|\phi_l\>$ are orthonormal vectors and $\lambda_l =0$ for all $l>p$.
Furthermore let $L_j$ be the solution to 
$\frac{1}{2}(\rho L_j + L_j \rho ) = i \rho  E^j - i   E^j \rho$ according to Lemma \ref{lem:sld-unitary}. Then 
\begin{align}
&\tr (L_u \rho L_v )\notag\\
=&
4 \sum_{k=1}^p\sum_{l=1}^p
\lambda_l (  f(k,l)^2 - 1) 
\<\phi_{k}|{  E^u}  |\phi_{l}\>
\<\phi_l|{  E^v}|\phi_k\>  \notag\\
&+
4 \sum_{l=1}^p  \lambda_l \<\phi_{l}|{  E^v} {  E^u}|\phi_l\>.\notag
\end{align} 
\end{lemma}

\begin{proof}[Proof of Lemma \ref{lem:eval trace LjrhoLk}]
From Lemma \ref{lem:sld-unitary}, note that 
\begin{align}
&\tr (L_u \rho L_v )\notag\\
=&
-4\tr  \sum_{k,l} f(k,l)
\<\phi_k|{  E^u}|\phi_l\> |\phi_k\>\<\phi_l|
\notag\\
&\times\sum_{a=1}^p \lambda_a |\phi_a\>\<\phi_a| \sum_{k',l'} f(k',l')
\<\phi_{k'}|{  E^v}|\phi_{l'}\> |\phi_{k'}\>\<\phi_{l'}|.\notag
\end{align}
We can simplify this expression to 
\begin{align}
& \tr (L_u \rho L_v )\notag\\
=&
-4   \sum_{k}\sum_{l=1}^p
\lambda_l  f(k,l) f(l,k) 
\<\phi_k|{  E^u}|\phi_l\> 
\<\phi_{l}|{  E^v}|\phi_{k}\>  \notag\\
=&
4   \sum_{k}\sum_{l=1}^p
\lambda_l  f(k,l)^2  
\<\phi_{k}|{ E^u}  |\phi_{l}\> 
\<\phi_l|{  E^v}|\phi_k\>  \notag\\
=&
4  \sum_{k=1}^p\sum_{l=1}^p
\lambda_l  f(k,l)^2  
\<\phi_{k}|{  E^u}  |\phi_{l}\> 
\<\phi_l|{  E^v}|\phi_k\>  \notag\\
&+
4  \sum_{l=1}^p \lambda_l
  \<\phi_{l}|{  E^v} (  \sum_{k>p}  |\phi_{k}\> \<\phi_k|){  E^u}|\phi_l\>  \notag\\
=&
4   \sum_{k=1}^p\sum_{l=1}^p
\lambda_l  f(k,l)^2  
\<\phi_{k}|{  E^u}  |\phi_{l}\> 
\<\phi_l|{  E^v}|\phi_k\>  \notag\\
&+
4 \sum_{l=1}^p  \lambda_l
  \<\phi_{l}|{  E^v} ( I - \sum_{k=1}^p  |\phi_{k}\> \<\phi_k|) { E^u}|\phi_l\>\notag\\
=&
4  \sum_{k=1}^p\sum_{l=1}^p
\lambda_l (  f(k,l)^2 - 1)  
\<\phi_{k}|{  E^u}  |\phi_{l}\>
\<\phi_l|{  E^v}|\phi_k\>  \notag\\
&+
4 \sum_{l=1}^p  \lambda_l \<\phi_{l}|{  E^v} {  E^u}|\phi_l\>.\notag
\end{align}
\end{proof}

\

\subsection{Proof of Lemma \ref{L12}}
\label{app:SLD calculations}
The goal of this section is to show Lemma \ref{L12}, i.e., \eqref{full-expression}, which expresses 
$\tr[\bar L_u \bar \rho_\gamma \bar L_v]$ in terms of the eigenvectors and eigenvalues of $\rho_\gamma$ and the operators $E^u$ and $E^v$. 
To achieve this goal, 
we treat \eqref{full-expression} as 
a special case of Lemma \ref{lem:eval trace LjrhoLk}.
Then,
 we substituting $\rho$ for $\bar \rho_\gamma$, and note that $\bar \rho_\gamma$, has two non-zero eigenvalues, so we can set $p=2$. 
Then we get 
\begin{align}
&\tr (\bar L_u \bar\rho_\gamma \bar L_v )\notag\\
\stackrel{(a)}{=}&
4 \sum_{k=1}^2\sum_{l=1}^2
\lambda_l (f(k,l)^2-1) 
\<\phi_{k}|{  F^u}  |\phi_{l}\>
\<\phi_l|{  F^v}|\phi_k\>  \notag\\
&+
4 \sum_{l=1}^2  \lambda_l \<\phi_{l}|{  F^v} {  F^u}|\phi_l\>\\
\stackrel{(b)}{=}&
4 (f(1,1)^2-1) 
\lambda_1 
\<\phi_{1}|{  F^u}  |\phi_{1}\>
\<\phi_1|{  F^v}|\phi_1\>  \notag\\
&+
4 (f(2,2)^2-1)  \lambda_2 
\<\phi_{2}|{  F^u}  |\phi_{2}\>
\<\phi_2|{  F^v}|\phi_2\>  \notag\\
&+
4 g  \lambda_1 
\<\phi_{2}|{  F^u}  |\phi_{1}\>
\<\phi_1|{  F^v}|\phi_2\>  \notag\\
&+
4g \lambda_2 
\<\phi_{1}|{  F^u}  |\phi_{2}\>
\<\phi_2|{  F^v}|\phi_1\>  \notag\\
&+
4 \sum_{l=1}^2  \lambda_l \<\phi_{l}|{  F^v} {  F^u}|\phi_l\>\\
\stackrel{(c)}{=}&
-4
\lambda_1 
\<\phi_{1}|{  F^u}  |\phi_{1}\>
\<\phi_1|{  F^v}|\phi_1\>  \notag\\
&-
4  \lambda_2 
\<\phi_{2}|{  F^u}  |\phi_{2}\>
\<\phi_2|{  F^v}|\phi_2\>  \notag\\
&+
4 g  \lambda_1 
\<\phi_{2}|{  F^u}  |\phi_{1}\>
\<\phi_1|{  F^v}|\phi_2\>  \notag\\
&+
4g \lambda_2 
\<\phi_{1}|{  F^u}  |\phi_{2}\>
\<\phi_2|{  F^v}|\phi_1\>  \notag\\
&+
4 \sum_{l=1}^2  \lambda_l \<\phi_{l}|{  F^v} {  F^u}|\phi_l\>.\label{proof fullexp}
\end{align}
where in (a) we use $\bar \rho_\gamma = \lambda_1 |\phi_1\>\<\phi_1| + \lambda_2 |\phi_2\>\<\phi_2|$
and  replace $\rho$ with $\bar \rho_\gamma$, $L^u$ with $\bar L^u$, and $E^j$ with $F^j$ in the SLD equation $\frac{1}{2}(\rho L_j + L_j \rho ) = i \rho  E^j - i   E^j \rho$,
in (b) we use $f(1,2)^2=f(2,1)^2$ and $g=f(1,2)^2-1$,
in (c) we use $f(1,1) = f(2,2) = 0$. 
Since we show \eqref{proof fullexp}, \eqref{full-expression} follows.

\subsection{Proof of Lemma \ref{lem:E1E2-cross-products}}
\label{app:prove lemma 12}
Here, we prove Lemma \ref{lem:E1E2-cross-products}.
Note that  
\begin{align}
&\<\phi| F^1 | \epsilon\>\notag\\
\stackrel{(a)}{=}&
\frac{1}{n+1}
\sum_{u=0}^n 
 \<D^n_u|
(I-\frac{1}{2}{\tilde{A}_\gamma}^\dagger {\tilde{A}_\gamma})    {E^1} {\tilde{A}_\gamma}
 |D^n_u\>\notag\\
\stackrel{(b)}{=}&
\frac{\gamma}{n+1}
\sum_{u=1}^n 
 \<D^n_u|
(I-\frac{1}{2}{\tilde{A}_\gamma}^\dagger {\tilde{A}_\gamma})   {E^1} \sqrt{(n-u+1)u}
 |D^n_{u-1}\>\notag\\
\stackrel{(c)}{=}&
\frac{\gamma}{2(n+1)}
\sum_{u=1}^n 
 \<D^n_u|
(I-\frac{1}{2}{\tilde{A}_\gamma}^\dagger {\tilde{A}_\gamma})  (n-u+1)u
 |D^n_{u}\>\notag\\
\stackrel{(d)}{=}&
\frac{\gamma}{2(n+1)}
\sum_{u=1}^n 
 \<D^n_u|
(1-\frac{\gamma^2}{2} (n-u+1)u )  
(n-u+1)u
 |D^n_{u}\>\notag\\
\stackrel{(e)}{=}&
\frac{\gamma}{2(n+1)}
\sum_{u=1}^n 
(1-\frac{\gamma^2}{2} (n-u+1)u )  
(n-u+1)u
\notag\\
\stackrel{(f)}{=}&
 \frac{n\gamma}{12}(n+2) 
 -
 \frac{n\gamma^3}{120}(n+2)(n^2+2n+2).\notag
\end{align}
In (a) we use the definition of $|\phi\>$ and $|\epsilon\>$.
In (b) we use \eqref{tildeAgam sym}.
In (c) we use \eqref{E1 sym}.
In (d) we use \eqref{tildeAgam sym}.
In (e) we use the orthonormality of Dicke states.
In (f) we perform the summation over $u$.

We can similarly calculate $\<\phi| F^2 | \epsilon\>
$.
Furthermore, using the same ideas, we can calculate
\begin{align}
    \<\phi| F^3 |\phi\>
    &= 
    \<\Phi|F^3 |\Phi\>
    -  \<\Phi|\bar E^3 (A^\dagger A \otimes I_C) |\Phi\>
    + O(\gamma^4)
    \notag\\
    &=
   0 - \gamma^2
    \sum_{a=1}^n (n-a+1)a (n/2-a)  
    + O(\gamma^4)\notag\\
    &=
    \gamma^2 n(n+2)/12 + O(\gamma^4),
    \\
    \<\epsilon| F^3 |\epsilon\>
    &= 
    \frac{1}{n+1}\sum_{a=0}^{n-1}(n-a)(a+1)(n/2-a)
    \notag\\
    &
    =\gamma^2 n(n+2)/12.\notag
\end{align}
 
 \subsection{Proof of Lemma \ref{lem:commutator perturbation}}
 \label{app:commutator perturbation}
  Here we prove Lemma \ref{lem:commutator perturbation}.
First, we consider the spectral decomposition of $\rho_\gamma $ given by
 \begin{align}
\rho_\gamma = \sum_{j=1}^{(n+1)^2} \lambda'_j |\phi'_j\>\<\phi'_j|,\notag
 \end{align}
 where $\lambda'_j  \ge \lambda'_{j+1}$. 
 Recall that the spectral decomposition of $\bar \rho_\gamma$ is given by 
 \begin{align}
\bar \rho_\gamma =  \lambda_1 |\phi_1\>\<\phi_1|+   \lambda_2|\phi_2\>\<\phi_2|. \notag
 \end{align}

 \begin{lemma}\label{LL23}
For eigenvalues 
$\lambda_1, \lambda_2$ and
$\lambda_1', \lambda_2', \ldots, \lambda_{n+1}'$, we have
 \begin{align}
 \lambda'_1 - \lambda_1 &= O(\gamma^4),\quad
 \lambda'_2 - \lambda_2 = O(\gamma^4),\notag\\
 \lambda'_j &= O(\gamma^4),\quad \mbox{for all $j \ge 3$}\label{lam perturbations},
 \end{align}
and
\begin{align}
  \frac{(\lambda_2-\lambda_1)^2}{(\lambda_2 + \lambda_1)^2} - 1
&   =c_2 \gamma^2 + O(\gamma^4),
   \label{gval2} \\
 \frac{(\lambda'_2-\lambda'_1)^2}{(\lambda'_2 + \lambda'_1)^2} - 1
& = c_2 \gamma^2 + O(\gamma^4),
 \label{g12' value}  \\
   \frac{(\lambda'_k-\lambda'_1)^2}{(\lambda'_k + \lambda'_1)^2} - 1  
   &= O(\gamma^4), 
   \quad \mbox{for all $k \ge 3$,}
    \label{g1k' value} \\
   \frac{(\lambda'_k-\lambda'_2)^2}{(\lambda'_k + \lambda'_2)^2} - 1 
   & = O(\gamma^2), 
   \quad \mbox{for all $k \ge 3$,}
    \label{g2k' value}
\end{align} 
where $c_2 =  2 n(n+2)/3$.
\end{lemma}
\begin{proof}
The relations \eqref{lam perturbations} follow from the Gersgorin circle theorem \cite{varga-GCT}, 

The relation \eqref{gval2} follows from \eqref{gval}.
Now, we have
\begin{align}
&\left( \frac{(\lambda'_2-\lambda'_1)^2}{(\lambda'_2 + \lambda'_1)^2} - 1\right)
-\left(\frac{(\lambda_2-\lambda_1)^2}{(\lambda_2 + \lambda_1)^2} - 1\right)
\notag\\
\stackrel{(a)}{=}&
 \frac{(\lambda'_2-\lambda'_1)^2}{(\lambda'_2 + \lambda'_1)^2} 
-\frac{(\lambda'_2-\lambda'_1)^2}{(\lambda_2 + \lambda_1)^2} 
+
 \frac{(\lambda'_2-\lambda'_1)^2}{(\lambda_2 + \lambda_1)^2} 
-\frac{(\lambda_2-\lambda_1)^2}{(\lambda_2 + \lambda_1)^2} 
\notag\\
\stackrel{(b)}{=}&
(\lambda'_2-\lambda'_1)^2 ( (\lambda'_2+\lambda'_1 )^{-2} - (\lambda_2+\lambda_1 )^{-2} )
\notag\\
&
+
( (\lambda'_2-\lambda'_1)^2  -  (\lambda_2-\lambda_1)^2) ( \lambda_2+\lambda_1 )^{-2}
\notag\\
\stackrel{(c)}{=}&
(\lambda'_2-\lambda'_1)^2 O(\gamma^4) 
+
( (\lambda'_2-\lambda'_1)^2  -  (\lambda_2-\lambda_1)^2) ( 1+ O(\gamma^4) )
\notag\\
\stackrel{(d)}{=}&
(1+O(\gamma^2)) O(\gamma^4) 
+
O(\gamma^4) ( 1+ O(\gamma^4) ) =   O(\gamma^4)  ,
\label{g12 perturb}
 \end{align}
 where in (a) we use a telescoping sum,
 in (b) we collect terms in the telescoping sum,
 in (c) we use $\lambda_1+ \lambda_2 = 1+O(\gamma^4)$,
 $\lambda'_1+ \lambda'_2 = 1+O(\gamma^4)$,
 which implies that 
 $(\lambda_1+ \lambda_2)^{-2} = 1+ O(\gamma^4)$
 and
 $(\lambda'_1+ \lambda'_2)^{-2} = 1+ O(\gamma^4)$,
 in (d) we use 
 $(\lambda'_2 - \lambda'_1)^2
 =(\lambda_2 - \lambda_1)^2 + O(\gamma^4)$ because of \eqref{lam perturbations}
 and also $\lambda'_2-\lambda_1' =2c_2\gamma^2-1 + O(\gamma^4)$.
Hence, from \eqref{gval2} and \eqref{g12 perturb}, we have
\eqref{g12' value}.  

Since $\lambda'_1 =  1- c_2\gamma^2 + O(\gamma^4)$ and $\lambda'_k = O(\gamma^4)$ for all $k\ge 3$, we have 
\eqref{g1k' value}.
Since $\lambda'_2 = c_2\gamma^2 + O(\gamma^4)$ and $\lambda'_k = O(\gamma^4)$ for all $k\ge 3$, we have 
\eqref{g2k' value}.
\end{proof}

We prepare several lemmas.

\begin{lemma}\label{LL20}
We have 
\begin{align}
\rho_\gamma- \bar \rho_\gamma =
O(\gamma^4) .\notag
\end{align}
\end{lemma}

\begin{proof}
We have
\begin{align}
\rho_\gamma  
\stackrel{(a)}{=}&
(\Lambda_{0,\gamma} \otimes \iota_C )(|\Phi\>\<\Phi|)
\notag\\
\stackrel{(b)}{=}&
(e^{\mathcal L_{0}} \otimes \iota_C) (|\Phi\>\<\Phi|) \notag\\
\stackrel{(c)}{=}&
|\Phi\>\<\Phi|  + (\tilde{A}_\gamma \otimes I_C) |\Phi\>\<\Phi|    (\tilde{A}_\gamma \otimes I_C)^\dagger 
\notag\\
&
- \frac{1}{2}    (\tilde{A}_\gamma^\dagger \tilde{A}_\gamma \otimes I_C)   |\Phi\>\<\Phi|
- \frac{1}{2}  |\Phi\>\<\Phi| (\tilde{A}_\gamma^\dagger \tilde{A}_\gamma \otimes I_C)  
+ O(\gamma^4) \notag\\
\stackrel{(d)}{=}&
|\Phi\>\<\Phi|  + (\tilde{A}_\gamma \otimes I_C) |\Phi\>\<\Phi|    (\tilde{A}_\gamma \otimes I_C)^\dagger 
\notag\\
&
- \frac{1}{2}    (\tilde{A}_\gamma^\dagger \tilde{A}_\gamma \otimes I_C)   |\Phi\>\<\Phi|
- \frac{1}{2}  |\Phi\>\<\Phi| (\tilde{A}_\gamma^\dagger \tilde{A}_\gamma \otimes I_C)  
\notag\\
&+ \frac{1}{4}   (\tilde{A}_\gamma^\dagger \tilde{A}_\gamma \otimes I_C) 
 |\Phi\>\<\Phi| (\tilde{A}_\gamma^\dagger \tilde{A}_\gamma \otimes I_C)  
+ O(\gamma^4) \notag\\
\stackrel{(e)}{=}&
 |\phi\>\<\phi| + 
|\epsilon\>\<\epsilon| +  O(\gamma^4) 
\stackrel{(f)}{=}
\bar \rho_\gamma + O(\gamma^4),\label{HH2}
\end{align}
where the derivation of each step are given as follows.
Step $(a)$ follows from the definition of $\rho_\gamma  $.
Step $(b)$ follows from \eqref{NMR5}.
Step $(c)$ follows from \eqref{BNJ} and \eqref{NMR5}.
Step $(d)$ follows from the fact that
$  (\tilde{A}_\gamma^\dagger \tilde{A}_\gamma \otimes I_C) 
 |\Phi\>\<\Phi| (\tilde{A}_\gamma^\dagger \tilde{A}_\gamma \otimes I_C)  $ is $O(\gamma^4)$.
Step  $(e)$ follows from \eqref{BYU1} and \eqref{BYU2}.
Step  $(f)$ follows from \eqref{BYU3}.
\end{proof}

 Next, we have the perturbation bound for the eigenvectors of $\rho_\gamma$.

 \begin{lemma}\label{lem:eigv perturbations}
 \begin{align}
 \ \| |\phi'_1  \> \<\phi'_1 |- |\phi_1  \> \<\phi_1 | \|_1 &= O(\gamma^4)\label{a1 perturb}\\
 \ \| |\phi'_2  \> \<\phi'_2 |- |\phi_2  \> \<\phi_2 | \|_1
 &= O(\gamma^2) \label{a2 perturb}.
 \end{align}
 \end{lemma}
 \begin{proof}\par
 
{\bf Step 1:\quad } 
First, we show \eqref{a2 perturb}.
Let us use the following argument according to \cite[Eq (7)]{fan2018eigenvector}.
 Let $A$ be a Hermitian matrix with nonzero eigenvalues $a_1, \dots
a_r$ with corresponding eigenvectors $v_1, \dots , v_r$, where $a_j
 \ge a_{j+1}$. Define the gap of $A$ as
\begin{align}
 g(A) = \min_{j=1,\dots, r} {a_j - a_{j+1}}\notag
\end{align}
Let $A+E$ be a Hermitian matrix with eigenvalues $\mu_j$, and eigenvectors $w_j$,
 where $\mu_j \ge  \mu_{j+1}.$
 Then for all $j=1,\dots ,r$, we have
\begin{align}
\| |v_j\> -| w_j\> \| \le  2 \sqrt 2 \| E \|_2 / g(A),\label{BAY3}
\end{align}
where $\| \cdot \|_2$ is the Hilbert-Schmidt norm. 
Now the smallest non-zero eigenvalue of $\rho_\gamma $ is $c\gamma^2$ for some $c>  0$ (see \eqref{ee}) and 
Lemma \ref{LL20}, i.e., 
$\rho_\gamma - \bar \rho_\gamma =O(\gamma^4)$.
Hence, 
\begin{align}
&\| |\phi_2\rangle \langle \phi_2|-
|\phi'_2\rangle \langle \phi'_2|\|_1 
\le \sqrt{ 1-|\langle \phi'_2|\phi_2\rangle|^2} \notag\\
\le & \| |\phi_2\rangle-|\phi'_2\rangle\|\notag \\
\le & 2 \sqrt 2 \| \rho_\gamma - \bar \rho_\gamma \|_2 / g(\rho_\gamma) 
=O(\gamma^4 / \gamma^2)  = O(\gamma^2),\notag
\end{align}
which shows \eqref{a2 perturb}.

{\bf Step 2:\quad } 
Toward the proof of \eqref{a1 perturb}, 
for simplicity in notation,
we define $B \coloneqq \widetilde A_\gamma \otimes I_C$ and
 \begin{align}
 |\phi^{(2)}\> &\coloneqq  (I_{AC} - \frac{1}{2} B^\dagger B + \frac{1}{8}B^\dagger B B^\dagger B)|\Phi\>\notag\\
 |\epsilon^{(2)}\> &\coloneqq  
 (B - \frac{1}{4} BB^\dagger B - \frac{1}{4} B^\dagger BB )|\Phi\>.\notag
 \end{align}
The vectors $ |\phi^{(2)}\>$ and $ |\epsilon^{(2)}\>$ are orthogonal for the same reason that 
$ |\phi\>$ and $ |\epsilon\>$ are orthogonal.
Let us define 
 \begin{align}
  \rho_\gamma^{(2)} :=  |\phi^{(2)}\> \<\phi^{(2)}|
  + |\epsilon^{(2)}\> \<\epsilon^{(2)}|.\notag
 \end{align}
 The aim of this step is to show the equation 
 \begin{align}
\rho_\gamma =  \rho_\gamma^{(2)} + O(\gamma^6).\label{NM76}
 \end{align}

Since $\theta = 0$, we have 
 \begin{align}
 \rho_\gamma = |\Phi\>\<\Phi| + \sum_{j=1}^\infty \frac{1}{j!}(\mathcal L_{\theta,\gamma} \otimes \iota_C)^j(|\Phi\>\<\Phi|).\notag
 \end{align}
 We note that 
 \begin{align}
 \rho_\gamma &= |\Phi\>\<\Phi| + \sum_{j=2}^\infty \frac{1}{j!}(\mathcal L_{\theta,\gamma} \otimes \iota_C)^j(|\Phi\>\<\Phi|) + O(\gamma^6).
\label{NH13}
 \end{align}

Now
 \begin{align}
& (\mathcal L_{\theta,\gamma} \otimes \iota_C)^2(|\Phi\>\<\Phi|)\notag\\
 =& BB  |\Phi\>\<\Phi|  B^\dagger B^\dagger 
 - \frac{1}{2}(BB^\dagger B|\Phi\>\<\Phi| B^\dagger + B |\Phi\>\<\Phi| B^\dagger B B^\dagger) \notag\\
&- \frac{1}{2} B^\dagger B B |\Phi\>\<\Phi| B^\dagger + \frac{1}{4} B^\dagger B B^\dagger B |\Phi\>\<\Phi| +  \frac{1}{4} B^\dagger B |\Phi\>\<\Phi| B^\dagger B
 \notag\\
 & 
 -\frac{1}{2} B |\Phi\>\<\Phi| B^\dagger B^\dagger B + \frac{1}{4} B^\dagger B |\Phi\>\<\Phi| B^\dagger B + \frac{1}{4} |\Phi\>\<\Phi|  B^\dagger B B^\dagger B
 \notag\\
=&  |\phi^{(2)}\> \<\phi^{(2)}|
  + |\epsilon^{(2)}\> \<\epsilon^{(2)}|=
\rho_\gamma^{(2)} . \label{NH12}
 \end{align}
The combination of \eqref{NH13} and \eqref{NH12}
implies \eqref{NM76}.

{\bf Step 3:\quad } 
Now, define 
 \begin{align}
 \lambda_1 ^{(2)}  &\coloneqq   \<\phi^{(2)} |\phi^{(2)}\> \notag\\
 \lambda_2 ^{(2)} & \coloneqq  \<\epsilon^{(2)} |\epsilon^{(2)}\>\notag\\
|\phi^{(2)}_1\> & \coloneqq |\phi^{(2)} \>/ \sqrt{ \lambda_1 ^{(2)} } \notag\\
|\phi^{(2)}_2\> & \coloneqq |\epsilon^{(2)} \>/ \sqrt{ \lambda_2 ^{(2)} } ,\notag
 \end{align}
The aim of this step is to show
\begin{align}
\| |\phi^{(2)}_1\> - |\phi _1\>\|
&=O(\gamma^4)  .\label{phi2 phi1}
\end{align}

 Now we have
 \begin{align}
\lambda_1 / \lambda_1^{(2)} 
&= \lambda_1 / ( \lambda_1 +  \lambda_1^{(2)}  - \lambda_1) \notag\\
&= (1+ ( \lambda_1^{(2)}  - \lambda_1)/\lambda_1 )^{-1}\notag\\
&\stackrel{(a)}{=} (1+  O(\gamma^4) (1+O(\gamma^2) )^{-1}\notag\\
&= (1+ O(\gamma^4  ))^{-1}\notag\\
&= 1+ O(\gamma^4)^{-1}\label{lam12 ratio}
 \end{align}
where in (a) we use $( \lambda_1^{(2)}  - \lambda_1) = O(\gamma^4)$ and $\lambda_1 = 1 + O(\gamma^2)$.

Since $|\phi_1\> = |\phi\>/\sqrt{\lambda_1}$
where $|\phi\>$ is defined in \eqref{BYU1},
we see that 
\begin{align}
\sqrt{\lambda_1}   |\phi_1\> &= 
( (I_A  - \frac{1}{2}  {\tilde{A}_\gamma} ^\dagger{\tilde{A}_\gamma})\otimes I_C)
|\Phi\>
\notag\\
&\stackrel{(a)}{=}
 (I_A\otimes I_C  - \frac{1}{2}  (\tilde{A}_\gamma \otimes I_C) ^\dagger
(\tilde{A}_\gamma \otimes I_C))
|\Phi\>\notag\\
&\stackrel{(b)}{=}
( I_{AC} - \frac{1}{2} B ^\dagger B)
|\Phi\>,
\label{phibb}
\end{align}
where
(a) is because we distribute the tensor product,
and (b) is because 
$I_{A} \otimes I_C = I_{AC}$ 
and 
$B =\tilde A_\gamma \otimes I_C$.

Now we can write 
 \begin{align}
|\phi^{(2)}_1\> &= 
(\lambda_1^{(2)})^{-1/2}
 (I_{AC} - \frac{1}{2} B^\dagger B + \frac{1}{8}B^\dagger B B^\dagger B)|\Phi\>
 \notag\\
 &\stackrel{(a)}{=}
 (\lambda_1^{(2)})^{-1/2}
 \sqrt{\lambda_1}|\phi_1\> + (\lambda_1^{(2)})^{-1/2} \frac{1}{8}B^\dagger B B^\dagger B|\Phi\>   \notag\\
 &\stackrel{(b)}{=}
(  1+ O(\gamma^4)  ) |\phi_1\> + (1+ O(\gamma^2)) \frac{1}{8}B^\dagger B B^\dagger B|\Phi\>  \notag 
 \end{align}
 where in (a) we 
use \eqref{phibb} which tells us that 
$\sqrt{\lambda_1} |\phi_1\> = (I_{AC}  - \frac{1}{2}B^\dagger B)|\Phi\>$, 
and
in (b) we use \eqref{lam12 ratio} and 
$(\lambda_1^{(2)})^{-1/2} =  1+ O(\gamma^2)$.
Hence 
\begin{align}
\| |\phi^{(2)}_1\> - |\phi _1\>\|
&=\| O(\gamma^4)    |\phi_1\> + (1+ O(\gamma^2)) \frac{1}{8}B^\dagger B B^\dagger B|\Phi\>\|\notag\\
&=O(\gamma^4)  .\notag
\end{align}

{\bf Step 4:\quad } 
The aim of this step is to derive \eqref{a1 perturb}
by using \eqref{NM76} and \eqref{phi2 phi1}.

First, we notice the spectral decomposition 
  \begin{align}
  \rho_\gamma^{(2)} =  \lambda_1 ^{(2)} |\phi^{(2)}_1\> \<\phi^{(2)}_1| 
  + \lambda_2 ^{(2)} |\phi^{(2)}_2\> \<\phi^{(2)}_2|.\notag
 \end{align}

Since the smallest non-zero eigenvalue of $\rho_\gamma^{(2)} = c\gamma^2 + O(\gamma^4)$ for some positive $c$, and 
\eqref{NM76},
we can use \eqref{BAY3} to find that 
\begin{align}
&\| |\phi^{(2)}_1\> - |\phi' _1\>\| 
\le  2 \sqrt 2 \| \rho_\gamma - \bar \rho_\gamma^{(2)} \|_2 / 
g(\rho_\gamma^{(2)}) \notag\\
= &O (\gamma^6 ) / ( c\gamma^2 + O(\gamma^4)  ) = O(\gamma^4) \label{phi2 alpha}.
\end{align}
Hence combining \eqref{phi2 phi1} with \eqref{phi2 alpha} with the triangle inequality gives
\begin{align}
& \| |\phi_1\>\< \phi_1| - |\phi' _1\>\< \phi' _1|\|_1
\le \| |\phi_1\> - |\phi' _1\>\| \notag\\
\le & \| |\phi_1^{(2)}\> - |\phi' _1\>\| + \| |\phi_1^{(2)}\> - |\phi _1\>\| 
= O(\gamma^4).\notag
\end{align}
  \end{proof}

Next, using Lemmas \ref{lem:eval trace LjrhoLk}
and \ref{LL20}, and defining $Q: = {  F^v} {  F^u}-F^u F^v$, we prove the following lemma.
\begin{lemma}\label{LL22}
Let $u,v=1,2,3$. Then we have
\begin{align}
&\tr (L_u \rho_\gamma L_v ) - \tr (L_v \rho_\gamma L_u )
\notag\\
=&
4 \sum_{\substack{1\le k,l\le p'\\ k\neq l}}
\lambda'_l \left(  \frac{(\lambda'_k-\lambda'_l)^2}{(\lambda'_k + \lambda'_l)^2} - 1\right) 
\<\phi'_{k}|{  F^u}  |\phi'_{l}\>
\<\phi'_l|{  F^v}|\phi'_k\>  \notag\\
&+
4 \sum_{l=1}^{p'}  \lambda'_l \<\phi'_{l}|Q|\phi'_l\>,\label{true commutator}
\end{align}  
and
\begin{align}  
|\<\phi'_{1}|Q|\phi'_1\> -\<\phi_{1}|Q|\phi_1\>|
&\le  O(\gamma^4),\label{BN1}\\
|\<\phi'_{2}|Q|\phi'_2\> -\<\phi_{2}|Q|\phi_2\>|
&\le O(\gamma^2),\label{BN2} \\
4 \sum_{l=1}^{p'}  \lambda'_l \<\phi'_{l}|Q|\phi'_l\>
-
4 \sum_{l=1}^{2}  \lambda_l \<\phi_{l}|Q|\phi_l\> 
&=  O(\gamma^4).
   \label{C75} 
\end{align}  
For $(k,l)  \in \{(1,2), (2,1) \}$, we have
\begin{align}  
&|\<\phi'_{k}|{  F^u}  |\phi'_{l}\>
\<\phi'_l|{  F^v}|\phi'_k\> 
- 
\<\phi_{k}|{  F^u}  |\phi_{l}\>
\<\phi_l|{  F^v}|\phi_k\> |\le
O(\gamma^2).
\label{innerprod perturb}
\end{align}  
\end{lemma}
\begin{proof}
Now from Lemma \ref{lem:eval trace LjrhoLk} 
it follows that 
\begin{align}
&\tr (L_u \rho_\gamma L_v )\notag\\
=&
4 \sum_{k=1}^{p'}\sum_{l=1}^{p'}
\lambda'_l \left(  \frac{(\lambda'_k-\lambda'_l)^2}{(\lambda'_k + \lambda'_l)^2} - 1\right) 
\<\phi'_{k}|{  F^u}  |\phi'_{l}\>
\<\phi'_l|{  F^v}|\phi'_k\>  \notag\\
&+
4 \sum_{l=1}^{p'}  \lambda'_l \<\phi'_{l}|{  F^v} {  F^u}|\phi'_l\>.
\label{C70}
\end{align}
where $p'$ is the number of non-zero eigenvalues of $\rho_\gamma$.
Then from \eqref{C70}, we get \eqref{true commutator} as
\begin{align}
&\tr (L_u \rho_\gamma L_v ) - \tr (L_v \rho_\gamma L_u )
\notag\\
=&
4 \sum_{\substack{1\le k,l\le p'\\ k\neq l}}
\lambda'_l \left(  \frac{(\lambda'_k-\lambda'_l)^2}{(\lambda'_k + \lambda'_l)^2} - 1\right) 
\<\phi'_{k}|{  F^u}  |\phi'_{l}\>
\<\phi'_l|{  F^v}|\phi'_k\>  \notag\\
&+
4 \sum_{l=1}^{p'}  \lambda'_l \<\phi'_{l}|({  F^v} {  F^u}-F^u F^v)|\phi'_l\>.\notag
\end{align}

Lemma \ref{lem:eigv perturbations} implies \eqref{BN1} and \eqref{BN2} as
\begin{align}
|\<\phi'_{1}|Q|\phi'_1\> -\<\phi_{1}|Q|\phi_1\>|
&\le \|Q\|  \cdot \| |\phi'_1  \> \<\phi'_1 |- |\phi_1  \> \<\phi_1 | \|_1 \notag\\
&=\|Q\|   O(\gamma^4)=  O(\gamma^4)
\\
|\<\phi'_{2}|Q|\phi'_2\> -\<\phi_{2}|Q|\phi_2\>|
&\le \|Q\|  \cdot \| |\phi'_2  \> \<\phi'_2 |- |\phi_2  \> \<\phi_2 | \|_1\notag\\
&=\|Q\|   O(\gamma^2)=O(\gamma^2).\notag
\end{align}

From \eqref{lam perturbations} and 
it follows that 
\begin{align}
&4 \sum_{l=1}^{p'}  \lambda'_l \<\phi'_{l}|Q|\phi'_l\>
-
4 \sum_{l=1}^{2}  \lambda_l \<\phi_{l}|Q|\phi_l\>\notag\\
\stackrel{(a)}{=} &
4    \lambda'_1 \<\phi'_{1}|Q|\phi'_1\>
-
4 \lambda_l \<\phi_{1}|Q|\phi_1\>
\notag\\&
+
4 \sum_{l=2}^{p'}  \lambda'_l \<\phi'_{l}|Q|\phi'_l\>
-
4    \lambda_2 \<\phi_2|Q|\phi_2\>
\notag\\
\stackrel{(b)}{=} &
4    \lambda'_1 \<\phi'_{1}|Q|\phi'_1\>
-
4 \lambda_1 \<\phi_{1}|Q|\phi_1\>\notag\\
&+
4  \lambda'_2 \<\phi'_{2}|Q|\phi'_2\>
-
4   \lambda_2 \<\phi_{2}|Q|\phi_2\>+O(\gamma^4)
\notag\\ 
\stackrel{(c)}{=} &
4    \lambda'_1 \<\phi_{1}|Q|\phi_1\> + \lambda'_1 O(\gamma^4)
-
4 \lambda_1 \<\phi_{1}|Q|\phi_1\>\notag\\
&+
4    \lambda'_2 \<\phi_2|Q|\phi_2\> + \lambda'_2 O(\gamma^2)
-
4   \lambda_2 \<\phi_2|Q|\phi_2\>
\notag\\ 
\stackrel{(d)}{=} &
4  (  \lambda_1 + O(\gamma^4)) \<\phi_{1}|Q|\phi_1\> + O(\gamma^4)
-
4 \lambda_1 \<\phi_{1}|Q|\phi_1\>\notag\\
&+
4    (\lambda_2 + O(\gamma^4) )\<\phi_2|Q|\phi_2\> + O(\gamma^2)
-
4    \lambda_2 \<\phi_2|Q|\phi_2\> \notag
\end{align}
where in (a) we separate terms in the summation, 
in (b) we use \eqref{lam perturbations} which means that $\lambda'_l = O(\gamma^4)$ when $l \ge 3$,
in (c) we use 
\eqref{BN1} and \eqref{BN2},
in (d) we use \eqref{lam perturbations}.
Simplifying the above, we find \eqref{C75}

Next note that for $(k,l)  \in \{(1,2), (2,1) \}$, we have
\begin{align}
&|\<\phi'_{k}|{  F^u}  |\phi'_{l}\>
\<\phi'_l|{  F^v}|\phi'_k\> 
- 
\<\phi_{k}|{  F^u}  |\phi_{l}\>
\<\phi_l|{  F^v}|\phi_k\> |\notag\\
\stackrel{(a)}{=}&
|\tr( {  F^u}  |\phi'_{l}\>
\<\phi'_l|{  F^v}|\phi'_k\> \<\phi'_{k}|)
- 
\tr( {  F^u}  |\phi_{l}\>
\<\phi_l|{  F^v}|\phi_k\>\<\phi_{k}| )|
\notag\\
\stackrel{(b)}{=}&
|\tr( {  F^u}  |\phi'_{l}\>
\<\phi'_l|{  F^v}|\phi'_k\> \<\phi'_{k}|)
- 
\tr( {  F^u}  |\phi'_{l}\>
\<\phi'_l|{  F^v}|\phi_k\>\<\phi_{k}| )|
\notag\\
&
+
|\tr( {  F^u}  |\phi'_{l}\>
\<\phi'_l|{  F^v}|\phi_k\> \<\phi_{k}|)
- 
\tr( {  F^u}  |\phi_{l}\>
\<\phi_l|{  F^v}|\phi_k\>\<\phi_{k}| )|
\notag\\
\stackrel{(c)}{=}&
O(\| |\phi'_k\>\<\phi'_k| - |\phi_k\>\<\phi_k|  \|_1 )
+ O(\| |\phi'_l\>\<\phi'_l| - |\phi_l\>\<\phi_l|  \|_1 )
\notag\\
\stackrel{(d)}{=}&
O(\gamma^2).\notag
\end{align}
where (a) follows by the cyclicity of the trace,
(b) uses a telescoping sum with the triangle inequality,
(c) follows 
because 
\begin{align}
&|\tr(F^u |\phi'_l\>\<\phi'_l|F^v) ( |\phi'_k\>\<\phi'_k| - |\phi_k\>\<\phi_k|  )  |\notag\\
\le& 
\| F^u |\phi'_l\>\<\phi'_l|F^v  \| \cdot \| |\phi'_k\>\<\phi'_k| - |\phi_k\>\<\phi_k|  \|_1 \notag \\
\le& 
\| F^u \| \|F^v  \| \cdot \| |\phi'_k\>\<\phi'_k| - |\phi_k\>\<\phi_k|  \|_1,\notag
\end{align}
and 
\begin{align}
&|\tr ( |\phi'_l\>\<\phi'_l| - |\phi_l\>\<\phi_l|  ) (F^v |\phi_k\>\<\phi_k|F^u   )|\notag\\
\le &
\|  |\phi'_l\>\<\phi'_l| - |\phi_l\>\<\phi_l| \|_1 
\cdot\|F^v |\phi_k\>\<\phi_k|F^u \| \notag\\
\le &
\| F^u \| \|F^v  \| \cdot \| |\phi'_k\>\<\phi'_k| - |\phi_k\>\<\phi_k|  \|_1,\notag
\end{align}
and (d) follows from Lemma \ref{lem:eigv perturbations}.
Hence, we obtain \eqref{innerprod perturb}.
\end{proof}

Finally, using Lemmas \ref{LL22} and \ref{LL23},
we derive \eqref{EE160}.
We can get
\begin{align}
&( \tr(L_u \rho_\gamma L_v) - \tr (L_v \rho_\gamma L_u) )
- ( \tr(\bar L_u \bar \rho_\gamma \bar L_v) - \tr (\bar L_v \bar \rho_\gamma \bar L_u) ) \notag\\
\stackrel{(a)}{=}&
4 \sum_{\substack{1\le k,l\le p'\\ k\neq l}}
\lambda'_l \left(  \frac{(\lambda'_k-\lambda'_l)^2}{(\lambda'_k + \lambda'_l)^2} - 1\right) 
\<\phi'_{k}|{  F^u}  |\phi'_{l}\>
\<\phi'_l|{  F^v}|\phi'_k\>  \notag\\
&-
4 \sum_{\substack{1\le k,l\le 2\\ k\neq l}}
\lambda_l \left(  \frac{(\lambda_k-\lambda_l)^2}{(\lambda_k + \lambda_l)^2} - 1\right) 
\<\phi_{k}|{  F^u}  |\phi_{l}\>
\<\phi_l|{  F^v}|\phi_k\>  \notag\\
&+
4 \sum_{l=1}^{p'}  \lambda'_l \<\phi'_{l}|({  F^v} {  F^u}-F^u F^v)|\phi'_l\>
\notag\\
&
-
4 \sum_{l=1}^{2}  \lambda_l \<\phi_{l}|({  F^v} {  F^u}-F^u F^v)|\phi_l\>
\notag\\
\stackrel{(b)}{=}&
4 \sum_{\substack{1\le k,l\le p'\\ k\neq l}}
\lambda'_l \left(  \frac{(\lambda'_k-\lambda'_l)^2}{(\lambda'_k + \lambda'_l)^2} - 1\right) 
\<\phi'_{k}|{  F^u}  |\phi'_{l}\>
\<\phi'_l|{  F^v}|\phi'_k\>  \notag\\
&-
4 \sum_{\substack{1\le k,l\le 2\\ k\neq l}}
\lambda_l \left(  \frac{(\lambda_k-\lambda_l)^2}{(\lambda_k + \lambda_l)^2} - 1\right) 
\<\phi_{k}|{  F^u}  |\phi_{l}\>
\<\phi_l|{  F^v}|\phi_k\>  \notag\\
&+ O(\gamma^4)
\notag\\
\stackrel{(c)}{=}&
4 \sum_{\substack{1\le k,l\le 2\\ k\neq l}}
\lambda'_l \left(  \frac{(\lambda'_k-\lambda'_l)^2}{(\lambda'_k + \lambda'_l)^2} - 1\right) 
\<\phi'_{k}|{  F^u}  |\phi'_{l}\>
\<\phi'_l|{  F^v}|\phi'_k\>  \notag\\
&-
4 \sum_{\substack{1\le k,l\le 2\\ k\neq l}}
\lambda_l \left(  \frac{(\lambda_k-\lambda_l)^2}{(\lambda_k + \lambda_l)^2} - 1\right) 
\<\phi_{k}|{  F^u}  |\phi_{l}\>
\<\phi_l|{  F^v}|\phi_k\>  \notag\\
&+ 
4 \sum_{\substack{l= 1, k > 2\\ k\neq l}}
\lambda'_l \left(  \frac{(\lambda'_k-\lambda'_l)^2}{(\lambda'_k + \lambda'_l)^2} - 1\right) 
\<\phi'_{k}|{  F^u}  |\phi'_{l}\>
\<\phi'_l|{  F^v}|\phi'_k\>  \notag\\
&+ 
4 \sum_{\substack{l = 2, k > 2\\ k\neq l}}
\lambda'_l \left(  \frac{(\lambda'_k-\lambda'_l)^2}{(\lambda'_k + \lambda'_l)^2} - 1\right) 
\<\phi'_{k}|{  F^u}  |\phi'_{l}\>
\<\phi'_l|{  F^v}|\phi'_k\>  \notag\\
&+ O(\gamma^4)
\notag\\
\stackrel{(d)}{=}&
4 \sum_{\substack{1\le k,l\le 2\\ k\neq l}}
\lambda'_l \left(  \frac{(\lambda'_k-\lambda'_l)^2}{(\lambda'_k + \lambda'_l)^2} - 1\right) 
\<\phi'_{k}|{  F^u}  |\phi'_{l}\>
\<\phi'_l|{  F^v}|\phi'_k\>  \notag\\
&-
4 \sum_{\substack{1\le k,l\le 2\\ k\neq l}}
\lambda_l \left(  \frac{(\lambda_k-\lambda_l)^2}{(\lambda_k + \lambda_l)^2} - 1\right) 
\<\phi_{k}|{  F^u}  |\phi_{l}\>
\<\phi_l|{  F^v}|\phi_k\>  \notag\\
&+ 
4 \sum_{\substack{k > 2\\ k\neq l}}
\lambda'_1 O(\gamma^4)
\<\phi'_{k}|{  F^u}  |\phi'_1\>
\<\phi'_1|{  F^v}|\phi'_k\>  \notag\\
&+ 
4 \sum_{\substack{k > 2\\ k\neq l}}
\lambda'_2  O(\gamma^2)
\<\phi'_{k}|{  F^u}  |\phi'_2\>
\<\phi'_2|{  F^v}|\phi'_k\> + O(\gamma^4),
\label{halfway}
\end{align}
where 
(a) follows from \eqref{true commutator} and Lemma \ref{L12},
(b) follows from \eqref{C75}
(c) follows from expanding the first summation, and using $\lambda'_l = O(\gamma^4)$ for $l >2$ from \eqref{lam perturbations},
(d) follows from
\eqref{g1k' value} and \eqref{g2k' value}.

Hence we get
\begin{align}
&( \tr(L_u \rho_\gamma L_v) - \tr (L_v \rho_\gamma L_u) )
- ( \tr(\bar L_u \bar \rho_\gamma \bar L_v) - \tr (\bar L_v \bar \rho_\gamma \bar L_u) ) \notag\\
\stackrel{(a)}{=}&
4 \sum_{\substack{1\le k,l\le 2\\ k\neq l}}
\lambda'_l \left(  \frac{(\lambda'_k-\lambda'_l)^2}{(\lambda'_k + \lambda'_l)^2} - 1\right) 
\<\phi'_{k}|{  F^u}  |\phi'_{l}\>
\<\phi'_l|{  F^v}|\phi'_k\>  \notag\\
&-
4 \sum_{\substack{1\le k,l\le 2\\ k\neq l}}
\lambda_l \left(  \frac{(\lambda_k-\lambda_l)^2}{(\lambda_k + \lambda_l)^2} - 1\right) 
\<\phi_{k}|{  F^u}  |\phi_{l}\>
\<\phi_l|{  F^v}|\phi_k\>  \notag\\ 
&+O(\gamma^4)\notag\\
\stackrel{(b)}{=}&
4 \sum_{\substack{1\le k,l\le 2\\ k\neq l}}
\lambda'_l (c_2 \gamma^2 + O(\gamma^4))
\<\phi'_{k}|{  F^u}  |\phi'_{l}\>
\<\phi'_l|{  F^v}|\phi'_k\>  \notag\\
&-
4 \sum_{\substack{1\le k,l\le 2\\ k\neq l}}
\lambda_l (c_2 \gamma^2 + O(\gamma^4))
\<\phi_{k}|{  F^u}  |\phi_{l}\>
\<\phi_l|{  F^v}|\phi_k\>+O(\gamma^4)  \notag\\  
\stackrel{(c)}{=}&
4 \sum_{\substack{1\le k,l\le 2\\ k\neq l}}
\lambda'_l (c_2 \gamma^2 + O(\gamma^4))
\<\phi'_{k}|{  F^u}  |\phi'_{l}\>
\<\phi'_l|{  F^v}|\phi'_k\>  \notag\\
&-
4 \sum_{\substack{1\le k,l\le 2\\ k\neq l}}
\lambda'_l (c_2 \gamma^2 + O(\gamma^4))
\<\phi_{k}|{  F^u}  |\phi_{l}\>
\<\phi_l|{  F^v}|\phi_k\>  \notag\\  
&
+4 \sum_{\substack{1\le k,l\le 2\\ k\neq l}}
\lambda'_l (c_2 \gamma^2 + O(\gamma^4))
\<\phi'_{k}|{  F^u}  |\phi'_{l}\>
\<\phi'_l|{  F^v}|\phi'_k\>  \notag\\
&-
4 \sum_{\substack{1\le k,l\le 2\\ k\neq l}}
\lambda_l (c_2 \gamma^2 + O(\gamma^4))
\<\phi_{k}|{  F^u}  |\phi_{l}\>
\<\phi_l|{  F^v}|\phi_k\> +O(\gamma^4) \notag\\    
\stackrel{(d)}{=}&
  \sum_{\substack{1\le k,l\le 2\\ k\neq l}}
\lambda'_l   \gamma^2  
O(\gamma^2)    
+  \sum_{\substack{1\le k,l\le 2\\ k\neq l}}
O(\gamma^2)     \gamma^2  
\<\phi'_{k}|{  F^u}  |\phi'_{l}\>
\<\phi'_l|{  F^v}|\phi'_k\>  
\notag\\
\stackrel{(e)}{=}&
O(\gamma^4),\notag
\end{align}
where (a) is because the last two summations in \eqref{halfway} are $O(\gamma^4)$,
(b) is because of \eqref{g12' value} and \eqref{gval2},
(c) is because of a telescoping sum,
(d) follows from collecting terms in the telescoping sum, and using \eqref{innerprod perturb} in the first pair of summations and \eqref{lam perturbations} in the second pair of summations and collecting constants into the big-O notation,
and (e) follows from addition in the big-O notation.
Hence, 
the result follows.

\section{Semidefinite program formulation with CVX}
\label{sec:SDPs}

Now, we will formulate the mathematical optimization problems that correspond to $J_2$, $J_3$, $J_4$, $J_5$ as semidefinite programs to be used with the CVX package and provide the corresponding MatLab code. These formulations allow us to numerically evaluate the optimal values and solutions of $J_2$, $J_3$, $J_4$, $J_5$.

For the numerical formulations, we consider probe states of finite dimension $d_A$, and quantum channels that map finite dimensional probe states of dimension $d_A$ to states of finite dimension $d_B$. 
Then, for these optimization programs that correspond to $J_2$, $J_3$, $J_4$, $J_5$, the optimization variable is $Y$, which is a complex semidefinite matrix of size $(d+1)d_B d_A$.
Hence, in our MatLab script, we define the dimension of $Y$ as 
\begin{align}
{\small\texttt{dY = (d+1)*dB*dA;}}\notag
\end{align}
The matrix $Y$ has support on the tensor product space $\mathcal R \otimes \mathcal H_B \otimes H_A$, and we like to interpret $Y$ as a block matrix with blocks $Y_{j,k}$, where $j,k=0,\dots, d$, and $Y_{j,k}$ are matrices on $\mathcal H_B \otimes \mathcal H_A$.

The optimization programs $J_k$ take as input the size $d$ weight matrix $G$,
the channel's Choi matrix $T$, 
and the partial derivatives of the Choi Matrix. 
It is also convenient to specify the input and output dimensions of the quantum channel. 
Hence we can write the first line of the optimization programs 
$J_2, \dots, J_5$ respectively as 
\begin{align}
&{\small\texttt{function [opt\_val, opt\_Y]=J2(G,T,DT,dA,dB)}}\notag\\
&{\small\texttt{function [opt\_val, opt\_Y]=J3(G,T,DT,dA,dB)}}\notag\\
&{\small\texttt{function [opt\_val, opt\_Y]=J4(G,T,DT,dA,dB)}}\notag\\
&{\small\texttt{function [opt\_val, opt\_Y]=J5(G,T,DT,dA,dB)}}.\notag
\end{align}
Here, $\texttt{DT}$ is a cell of derivatives of $\texttt{T}$, and we can access the $j$th derivative of $\texttt{T}$ using $\texttt{DT\{j\}}$. 
The channel's Choi matrix $T$ is a matrix on the space $\mathcal H_B \otimes\mathcal H_A$. 
Each of the functions $\texttt{J2}$, $\texttt{J3}$, $\texttt{J4}$, and $\texttt{J5}$ will find the optimal value and optimal solutions of the optimization problem that corresponds to $J_k$, given by $\texttt{opt\_val}$ and $\texttt{opt\_Y}$ respectively.

We can access the submatrices $Y_{j,k}$ of $Y$ using the following Matlab function:
\begin{align}
&{\small\texttt{function Yjk = jk\_blk(Y,j,k,d)}}\notag\\
&{\small\quad\texttt{dT=size(Y,1)/d;}}\notag\\
&{\small\quad\texttt{idx\_row\_start = j*dT+1;}}\notag\\
&{\small\quad\texttt{idx\_row\_end = j*dT+dT;}}\notag\\
&{\small\quad\texttt{idx\_col\_start = k*dT+1;}}\notag\\
&{\small\quad\texttt{idx\_col\_end = k*dT+dT;}}\notag\\
&{\small\quad\texttt{rows=idx\_row\_start:idx\_row\_end}}\notag\\
&{\small\quad\texttt{cols=idx\_col\_start:idx\_col\_end}}\notag\\
&{\small\quad\texttt{Yjk = Y(rows,cols);}}\notag\\
&{\small\texttt{end}}\notag
\end{align}

The objective function to be minimized is $\tr Y (G\otimes T)$. We can write this objective function in MatLab code by first defining the variables 
\begin{align}
&{\small \texttt{Gprime = [zeros(1,d+1); zeros(d,1),G];}}\notag\\
&{\small \texttt{GT = kron(Gprime, T);}}\notag
\end{align} 
and minimizing $\texttt{trace(Y*GT)}$.

The constraint that $Y$ belongs to the cone $\mathcal B''$,
which corresponds to 
$Y_{k,0}$ being Hermitian for all $k=0,\dots,d$ and $Y_{k,j} = Y_{j,k}^\dagger$ for all $j,k=0,\dots,d$ 
 can be written as 
\begin{align}
&{\small\texttt{for j=0:d}}\notag\\
&{\small\quad\texttt{jk\_blk(Y,j,0,d) == jk\_blk(Y,j,0,d)'}}\notag\\
&{\small\quad\texttt{for k=0:d}}\notag\\
&{\small\quad\quad\texttt{jk\_blk(Y,j,k,d) == jk\_blk(Y,k,j,d)'}}\notag\\
&{\small\quad\texttt{end}}\notag\\
&{\small\texttt{end}}\label{cvx-B''}
\end{align}

Now we will use the variable $\texttt{vecb}$ to represent the matrix $|0\> \otimes |b\> \otimes I_A $.
We use $\texttt{IR=eye(d+1)}$ to represent $I_R$, $\texttt{zket\_R= IR(:,1)}$ to represent $|0\>$,  $\texttt{IA=eye(dA)}$ to represent $I_A$, $\texttt{IB=eye(dB)}$ to represent $I_B$, and $\texttt{IB(:,b)}$ to represent $|b\>$. Then the constraint for (i') can be written as
\begin{align}
&{\small\texttt{for b = 1:dB}}\notag\\
&{\small\quad\texttt{vecb=kron(kron(zket\_R,IB(:,b)), IA);}}\notag\\
&{\small\quad\texttt{for bp = 1:dB}}\notag\\
&{\small\quad\quad\texttt{vecbp=kron(kron(zket\_R, IB(:,bp)), IA );}}\notag\\
&{\small\quad\quad\texttt{if b == bp}}\notag\\
&{\small\quad\quad\quad\texttt{trace(vecb'*Y*vecb) == 1}}\notag\\
&{\small\quad\quad\texttt{else}}\notag\\
&{\small\quad\quad\quad\texttt{vecbp'*Y*vecb == zeros(dA)}}\notag\\
&{\small\quad\quad\quad\texttt{vecb'*Y*vecb == vecbp'*Y*vecbp}}\notag\\
&{\small\quad\quad\texttt{end}}\notag\\
&{\small\quad\texttt{end}}\notag\\
&{\small\texttt{end}}\label{cvx-i'}
\end{align}

The constraint for (ii), which is the channel analogue of the locally unbiased constraint, can be written as 
\begin{align}
&{\small\texttt{for j = 1:d}}\notag\\
&{\small\quad\texttt{for jp = 1:d}}\notag\\
&{\small\quad\quad\texttt{jp\_ket\_R = IR(:,jp+1);}}\notag\\
&{\small\quad\quad\texttt{zjp = zket\_R*jp\_ket\_R';}}\notag\\
&{\small\quad\quad\texttt{if j == jp}}\notag\\
&{\small\quad\quad\quad\texttt{0.5*trace(Y*kron(zjp + zjp', DT\{j\}))==1}}\notag\\
&{\small\quad\quad\texttt{else}}\notag\\
&{\small\quad\quad\quad\texttt{0.5*trace(Y*kron(zjp + zjp', DT\{j\}))==0}}\notag\\
&{\small\quad\quad\texttt{end}}\notag\\
&{\small\quad\texttt{end}}\notag\\
&{\small\texttt{end}}\label{cvx-ii}
\end{align}
            
Now we will proceed to write $J_5$ as a semidefinite program that can be run using the CVX package in MatLab. For $J_5$, the cone $S^5_{AB}$ requires the constraints $Y \ge 0$, $\Tr Y (|j\>\<i|\otimes T ) = 0$ and that $Y \in \mathcal B''$. Hence we can write $J_5$ as
\begin{align}
&{\small\texttt{cvx\_begin sdp \%J5}}\notag\\
&{\small\quad\texttt{variable Y(dY,dY) complex semidefinite}}\notag\\
&{\small\quad\texttt{minimize(trace(Y*GT))}}\notag\\
&{\small\quad\texttt{subject to }}\notag\\
&\mbox{Insert code for \eqref{cvx-i'}}\notag\\
&\mbox{Insert code for \eqref{cvx-ii}}\notag\\
&\mbox{Insert code for \eqref{cvx-B''}}\notag\\
&{\small\quad\texttt{for j = 1:d}}\notag\\
&{\small\quad\quad\texttt{j\_ket\_R = IR(:,j+1);}}\notag\\
&{\small\quad\quad\texttt{for jp = 1:d}}\notag\\
&{\small\quad\quad\quad\texttt{jp\_ket\_R = IR(:,jp+1);}}\notag\\           
&{\small\quad\quad\quad\texttt{j\_jp = j\_ket\_R*jp\_ket\_R';}}\notag\\
&{\small\quad\quad\quad\texttt{imag(trace(Y*kron(j\_jp,T)))==0}}\notag\\
&{\small\quad\quad\texttt{end}}\notag\\
&{\small\quad\texttt{end}}\notag\\
&{\small\texttt{cvx\_end}}
\end{align}
For $J_4$, we can write the optimization as 
\begin{align}
&{\small\texttt{cvx\_begin sdp \%J4}}\notag\\
&{\small\quad\texttt{variable Y(dY,dY) complex semidefinite}}\notag\\
&{\small\quad\texttt{minimize(trace(Y*GT))}}\notag\\
&{\small\quad\texttt{subject to }}\notag\\
&\mbox{Insert code for \eqref{cvx-i'}}\notag\\
&\mbox{Insert code for \eqref{cvx-ii}}\notag\\
&\mbox{Insert code for \eqref{cvx-B''}}\notag\\
&{\small\texttt{cvx\_end}}
\end{align}
For $J_3$, we will need a positive partial transpose constraint on $Y$.
We can write the partial transpose of $Y$ as $\texttt{Y\_PT}$ where
\begin{align}
&{\small\texttt{Y\_PT = zeros((d+1)*dA*dB)}}\notag\\
&{\small\texttt{for j = 0:d}}\notag\\
&{\small\quad\texttt{idx\_row\_start = j*dT+1;}}\notag\\
&{\small\quad\texttt{idx\_row\_end = j*dT+dT;}}\notag\\
&{\small\quad\texttt{rows=idx\_row\_start:idx\_row\_end}}\notag\\
&{\small\quad\texttt{for k = 0:d}}\notag\\
&{\small\quad\quad\texttt{idx\_col\_start = k*dT+1;}}\notag\\
&{\small\quad\quad\texttt{idx\_col\_end = k*dT+dT;}}\notag\\
&{\small\quad\quad\texttt{cols=idx\_col\_start:idx\_col\_end}}\notag\\
&{\small\quad\quad\texttt{Y\_PT(rows,cols)=jk\_blk(Y,j,k,d);}}\notag\\           
&{\small\quad\texttt{end}}\notag\\
&{\small\texttt{end}}\label{cvx-Y-PT}
\end{align}
 and hence we can write $J_3$ as 
\begin{align}
&{\small\texttt{cvx\_begin sdp \%J3}}\notag\\
&{\small\quad\texttt{variable Y(dY,dY) complex semidefinite}}\notag\\
&{\small\quad\texttt{minimize(trace(Y*GT))}}\notag\\
&{\small\quad\texttt{subject to }}\notag\\
&\mbox{Insert code for \eqref{cvx-i'}}\notag\\
&\mbox{Insert code for \eqref{cvx-ii}}\notag\\
&\mbox{Insert code for \eqref{cvx-B''}}\notag\\
&\mbox{Insert code for \eqref{cvx-Y-PT}}\notag\\
&{\small\quad\texttt{Y\_PT>=0}}\notag\\
&{\small\texttt{cvx\_end}}\notag
\end{align}
For $J_2$, we need to optimize over the cone $\mathcal B$ instead of $\mathcal B''$. 
The constraints for the cone $\mathcal B$ can be written as
\begin{align}
&{\small\texttt{for j=0:d}}\notag\\
&{\small\quad\texttt{for k=0:d}}\notag\\
&{\small\quad\quad\texttt{jk\_blk(Y,j,k,d) == jk\_blk(Y,k,j,d)}}\notag\\
&{\small\quad\quad\texttt{jk\_blk(Y,j,k,d) == jk\_blk(Y,j,k,d)'}}\notag\\
&{\small\quad\texttt{end}}\notag\\
&{\small\texttt{end}}\label{cvx-B}
\end{align}
Hence we can write $J_2$ as
\begin{align}
&{\small\texttt{cvx\_begin sdp \%J2}}\notag\\
&{\small\quad\texttt{variable Y(dY,dY) complex semidefinite}}\notag\\
&{\small\quad\texttt{minimize(trace(Y*GT))}}\notag\\
&{\small\quad\texttt{subject to }}\notag\\
&\mbox{Insert code for \eqref{cvx-i'}}\notag\\
&\mbox{Insert code for \eqref{cvx-ii}}\notag\\
&\mbox{Insert code for \eqref{cvx-B}}\notag\\ 
&{\small\texttt{cvx\_end}}\notag
\end{align}

For the optimization problems that correspond to $J_2, \dots, J_5$, once we find the optimal $Y^*$, the corresponding optimal probe state $\rho^*_A$ can be written as
\begin{align}
\rho^*_A := \Tr_{RB}[ Y( |0\>\<0| \otimes I_{AB} ) ]/ d_B  .
\end{align}

\end{document}